\documentclass[11pt,letterpaper]{article}
\usepackage[left=1in,right=1in,top=1in,bottom=1in]{geometry}
\usepackage{xspace}
\usepackage{latexsym}
\usepackage{stmaryrd} 
\usepackage{amsmath}
\usepackage{amssymb}
\usepackage{amsthm}
\usepackage{comment}
\usepackage{xy}
\usepackage{hyperref}
\usepackage{url}

\excludecomment{comment:FOCS}
\includecomment{comment:inFOCS}

\theoremstyle{definition} 
\newtheorem{definition}{Definition}[section]
\newtheorem{open}[definition]{Open Question}
\newtheorem{question}[definition]{General Question}
\newtheorem{example}[definition]{Example}
\newtheorem{remark}[definition]{Remark}
\newtheorem{observation}[definition]{Observation}

\newtheorem*{definition*}{Definition}

\theoremstyle{plain}
\newtheorem{proposition}[definition]{Proposition}
\newtheorem{lemma}[definition]{Lemma}
\newtheorem{corollary}[definition]{Corollary}
\newtheorem{theorem}[definition]{Theorem}

\newtheorem*{VNPthm}{Theorem~\ref{thm:VNP}}
\newtheorem*{VNPlem}{Lemma~\ref{lem:VNP}}
\newtheorem*{EFthm}{Theorem~\ref{thm:EF}}
\newtheorem*{AC0thm}{Theorem~\ref{thm:AC0}}

\newtheorem*{AC0pcor}{Corollary~\ref{cor:AC0p}}
\newtheorem*{pitassiProp}{Proposition~\ref{prop:pitassi}}
\newtheorem*{depthThm}{Theorem~\ref{thm:depth}}

\newcommand{\rmkqed}{\nobreak\hfill$\lhd$}
\newenvironment{proof-idea}{\noindent \textit{Proof idea.}}{\nobreak\hfill\qed\bigskip}

\input xy
\xyoption{all}

\newcommand{\myparagraph}[1]{
\medskip
\noindent \textbf{#1}
\medskip
}

\newcommand{\definedWord}[1]{\emph{#1}}
\newcommand{\ie}{i.\,e.\xspace}
\newcommand{\eg}{e.\,g.\xspace}

\newcommand{\Z}{\ensuremath{\mathbb{Z}}}

\newcommand{\C}{\ensuremath{\mathbb{C}}}

\newcommand{\F}{\ensuremath{\mathbb{F}}}

\DeclareMathOperator{\poly}{poly}

\newcommand{\image}{\mathrm{Im}}

\newcommand{\defeq}{\stackrel{def}{=}}

\newcommand{\parity}{\oplus}
\newcommand{\cc}[1]{\ensuremath{\mathsf{#1}}} 

\newcommand{\Burgisser}{B\"{u}rgisser\xspace}
\newcommand{\Grobner}{Gr\"{o}bner\xspace}
\newcommand{\Krajicek}{Kraj{\'{\i}}{\v{c}}ek\xspace}

\newcommand{\I}{\ensuremath{\text{IPS}}\xspace}
\newcommand{\Itext}{IPS\xspace}
\newcommand{\f}{y}
\newcommand{\prop}[1]{\underline{#1}}

\begin{document}

\title{Circuit Complexity, Proof Complexity and Polynomial Identity Testing
}
\author{Joshua A. Grochow and Toniann Pitassi}

\maketitle

\pagestyle{myheadings}
\markboth{Circuit Complexity, Proof Complexity, and PIT- J. A. Grochow and T. Pitassi}{Circuit Complexity, Proof Complexity, and PIT - J. A. Grochow and T. Pitassi}

\begin{abstract}
We introduce a new and very natural algebraic proof system, which has tight connections to (algebraic) circuit complexity. In particular, we show that any super-polynomial lower bound on any Boolean tautology in our proof system implies that the permanent does not have polynomial-size algebraic circuits ($\cc{VNP} \neq \cc{VP}$). 
As a corollary to the proof, we also show that super-polynomial lower bounds on the number of lines in Polynomial Calculus proofs (as opposed to the usual measure of number of monomials) imply the Permanent versus Determinant Conjecture.
Note that, prior to our work, there was no proof system for which lower bounds on an arbitrary tautology implied \emph{any} computational lower bound. 

Our proof system helps clarify the relationships between previous algebraic proof systems, and begins to shed light on why proof complexity lower bounds for various proof systems have been so much harder than lower bounds on the corresponding circuit classes. In doing so, we highlight the importance of polynomial identity testing (PIT) for understanding proof complexity.

More specifically, we introduce certain propositional axioms satisfied by any Boolean circuit computing PIT. (The existence of efficient proofs for our PIT axioms appears to be somewhere in between the major conjecture that PIT$\in \cc{P}$ and the known result that PIT$\in \cc{P/poly}$.) We use these PIT axioms to shed light on $\cc{AC}^0[p]$-Frege lower bounds, which have been open for nearly 30 years, with no satisfactory explanation as to their apparent difficulty. We show that either:
\begin{enumerate}
\renewcommand{\theenumi}{\alph{enumi}}
\item Proving super-polynomial lower bounds on $\cc{AC}^0[p]$-Frege implies $\cc{VNP}_{\F_p}$ does not have polynomial-size circuits of depth $d$---a notoriously open question for any $d \geq 4$---thus explaining the difficulty of lower bounds on $\cc{AC}^0[p]$-Frege, or

\item $\cc{AC}^0[p]$-Frege cannot efficiently prove the depth $d$ PIT axioms, and hence we have a lower bound on $\cc{AC}^0[p]$-Frege.
\end{enumerate}
We also prove many variants on this statement for other proof systems and other computational lower bounds.

Finally, using the algebraic structure of our proof system, we propose a novel way to extend techniques from algebraic circuit complexity to prove lower bounds in proof complexity. Although we have not yet succeeded in proving such lower bounds, this proposal should be contrasted with the difficulty of extending $\cc{AC}^0[p]$ circuit lower bounds to $\cc{AC}^0[p]$-Frege lower bounds.
\end{abstract}

\section{Extended abstract} \label{sec:eabs}

\subsection{Introduction}
$\cc{NP}$ versus $\cc{coNP}$ is the very natural question of whether, for every graph that doesn't have a Hamiltonian path, there is a short proof of this fact. One of the arguments for the utility of proof complexity is that by proving lower bounds against stronger and stronger proof systems, we ``make progress'' towards proving $\cc{NP} \neq \cc{coNP}$. However, until now this argument has been more the expression of a philosophy or hope, as there is no known proof system for which lower bounds imply computational complexity lower bounds of any kind, let alone $\cc{NP} \neq \cc{coNP}$.

We remedy this situation by introducing a very natural algebraic proof system, which has tight connections to (algebraic) circuit complexity. We show that any super-polynomial lower bound on any Boolean tautology in our proof system implies that the permanent does not have polynomial-size algebraic circuits ($\cc{VNP} \neq \cc{VP}$). Note that, prior to our work, essentially all implications went the opposite direction: a circuit complexity lower bound implying a proof complexity lower bound. We use this result to begin to explain why several long-open lower bound questions in proof complexity---lower bounds on Extended Frege, on $\cc{AC}^0[p]$-Frege, and on number-of-lines in Polynomial Calculus-style proofs---have been so apparently difficult.

\subsubsection{Background and Motivation}
\paragraph{Algebraic Circuit Complexity.} The most natural way to compute a polynomial function $f(x_1,\dotsc,x_n)$ is
with a sequence of instructions $g_1,\dotsc,g_m = f$, starting from the inputs $x_1, \dotsc, x_n$, and where each instruction $g_i$ is of the form $g_j \circ g_k$ for some $j,k < i$, where $\circ$ is either a linear combination or multiplication. Such computations are called algebraic circuits or straight-line programs. The goal of algebraic complexity is to understand the optimal asymptotic complexity of computing a given polynomial family $(f_n(x_1,\dotsc,x_{\poly(n)})_{n=1}^{\infty}$, typically in terms of size and depth. In addition to the intrinsic interest in these questions, since Valiant's work \cite{valiant, valiantPerm, valiantProjections} algebraic complexity has become more and more important for Boolean computational complexity. Valiant argued that understanding algebraic complexity could give new intuitions that may lead to better understanding of other models of computation (see also \cite{Gat2}); several direct connections have been found between algebraic and Boolean complexity \cite{kabanetsImpagliazzo, burgisserCookValiant, jansenSanthanam, mulmuleyPRAM}; and the Geometric Complexity Theory Program (see, \eg, the survey \cite{gctCACM} and references therein) suggests how algebraic techniques might be used to resolve major Boolean complexity conjectures.

Two central functions in this area are the determinant and permanent polynomials, 
which are fundamental both because of their prominent role in many areas of mathematics and because they are complete for various natural complexity classes. In particular, the permanent of $\{0,1\}$-matrices is $\cc{\# P}$-complete, and the permanent of arbitrary matrices is $\cc{VNP}$-complete. Valiant's Permanent versus Determinant Conjecture \cite{valiant} states that the permanent of an $n \times n$ matrix, as a polynomial in $n^2$ variables, cannot be written as the determinant of any polynomially larger matrix all of whose entries are variables or constants. In some ways this is an algebraic analog of $\cc{P} \neq \cc{NP}$, although it is in fact much closer to $\cc{FNC}^2 \neq \cc{\# P}$. In addition to this analogy, the Permanent versus Determinant Conjecture is also known to be a formal consequence of the nonuniform lower bound $\cc{NP} \not\subseteq \cc{P/poly}$ \cite{burgisserCookValiant}, and is thus thought to be an important step towards showing $\cc{P} \neq \cc{NP}$.

Unlike in Boolean circuit complexity, (slightly) non-trivial lower bounds for the size of algebraic circuits are known \cite{strassenDegree,baurStrassen}. Their methods, however, only give lower bounds up to $\Omega (n\log n)$. Moreover, their methods are based on  a degree analysis of certain algebraic varieties and do not give lower bounds for polynomials of constant degree. Recent exciting work \cite{agrawalVinay, koiranChasm, tavenas} has shown that polynomial-size algebraic circuits computing functions of polynomial degree can in fact be computed by subexponential-size depth 4 algebraic circuits. Thus, strong enough lower bounds for depth 4 algebraic circuits for the permanent would already prove $\cc{VP} \neq \cc{VNP}$.

\medskip

\paragraph{Proof Complexity.} Despite considerable progress obtaining super-polynomial lower bounds for many weak proof systems (resolution, cutting planes, bounded-depth Frege systems), there has been essentially no progress in the last 25 years for stronger proof systems such as Extended Frege systems or Frege systems. More surprisingly, no nontrivial lower bounds are known for the seemingly weak $\cc{AC}^0[p]$-Frege system. Note that in contrast, the analogous result in circuit complexity---proving  super-polynomial $\cc{AC}^0[p]$ lower bounds for an explicit function---was resolved by Smolensky over 25 years ago \cite{smolensky}. To date, there has been no satisfactory explanation for this state of affairs.

In proof complexity, there are no known formal barriers such as relativization \cite{bakerGillSolovay}, Razborov--Rudich natural proofs \cite{razborovRudich}, or algebrization \cite{aaronsonWigderson} that exist in Boolean function complexity. Moreover, there has not even been progress by way of conditional lower bounds. That is, trivially $\cc{NP} \neq \cc{coNP}$ implies superpolynomial lower bounds for $\cc{AC}^0[p]$-Frege, but we know of no weaker complexity assumption that implies such lower bounds. The only formal implication in this direction shows that certain circuit lower bounds imply lower bounds for proof systems that admit feasible interpolation, but unfortunately only weak proof systems (not Frege nor even $\cc{AC}^0$-Frege) have this property \cite{Bonet,Bonet2}. In the converse direction, there are essentially no implications at all. For example, we do not know if $\cc{AC}^0[p]$-Frege lower bounds---nor even Frege nor Extended Frege lower bounds---imply any nontrivial circuit lower bounds.

\subsubsection{Our Results}
In this paper, we define a simple and natural proof system that we call the Ideal Proof System (IPS)
based on Hilbert's Nullstellensatz. Our system is similar in spirit to related
algebraic proof systems that have been previously studied, but is different in a crucial way that we explain below.

Given a set of polynomials $F_1,\ldots,F_m$ in $n$ variables $x_1,\ldots,x_n$ over a field $\F$ without a
common zero over the algebraic closure of $\F$, Hilbert's Nullstellensatz says that there exist polynomials
$G_1,\ldots,G_m \in \F[x_1,\ldots,x_n]$ such that $\sum F_i G_i =1$, \ie, that $1$ is in the ideal generated by the $F_i$. In the Ideal Proof System, we introduce new variables $\f_i$ which serve as placeholders into which the original polynomials $F_i$ will
eventually be substituted:

\begin{definition}[Ideal Proof System] \label{def:IPS}
An \definedWord{\I certificate} that a system of $\F$-polynomial equations
$F_1(\vec{x})=F_2(\vec{x}) = \dotsb = F_m(\vec{x}) = 0$ is unsatisfiable over $\overline{\F}$ is
a polynomial $C(\vec{x}, \vec{\f})$ in the variables $x_1,\ldots,x_n$ and $\f_1,\ldots,\f_m$ such that
\begin{enumerate}
\item \label{condition:ideal} $C(x_1,\dotsc,x_n,\vec{0}) = 0$, and
\item \label{condition:nss} $C(x_1,\dotsc,x_n,F_1(\vec{x}),\dotsc,F_m(\vec{x})) = 1$.
\end{enumerate}
The first condition is equivalent to $C$ being in the ideal generated by $\f_1, \dotsc, \f_m$, and the two conditions together therefore imply that $1$ is in the ideal generated by the $F_i$, and hence that $F_1(\vec{x}) = \dotsb = F_m(\vec{x})=0$ is unsatisfiable. 

An \definedWord{\I proof} of the unsatisfiability of the polynomials $F_i$ is an $\F$-algebraic circuit on inputs $x_1,\ldots,x_n,\f_1,\ldots,\f_m$ computing some \I certificate of unsatisfiability.
\end{definition}

For any class $\mathcal{C}$ of polynomial families, we may speak of $\mathcal{C}$-\I proofs of a family of systems of equations $(\mathcal{F}_n)$ where $\mathcal{F}_n$ is $F_{n,1}(\vec{x}) = \dotsb = F_{n,\poly(n)}(\vec{x}) = 0$. When we refer to \I without further qualification, we mean $\cc{VP}$-\I, that is, the family of \I proofs should be computed by circuits of polynomial size \emph{and polynomial degree}, unless specified otherwise.

The Ideal Proof System (without any size bounds) is easily shown to be sound, and its completeness follows from the Nullstellensatz.

We typically consider \I as a propositional proof system by translating a CNF tautology $\varphi$ into a system of equations as follows. We translate a clause $\kappa$ of $\varphi$ into a single algebraic equation $F(\vec{x})$ as follows: $x \mapsto 1-x$, $x \vee y \mapsto xy$. This translation has the property that a $\{0,1\}$ assignment satisfies $\kappa$ if and only if it satisfies the equation $F = 0$. Let $\kappa_1, \dotsc, \kappa_m$ denote all the clauses of $\varphi$, and let $F_i$ be the corresponding polynomials. Then the system of equations we consider is $F_1(\vec{x}) = \dotsb = F_m(\vec{x}) = x_1^2 - x_1 = \dotsb = x_n^2 - x_n = 0$. The latter equations force any solution to this system of equations to be $\{0,1\}$-valued. Despite our indexing here, when we speak of the system of equations corresponding to a tautology, we always assume that the $x_i^2 - x_i$ are among the equations.

Like previously defined algebraic systems \cite{BIKPP,CEI,pitassi96,pitassiICM}, proofs in our system can be
checked in randomized polynomial time.
The key difference between our system and previously studied
ones is that those systems are axiomatic in the sense that they require that \emph{every}
sub-computation (derived polynomial) be in the ideal generated by the original polynomial equations $F_i$, and thus be a sound consequence of the equations $F_1=\dotsb=F_m=0$.
In contrast our system has no such requirement; an \I proof can compute potentially
unsound sub-computations (whose vanishing does not follow from $F_1=\dotsb=F_m=0$), as long as the \emph{final polynomial} is in the ideal
generated by the equations. This key difference allows \I proofs to be
\emph{ordinary algebraic circuits}, and thus nearly all results in
algebraic circuit complexity apply directly to the Ideal Proof System. To quote the tagline of a common US food chain, the Ideal Proof System is a ``No rules, just right'' proof system.

Our first main theorem shows one of the advantages of this close connection with algebraic circuits. To the best of our knowledge, this is the first implication showing that a proof complexity lower bound implies any sort of computational complexity lower bound. 

\begin{VNPthm}
Super-polynomial lower bounds for the Ideal Proof System imply that the permanent does not have polynomial-size
algebraic circuits, that is, $\cc{VNP} \neq \cc{VP}$.
\end{VNPthm}

From the proof of this result, together with one of our simulation results (Proposition~\ref{prop:pitassi}), we also get:

\begin{corollary} \label{cor:PC}
Super-polynomial lower bounds on the number of lines in Polynomial Calculus proofs imply the Permanent versus Determinant Conjecture.\footnote{Although Corollary~\ref{cor:PC} may seem to be saying that lower bounds on PC imply a circuit lower bound, this is not precisely the case, because complexity in PC is emphatically not measured by the number of lines, but rather by the total number of monomials appearing in a PC proof. This is true both definitionally and in practice, in that all previous papers on PC use the number-of-monomials complexity measure.}
\end{corollary}

Under a reasonable assumption on polynomial identity testing (PIT), which we discuss further below, we are able to show that Extended Frege is equivalent to the Ideal Proof System. Extended Frege (EF) is the strongest natural deduction-style propositional proof system that has been proposed, and is the proof complexity analog of $\cc{P/poly}$ (that is, Extended Frege = $\cc{P/poly}$-Frege). 

\begin{EFthm}
Let $K$ be a family of polynomial-size Boolean circuits for PIT such that the PIT axioms for $K$ (see Definition~\ref{def:PITaxioms}) have polynomial-size EF proofs. Then EF polynomially simulates \I, and hence the EF and \I are polynomially equivalent.
\end{EFthm}

Under this assumption about PIT, Theorems~\ref{thm:VNP} and \ref{thm:EF} in combination suggest a precise reason that proving lower bounds on Extended Frege is so difficult, namely, that doing so implies $\cc{VP} \neq \cc{VNP}$. Theorem~\ref{thm:EF} also suggests that to make progress toward proving lower bounds in proof complexity, it may be necessary to prove lower bounds for the 
Ideal Proof System, which we feel is more natural, and creates the possibility of harnessing tools from algebra, representation theory, and algebraic circuit complexity. We give a specific suggestion of how to apply these tools towards proof complexity lower bounds in Section~\ref{sec:syzygy}.

\begin{remark} \label{rmk:PIT}
Given that $PIT \in \cc{P}$ is known to imply lower bounds, one may wonder if the combination of the above two theorems really gives any explanation at all for the difficulty of proving lower bounds on Extended Frege. There are at least two reasons that it does. 

First, the best lower bound known to follow from $PIT \in \cc{P}$ is an algebraic circuit-size lower bound on an integer polynomial that can be evaluated in $\cc{NEXP} \cap \cc{coNEXP}$ \cite{jansenSanthanam} (via personal communication we have learned that Impagliazzo and Williams have also proved similar results), whereas our conclusion is a lower bound on algebraic circuit-size for an integer polynomial computable in $\cc{\# P} \subseteq \cc{PSPACE}$. 

Second, the hypothesis that our PIT axioms can be proven efficiently in Extended Frege seems to be somewhat orthogonal to, and may be no stronger than, the widely-believed hypothesis that PIT is in $\cc{P}$. As Extended Frege is a nonuniform proof system, efficient Extended Frege proofs of our PIT axioms are unlikely to have any implications about the uniform complexity of PIT (and given that we already know unconditionally that PIT is in $\cc{P/poly}$, uniformity is what the entire question of derandomizing PIT is about). In the opposite direction, it's a well-known observation in proof complexity that nearly all natural uniform polynomial-time algorithms have feasible (Extended Frege) correctness proofs. If this phenomenon doesn't apply to PIT, it would be interesting for both proof complexity and circuit complexity, as it indicates the difficulty of proving that PIT is in $\cc{P}$. \rmkqed
\end{remark}

Although PIT has long been a central problem of study in computational complexity---both because of its importance in many algorithms, as well as its strong connection to circuit lower bounds---our theorems highlight the importance of PIT in proof complexity. Next we prove that Theorem~\ref{thm:EF} can be scaled down to obtain similar results for weaker Frege systems, and discuss some of its more striking consequences.

\begin{AC0thm}
Let $\mathcal{C}$ be any of the standard circuit classes $\cc{AC}^k, \cc{AC}^k[p], \cc{ACC}^k, \cc{TC}^k, \cc{NC}^k$. Let $K$ be a family of polynomial-size Boolean circuits for PIT (not necessarily in $\mathcal{C}$) such that the PIT axioms for $K$ have polynomial-size $\mathcal{C}$-Frege proofs. Then $\mathcal{C}$-Frege is polynomially equivalent to \I, and consequently to Extended Frege as well.
\end{AC0thm}

Theorem~\ref{thm:AC0} also highlights the importance of our PIT axioms for getting $\cc{AC}^0[p]$-Frege lower bounds, which has been an open question for nearly thirty years. (For even weaker systems, Theorem~\ref{thm:AC0} in combination with known results yields an unconditional lower bound on $\cc{AC}^0$-Frege proofs of the PIT axioms.) In particular, 
we are in the following win-win scenario: 

\begin{AC0pcor}
For any $d$, either:
\begin{itemize}
\item There are polynomial-size $\cc{AC}^0[p]$-Frege proofs of the depth $d$ PIT axioms, in which case \emph{any superpolynomial lower bounds on $\cc{AC}^0[p]$-Frege imply $\cc{VNP}_{\F_p}$ does not have polynomial-size depth $d$ algebraic circuits}, thus explaining the difficulty of obtaining such lower bounds, or

\item There are no polynomial-size $\cc{AC}^0[p]$-Frege proofs of the depth $d$ PIT axioms, in which case we've gotten $\cc{AC}^0[p]$-Frege lower bounds.
\end{itemize}
\end{AC0pcor}

Finally, in Section~\ref{sec:syzygy} we suggest a new framework for proving lower bounds for
the Ideal Proof System which we feel has promise. Along the way, we make precise the difference in difficulty between proof complexity lower bounds (on \I, which may also apply to Extended Frege via Theorem~\ref{thm:EF}) and algebraic circuit lower bounds. In particular, the set of \emph{all $\I$-certificates} for a given unsatisfiable system of equations is, in a certain precise sense, ``finitely generated.'' We suggest how one might take advantage of this finite generation to transfer techniques from algebraic circuit complexity to prove lower bounds on $\I$, and consequently on Extended Frege (since $\I$ p-simulates Extended Frege unconditionally), giving hope for the long-sought length-of-proof lower bounds on an algebraic proof system. We hope to pursue this approach in future work.

\subsubsection{Related Work}
We will see in Section~\ref{sec:others} that many previously studied proof systems can be p-simulated by \I, and furthermore can be viewed simply as different complexity measures on $\I$ proofs, or as $\mathcal{C}$-\I for certain classes $\mathcal{C}$. In particular, the Nullstellensatz system \cite{BIKPP}, the Polynomial Calculus (or \Grobner) proof system \cite{CEI}, and Polynomial Calculus with Resolution \cite{PCR} are all particular measures on \I, and Pitassi's previous algebraic systems \cite{pitassi96, pitassiICM} are subsystems of \I. 

Raz and Tzameret \cite{razTzameret} introduced various multilinear algebraic proof systems. Although their systems are not so easily defined in terms of \I, the Ideal Proof System nonetheless p-simulates all of their systems. Amongst other results, they show that a super-polynomial separation between two variants of their system---one representing lines by multilinear circuits, and one representing lines by general algebraic circuits---would imply a super-polynomial separation between general and multilinear circuits computing multilinear polynomials. However, they only get implications to lower bounds on multilinear circuits rather than general circuits, and they do not prove a statement analogous to our Theorem~\ref{thm:VNP}, that lower bounds on a single system imply algebraic circuit lower bounds.

\subsubsection{Outline}
The remainder of Section~\ref{sec:eabs} gives proofs of some foundational results, and summarizes the rest of the paper, giving detailed versions of all statements and discussing their proofs and significance. In Section~\ref{sec:eabs} many proofs are only sketched or are delayed until later in the paper, but all proofs of all results are present either in Section~\ref{sec:eabs} or in Sections~\ref{sec:simulations}--\ref{sec:PIT}.

We start in Section~\ref{sec:general}, by proving several basic facts about \I (some proofs are deferred to Section~\ref{sec:simulations}). We discuss the relationship between \I and previously studied proof systems.
We also highlight several consequences of results from algebraic complexity theory for the Ideal Proof System, such as division elimination \cite{strassenDivision} and the chasms at depth 3 \cite{GKKSchasm,tavenas} and 4 \cite{agrawalVinay,koiranChasm,tavenas}. 

In Section~\ref{sec:VNPeabs}, we outline the proof that lower bounds on \I imply algebraic circuit lower bounds (Theorem~\ref{thm:VNP}; full proof in Section~\ref{sec:VNP}). We also show how this result gives as a corollary a new, simpler proof that $\cc{NP} \not\subseteq \cc{coMA} \Rightarrow \cc{VNP}^0 \neq \cc{VP}^0$. In Section~\ref{sec:PITeabs} we introduce our PIT axioms in detail and outline the proof of Theorems~\ref{thm:EF} and \ref{thm:AC0} (full proofs in Section~\ref{sec:EF}). We also discuss in detail many variants of Theorem~\ref{thm:AC0} and their consequences, as briefly mentioned above. In Section~\ref{sec:syzygy} we suggest a new framework for transferring techniques from algebraic circuit complexity to (algebraic) proof complexity lower bounds. Finally, in Section~\ref{sec:conclusion} we gather a long list of open questions raised by our work, many of which we believe may be quite approachable.

Appendix~\ref{app:background} contains more complete preliminaries. In Appendices~\ref{app:RIPS} and \ref{app:geom} we introduce two variants of the Ideal Proof System---one of which allows certificates to be rational functions and not only polynomials, and one of which has a more geometric flavor---and discuss their relationship to \I. These systems further suggest that tools from geometry and algebra could potentially be useful for understanding the complexity of various propositional tautologies and more generally the complexity of individual instances of $\cc{NP}$-complete problems.

\subsection{A few preliminaries}
In this section we cover the bare bones preliminaries that we think may be less familiar to some of our readers. Remaining background material on algebraic complexity, proof complexity, and commutative algebra can be found in Appendix~\ref{app:background}. As general references, we refer the reader to \Burgisser--Clausen--Shokrollahi \cite{BCS} and the surveys \cite{shpilkaYehudayoff,chenKayalWigderson} for algebraic complexity, to \Krajicek \cite{krajicekBook} for proof complexity, and to any of the standard books \cite{eisenbud,atiyahMacdonald,matsumura,reidCA} for commutative algebra.

\subsubsection{Algebraic Complexity} \label{sec:prelim:comp}
Over a ring $R$, $\cc{VP}_{R}$ is the class of families $f=(f_n)_{n=1}^{\infty}$ of formal polynomials---that is, considered as symbolic polynomials, rather than as functions---$f_n$ such that $f_n$ has $\poly(n)$ input variables, is of $\poly(n)$ degree, and can be computed by algebraic circuits over $R$ of $\poly(n)$ size. $\cc{VNP}_{R}$ is the class of families $g$ of polynomials $g_n$ such that $g_n$ has $\poly(n)$ input variables and is of $\poly(n)$ degree, and can be written as
\[
g_n(x_1,\dotsc,x_{\poly(n)}) = \sum_{\vec{e} \in \{0,1\}^{\poly(n)}} f_n(\vec{e}, \vec{x})
\]
for some family $(f_n) \in \cc{VP}_{R}$.

A family of algebraic circuits is said to be \definedWord{constant-free} if the only constants used in the circuit are $\{0,1,-1\}$. Other constants can be used, but must be built up using algebraic operations, which then count towards the size of the circuit. We note that over a fixed finite field $\F_q$, $\cc{VP}^0_{\F_q} = \cc{VP}_{\F_q}$, since there are only finitely many possible constants. Consequently, $\cc{VNP}^0_{\F_q} = \cc{VNP}_{\F_q}$ as well. Over the integers, $\cc{VP}^0_{\Z}$ coincides with those families in $\cc{VP}_{\Z}$ that are computable by algebraic circuits of polynomial total \emph{bit-size}: note that any integer of polynomial bit-size can be constructed by a constant-free circuit by using its binary expansion $b_n \dotsb b_1 = \sum_{i=0}^{n-1} b_i 2^i$, and computing the powers of $2$ by linearly many successive multiplications. A similar trick shows that over the algebraic closure $\overline{\F}_p$ of a finite field, $\cc{VP}^0_{\overline{\F}_p}$ coincides with those families in $\cc{VP}_{\overline{\F}_p}$ that are computable by algebraic circuits of polynomial total bit-size, or equivalently where the constants they use lie in subfields of $\overline{\F}_p$ of total size bounded by $2^{n^{O(1)}}$. (Recall that $\F_{p^{a}}$ is a subfield of $\F_{p^b}$ whenever $a | b$, and that the algebraic closure $\overline{\F}_p$ is just the union of $\F_{p^{a}}$ over all integers $a$.)

\subsubsection{Proof Complexity} \label{sec:prelim:proof}
In brief, a \definedWord{proof system} for a language $L \in \cc{coNP}$ is a nondeterministic algorithm for $L$, or equivalently a deterministic polynomial-time verifier $P$ such that $x \in L \Leftrightarrow (\exists y)[P(x,y)=1]$, and we refer to any such $y$ as a $P$-proof that $x \in L$.\footnote{This notion is essentially due to Cook and Reckhow \cite{cookReckhow}; although their definition was formalized slightly differently, it is essentially equivalent to the one we give here.} We say that $P$ is \definedWord{polynomially bounded} if for every $x \in L$ there is a $P$-proof of length polynomially bounded in $|x|$: $|y| \leq \poly(|x|)$. We will generally be considering proof systems for the $\cc{coNP}$-complete language TAUT consisting of all propositional tautologies; there is a polynomially bounded proof system for TAUT if and only if $\cc{NP} = \cc{coNP}$.

Given two proof systems $P_1$ and $P_2$ for the same language $L \in \cc{coNP}$, we say that $P_1$ \definedWord{polynomially simulates} or \definedWord{p-simulates} $P_2$ if there is a polynomial-time function $f$ that transforms $P_1$-proofs into $P_2$-proofs, that is, $P_1(x,y)=1 \Leftrightarrow P_2(x,f(y))=1$. We say that $P_1$ and $P_2$ are \definedWord{polynomially equivalent} or \definedWord{p-equivalent} if each p-simulates the other. (This is the proof complexity version of Levin reductions between $\cc{NP}$ problems.)

For TAUT (or UNSAT), there are a variety of standard and well-studied proof systems. In this paper we will be primarily concerned with Frege---a standard, school-style line-by-line deductive system---and its variants such as Extended Frege (EF) and $\cc{AC}^0$-Frege. 
Bounded-depth Frege or $\cc{AC}^0$-Frege are Frege proofs but with the additional restriction that each formula appearing in the
proof has bounded depth \emph{syntactically} (the \emph{syntactic} nature of this condition is crucial: since every formula appearing in a proof is a tautology, semantically all such formulas are the constant-true function and can be computed by trivial circuits). As with $\cc{AC}^0$ circuits, $\cc{AC}^0$-Frege has rules for handling unbounded fan-in AND and OR connectives, in addition to negations. 

For almost any syntactically-defined class of circuits $\mathcal{C}$, one can similarly define $\mathcal{C}$-Frege. For example, $\cc{NC}^1$-Frege is p-equivalent to Frege. However, despite the seeming similarities, there are some differences between a circuit class and its corresponding Frege system. Exponential lower bounds are known for $\cc{AC}^0$-Frege \cite{BIKPPW}, which use the Switching Lemma as for lower bounds on $\cc{AC}^0$ circuits, but in a more complicated way. However, unlike the case of $\cc{AC}^0[p]$ circuits for which we have exponential lower bounds \cite{razborov, smolensky}, essentially no nontrivial lower bounds are known for $\cc{AC}^0[p]$-Frege.

Extended Frege systems generalize Frege systems by allowing, in addition to all of the Frege rules, a new axiom schema of the form $y \leftrightarrow A$, where $A$ can be any formula, and $y$ is a new variable not occurring in $A$. Whereas polynomial-size Frege proofs allow a polynomial number of lines, each of which must be a polynomial-sized formula, using the new axiom, polynomial-size EF proofs allow a polynomial number of lines, each of which can essentially be a polynomial-sized circuit (you can think of the new variables introduced by this axiom schema as names for the gates of a circuit, in that once a formula is named by a single variable, it can be reused without having to create another copy of the whole formula). In particular, a natural definition of $\cc{P/poly}$-Frege is equivalent to Extended Frege. Extended Frege is the strongest natural system known for proving propositional tautologies. One may also consider seemingly much stronger systems such as Peano Arithmetic or ZFC, but it is unclear and unknown if these systems can prove Boolean tautologies (with no quantifiers) any more efficiently than Extended Frege.

We define all of the algebraic systems we consider in Section~\ref{sec:others} below.

\subsection{Foundational results} \label{sec:general}
\subsubsection{Relation with \texorpdfstring{$\cc{coMA}$}{coMA}}

\begin{proposition} \label{prop:coMA}
For any field $\F$, if every propositional tautology has a polynomial-size constant-free $\I_{\F}$-proof, then $\cc{NP} \subseteq \cc{coMA}$, and hence the polynomial hierarchy collapses to its second level. 
\end{proposition}

If we wish to drop the restriction of ``constant-free'' (which, recall, is no restriction at all over a finite field), we may do so either by using the Blum--Shub--Smale analogs of $\cc{NP}$ and $\cc{coMA}$ using essentially the same proof, or over fields of characteristic zero using the Generalized Riemann Hypothesis (Proposition~\ref{prop:coMAGRH}).

\begin{proof}
Merlin nondeterministically guesses the polynomial-size constant-free \I proof, and then Arthur must check conditions (\ref{condition:ideal}) and (\ref{condition:nss}) of Definition~\ref{def:IPS}. (We need constant-free so that the algebraic proof has polynomial bit-size and thus can in fact be guessed by a Boolean Merlin.) Both conditions of Definition~\ref{def:IPS} are instances of Polynomial Identity Testing (PIT), which can thus be solved in randomized polynomial time by the standard Schwarz--Zippel--DeMillo--Lipton $\cc{coRP}$ algorithm for PIT.
\end{proof}

\subsubsection{Chasms, depth reduction, and other circuit transformations}
Recently, many strong depth reduction theorems have been proved for circuit complexity \cite{agrawalVinay,koiranChasm,GKKSchasm,tavenas}, which have been called ``chasms'' since Agrawal and Vinay \cite{agrawalVinay}. In particular, they imply that sufficiently strong lower bounds against depth 3 or 4 circuits imply super-polynomial lower bounds against arbitrary circuits. Since an \I proof is just a circuit, these depth reduction chasms apply equally well to \I proof size. Note that it was not clear to us how to adapt the proofs of these chasms to the type of circuits used in the Polynomial Calculus or other previous algebraic systems \cite{pitassiICM}, and indeed this was part of the motivation to move to our more general notion of \I proof.

\begin{observation}[Chasms for \I proof size] \label{obs:chasm}
If a system of $n^{O(1)}$ polynomial equations in $n$ variables has an \I proof of unsatisfiability of size $s$ and (semantic) degree $d$, then it also has:
\begin{enumerate}
\item A $O(\log d (\log s + \log d))$-depth \I proof of size $poly(ds)$ (follows from Valiant--Skyum--Berkowitz--Rackoff \cite{VSBR});

\item A depth 4 \I formula proof of size $n^{O(\sqrt{d})}$ (follows from Koiran \cite{koiranChasm}) or a depth 4 \I proof of size $2^{O(\sqrt{d \log(ds) \log n})}$ (follows from Tavenas \cite{tavenas}).

\item (Over fields of characteristic zero) A depth 3 \I proof of size $2^{O(\sqrt{d \log d \log n \log s})}$ (follows from Gupta, Kayal, Kamath, and Saptharishi \cite{GKKSchasm}) or even $2^{O(\sqrt{d \log n \log s})}$ (follows from Tavenas \cite{tavenas}). \rmkqed
\end{enumerate}
\end{observation}

This observation helps explain why size lower bounds for algebraic proofs for the stronger notion of size---namely number of lines, used here and in Pitassi \cite{pitassi96}, rather than number of monomials---have been difficult to obtain. This also suggests that size lower bounds for \I proofs in restricted circuit classes would be interesting, even for restricted kinds of depth 3 circuits.

Similarly, since \I proofs are just circuits, any \I certificate family of polynomially bounded degree that is computed by a polynomial-size family of algebraic circuits with divisions can also be computed by a polynomial-size family of algebraic circuits without divisions (follows from Strassen \cite{strassenDivision}). We note, however, that one could in principle consider \I certificates that were not merely polynomials, but even rational functions, under suitable conditions; divisions for computing these cannot always be eliminated. We discuss this ``Rational Ideal Proof System,'' the exact conditions needed, and when such divisions can be effectively eliminated in Appendix~\ref{app:RIPS}.

\subsubsection{Simulations and definitions of other algebraic proof systems in terms of \texorpdfstring{$\I$}{\Itext}}
\label{sec:others}
Previously studied algebraic proof systems can be viewed as particular complexity measures on the Ideal Proof System, including the Polynomial Calculus (or \Grobner) proof system (PC) \cite{CEI}, Polynomial Calculus with Resolution (PCR) \cite{PCR}, the Nullstellensatz proof system \cite{BIKPP}, and Pitassi's algebraic systems \cite{pitassi96, pitassiICM}, as we explain below.

Before explaining these, we note that although the Nullstellensatz says that if $F_1(\vec{x}) = \dotsb = F_m(\vec{x}) = 0$ is unsatisfiable then there always exists a certificate that is linear in the $\f_i$---that is, of the form $\sum \f_i G_i(\vec{x})$---our definition of \I certificate does not enforce $\vec{\f}$-linearity. The definition of \I certificate allows certificates with $\vec{\f}$-monomials of higher degree, and it is conceivable that one could achieve a savings in size by considering such certificates rather than only considering $\vec{\f}$-linear ones. As the linear form is closer to the original way Hilbert expressed the Nullstellensatz (see, \eg, the translation \cite{hilbertPapers}), we refer to certificates of the form $\sum \f_i G_i(\vec{x})$ as \definedWord{Hilbert-like \I certificates}. 

All of the previous algebraic proof systems are rule-based systems, in that they syntactically enforce the condition that every line of the proof is a polynomial in the ideal of the original polynomials $F_1(\vec{x}), \dotsc, F_m(\vec{x})$. Typically they do this by allowing two derivation rules: 1) from $G$ and $H$, derive $\alpha G + \beta H$ for $\alpha,\beta$ constants, and 2) from $G$, derive $Gx_i$ for any variable $x_i$. By ``rule-based circuits'' we mean circuits with inputs $\f_1, \dotsc, \f_m$ having linear combination gates and, for each $i=1,\dotsc,n$, gates that multiply their input by $x_i$. (Alternatively, one may view the $x_i$ as inputs, require that the circuit by syntactically \emph{linear} in the $\f_i$, and that each $x_i$ is only an input to multiplication gates, each of which syntactically depends on at least one $\f_i$. Again alternatively, one may view the $x_i$ as inputs, but with the requirement that the polynomial computed \emph{at each gate} is a polynomial of $\f_i$-degree one in the ideal $\langle \f_1, \dotsc, \f_m \rangle \subseteq \F[\vec{x}, \vec{\f}]$.) In particular, rule-based circuits necessarily produce Hilbert-like certificates.

Now we come to the definitions of previous algebraic proof systems in terms of complexity measures on the Ideal Proof System:
\begin{itemize}
\item Complexity in the Nullstellensatz proof system, or ``Nullstellensatz degree,'' is simply the minimal degree of any Hilbert-like certificate (for systems of equations of constant degree, such as the algebraic translations of tautologies.)

\item ``Polynomial Calculus size'' is the sum of the (semantic) number of monomials at each gate in $C(\vec{x}, \vec{F}(\vec{x}))$, where $C$ ranges over rule-based circuits.

\item ``PC degree'' is the minimum over rule-based circuits $C(\vec{x}, \vec{\f})$ of the maximum semantic degree at any gate in $C(\vec{x}, \vec{F}(\vec{x}))$. 

\item Pitassi's 1997 algebraic proof system \cite{pitassiICM} is essentially PC, except where size is measured by number of lines of the proof (rather than total number of monomials appearing). This corresponds exactly to the smallest size of any rule-based circuit $C(\vec{x}, \vec{\f})$ computing any Hilbert-like \I certificate.

\item Polynomial Calculus with Resolution (PCR) \cite{PCR} also allows variables $\overline{x}_i$ and adds the equations $\overline{x}_i = 1 - x_i$ and $x_i \overline{x}_i = 0$. This is easily accommodated into the Ideal Proof System: add the $\overline{x}_i$ as new variables, with the same restrictions as are placed on the $x_i$'s in a rule-based circuit, and add the polynomials $\overline{x}_i - 1 + x_i$ and $x_i \overline{x}_i$ to the list of equations $F_i$. Note that while this may have an effect on the PC size as it can decrease the total number of monomials needed, it has essentially no effect on the number of lines of the proof.
\end{itemize}

\begin{pitassiProp}
Pitassi's 1996 algebraic proof system \cite{pitassi96} is p-equivalent to Hilbert-like \I.

Pitassi's 1997 algebraic proof system \cite{pitassiICM}---equivalent to the number-of-lines measure on PC proofs---is p-equivalent to Hilbert-like $\det$-\I or $\cc{VP}_{ws}$-\I.
\end{pitassiProp}

Combining Proposition~\ref{prop:pitassi} with the techniques used in Theorem~\ref{thm:VNP} shows that super-polynomial lower bounds on the number of lines in PC proofs would positively resolve the Permanent Versus Determinant Conjecture, explaining the difficulty of such proof complexity lower bounds.

In light of this proposition (which we prove in Section~\ref{sec:pitassi}), 
we henceforth refer to the systems from \cite{pitassi96} and \cite{pitassiICM} as Hilbert-like \I and Hilbert-like $\det$-\I, respectively. Pitassi \cite[Theorem~1]{pitassi96} showed that Hilbert-like \I p-simulates Polynomial Calculus and Frege. Essentially the same proof shows that Hilbert-like \I p-simulates Extended Frege as well. 

Unfortunately, the proof of the simulation in \cite{pitassi96} does not seem to generalize to give a depth-preserving simulation. Nonetheless, our next proposition shows that there is indeed a depth-preserving simulation.

\begin{depthThm}
For any $d(n)$, depth-$(d+2)$ $\I_{\F_p}$ p-simulates depth-$d$ Frege proofs with unbounded fan-in $\lor, \land, MOD_p$ connectives (for $d=O(1)$, this is $\cc{AC}^0_{d}[p]$-Frege).
\end{depthThm}

\subsection{Lower bounds on \texorpdfstring{\I}{\Itext} imply circuit lower bounds} \label{sec:VNPeabs}
\begin{VNPthm}
A super-polynomial lower bound on [constant-free] Hilbert-like $\I_{R}$ proofs of any family of tautologies implies $\cc{VNP}_{R} \neq \cc{VP}_{R}$ [respectively, $\cc{VNP}^0_{R} \neq \cc{VP}^0_{R}$], for any ring $R$.

A super-polynomial lower bound on the number of lines in Polynomial Calculus proofs implies the Permanent versus Determinant Conjecture ($\cc{VNP} \neq \cc{VP}_{ws}$).
\end{VNPthm}

Together with Proposition~\ref{prop:coMA}, this immediately gives an alternative, and we believe simpler, proof of the following result:

\begin{corollary}
If $\cc{NP} \not\subseteq \cc{coMA}$, then $\cc{VNP}^{0}_{R} \neq \cc{VP}^0_{R}$, for any ring $R$.
\end{corollary}

For comparison, here is a brief sketch of the only previous proof of this result that we are aware of, which only seems to work when $R$ is a finite field or, assuming the Generalized Riemann Hypothesis, a field of characteristic zero, and uses several other significant results. The previous proof combines: 1) \Burgisser's results \cite{burgisserCookValiant} relating $\cc{VP}$ and $\cc{VNP}$ over various fields to standard Boolean complexity classes such as $\cc{NC/poly}$, $\cc{\# P/poly}$ (uses GRH), and $\cc{Mod}_{p}\cc{P/poly}$, and 2) the implication $\cc{NP} \not\subseteq \cc{coMA} \Rightarrow \cc{NC/poly} \neq \cc{\# P/poly}$ (and similarly with $\cc{\# P/poly}$ replaced by $\cc{Mod}_{p}\cc{P/poly}$), which uses the downward self-reducibility of complete functions for $\cc{\# P/poly}$ (the permanent \cite{valiant}) and $\cc{Mod}_{p}\cc{P/poly}$ \cite{feigenbaumFortnow}, as well as Valiant--Vazirani \cite{VV}.

The following lemma is the key to Theorem~\ref{thm:VNP}.

\begin{VNPlem}
Every family of CNF tautologies $(\varphi_n)$ has a Hilbert-like family of \I certificates $(C_n)$ in $\cc{VNP}^{0}_{R}$.
\end{VNPlem}

Here we show how Theorem~\ref{thm:VNP} follows from Lemma~\ref{lem:VNP}. Lemma~\ref{lem:VNP} is proved in Section~\ref{sec:VNP}.

\begin{proof}[Proof of Theorem~\ref{thm:VNP}, assuming Lemma~\ref{lem:VNP}]
For a given set $\mathcal{F}$ of unsatisfiable polynomial equations $F_1=\dotsb=F_m=0$, a lower bound on \I refutations of $\mathcal{F}$ is equivalent to giving the same circuit lower bound on \emph{all} \I certificates for $\mathcal{F}$. A super-polynomial lower bound on Hilbert-like \I implies that some function in $\cc{VNP}$---namely, the $\cc{VNP}$-\I certificate guaranteed by Lemma~\ref{lem:VNP}---cannot be computed by polynomial-size algebraic circuits, and hence that $\cc{VNP} \neq \cc{VP}$. Since Lemma~\ref{lem:VNP} even guarantees a constant-free certificate, we get the analogous consequence for constant-free lower bounds.

The second part of Theorem~\ref{thm:VNP} follows from the fact that number of lines in a PC proof is p-equivalent to Hilbert-like $\det$-\I (Proposition~\ref{prop:pitassi}). As in the first part, a super-polynomial lower bound on Hilbert-like $\det$-\I implies that some function family in $\cc{VNP}$ is not a p-projection of the determinant. Since the permanent is $\cc{VNP}$-complete under p-projections, the result follows.
\end{proof}

\subsection{PIT as a bridge between circuit complexity and proof complexity} \label{sec:PITeabs}
In this section we state our PIT axioms and give an outline of the proof of Theorems~\ref{thm:EF} and \ref{thm:AC0}, which say that Extended Frege (EF) (respectively, $\cc{AC}^0$- or $\cc{AC}^0[p]$-Frege) is polynomially equivalent to the Ideal Proof System if there are polynomial-size circuits for PIT whose correctness---suitably formulated---can be efficiently proved in EF (respectively, $\cc{AC}^0$- or $\cc{AC}^0[p]$-Frege).
More precisely, we identify a small set of natural axioms for PIT and show that if these axioms can be proven efficiently in EF, then EF is p-equivalent to $\I$. Theorem~\ref{thm:AC0} begins to explain why $\cc{AC}^0[p]$-Frege lower bounds have been so difficult to obtain, and highlights the importance of our PIT axioms for $\cc{AC}^0[p]$-Frege lower bounds.  We begin by describing and discussing these axioms.

Fix some standard Boolean encoding of constant-free algebraic circuits, so that the encoding of any size-$m$ constant-free algebraic circuit has size $\poly(m)$. We use ``$[C]$'' to denote the encoding of the algebraic circuit $C$. Let $K = \{K_{m,n}\}$ denote a family of Boolean circuits for solving polynomial identity testing. That is, $K_{m,n}$ is a Boolean function that takes as input the encoding of a size $m$ constant-free algebraic circuit, $C$, over variables $x_1, \ldots, x_n$, and if $C$ has polynomial degree, then $K$ outputs 1 if and only if the polynomial computed by $C$ is the 0 polynomial.

\paragraph{Notational convention:} We underline parts of a statement that involve propositional variables. For example, if in a propositional statement we write ``$[C]$'', this refers to a fixed Boolean string that is encoding the (fixed) algebraic circuit $C$. In contrast, if we write $\prop{[C]}$, this denotes a Boolean string \emph{of propositional variables}, which is to be interpreted as a description of an as-yet-unspecified algebraic circuit $C$; any setting of the propositional variables corresponds to a particular algebraic circuit $C$. Throughout, we use $\vec{p}$ and $\vec{q}$ to denote propositional variables (which we do not bother underlining except when needed for emphasis), and $\vec{x}, \vec{y}, \vec{z},\dotsc$ to denote the algebraic variables that are the inputs to algebraic circuits. Thus, $C(\vec{x})$ is an algebraic circuit with inputs $\vec{x}$, $[C(\vec{x})]$ is a fixed Boolean string encoding some particular algebraic circuit $C$, $\prop{[C(\vec{x})]}$ is a string of propositional variables encoding an unspecified algebraic circuit $C$, and $[C(\prop{\vec{p}})]$ denotes a Boolean string together with propositional variables $\vec{p}$ that describes a fixed algebraic circuit $C$ whose inputs have been set to the propositional variables $\vec{p}$.

\begin{definition} \label{def:PITaxioms}
Our PIT axioms for a Boolean circuit $K$ are as follows. (This definition makes sense even if $K$ does not correctly compute PIT, but that case isn't particularly interesting or useful.)

\begin{enumerate}
\item\label{axiom:Boolean} Intuitively, the first axiom states that if $C$ is a circuit computing
the identically 0 polynomial, then the polynomial evaluates to 0 on all
Boolean inputs.
\[
K(\prop{[C(\vec{x})]}) \rightarrow K(\prop{[C(\vec{p})]})
\]
Note that the only variables on the left-hand side of the implication
are Boolean propositional variables, $\vec{q}$, that encode an algebraic circuit of size $m$ over
$n$ algebraic variables $\vec{x}$ (these latter are \emph{not} propositional variables of the above formula). The variables on the right-hand side are $\vec{q}$ plus
Boolean variables ${\vec p}$, where some of the variables in $\vec{q}$---those encoding the $x_i$---have been replaced by constants or $\vec{p}$ in such a way that $[C(\vec{p})]$ encodes a circuit that plugs in the $\{0,1\}$-valued $p_i$ for its algebraic inputs $x_i$. In other words, when we say $\prop{[C({\vec p})]}$ we mean
the encoding of the circuit $C$ where Boolean constants are plugged in for
the original algebraic $\vec{x}$ variables, as specified by the variables $\vec{p}$.

\item\label{axiom:one} Intuitively, the second axiom states that if $C$ is a circuit computing
the zero polynomial, then the circuit $1-C$ does not compute the zero polynomial.
\[
K(\prop{[C({\vec x})]}) \rightarrow \neg K(\prop{[1-C({\vec x})]})
\]
Here, if $\vec{q}$ are the propositional variables describing $C$, these are the only variables that appear in the above statement. We abuse syntax slightly in writing $[1-C]$: it is meant to denote a Boolean formula $\varphi(\vec{q})$ such that if $\vec{q}=[C]$ describes a circuit $C$, then $\varphi(\vec{q})$ describes the circuit $1-C$ (with one subtraction gate more than $C$).

\item\label{axiom:subzero} Intuitively, the third axiom states that PIT circuits respect certain substitutions.
More specifically, if the polynomial computed by circuit
$G$ is 0, then $G$ can be substituted for the constant $0$.
\[
K(\prop{[G({\vec x})]}) \land K(\prop{[C({\vec x},0)]}) \rightarrow
	K(\prop{[C({\vec x},G({\vec x}))]})
\]
Here the notations $[C(\vec{x},0)]$ and $[C(\vec{x}, G(\vec{x}))]$ are similar abuses of notation to above; we use these and similar shorthands without further mention.

\item\label{axiom:perm} Intuitively, the last axiom states that PIT is closed under
permutations of the (algebraic) variables. More specifically if $C(\vec{x})$ is identically 0,
then so is $C(\pi(\vec{x}))$ for all permutations $\pi$.
\[
K(\prop{[C({\vec x})]}) \rightarrow K(\prop{[C(\pi({\vec x}))]})
\]
\end{enumerate}

\end{definition}

We can now state and discuss two of our main theorems precisely.

\begin{EFthm}
If there is a family $K$ of polynomial-size Boolean circuits that correctly compute PIT, such that the PIT axioms for $K$ have polynomial-size EF proofs, then EF is polynomially equivalent to $\I$.
\end{EFthm}

Note that the issue is not the existence of small circuits for PIT since we would be happy with nonuniform polynomial-size PIT circuits, which do exist. Unfortunately the known constructions are highly nonuniform---they involve picking uniformly random points---and we do not see how to prove the above axioms for these constructions. Nonetheless, it seems very plausible to us that there exists a polynomial-size family of PIT circuits where the above axioms are efficiently provable in EF, especially in light of Remark~\ref{rmk:PIT}.

To prove the theorem (which we do in Section~\ref{sec:EF}), we first show that EF is p-equivalent to $\I$ if a family of propositional formulas expressing soundness of $\I$ are  efficiently EF provable. Then we show that efficient EF proofs of $Soundness_{\I}$ follows from efficient EF proofs for the PIT axioms.

Our next main result shows that the previous result can be scaled down to much weaker proof systems than EF.

\begin{AC0thm}
Let $\mathcal{C}$ be any class of circuits closed under $\cc{AC}^0$ circuit reductions. If there is a family $K$ of polynomial-size Boolean circuits computing PIT such that the PIT axioms for $K$ have polynomial-size $\mathcal{C}$-Frege proofs, then $\mathcal{C}$-Frege is polynomially equivalent to $\I$, and consequently polynomially equivalent to Extended Frege.
\end{AC0thm}

Note that here we \emph{do not} need to restrict the circuit family $K$ to be in the class $\mathcal{C}$. This requires one more (standard) technical device compared to the proof of Theorem~\ref{thm:EF}, namely the use of auxiliary variables for the gates of $K$. Here we prove and discuss some corollaries of Theorem~\ref{thm:AC0}; the proof of Theorem~\ref{thm:AC0} is given in Section~\ref{sec:AC0p}.

As $\cc{AC}^0$ is known unconditionally to be strictly weaker than Extended Frege \cite{Ajtai}, we immediately get that $\cc{AC}^0$-Frege cannot efficiently prove the PIT axioms for any Boolean circuit family $K$ correctly computing PIT. 

Using essentially the same proof as Theorem~\ref{thm:AC0}, we also get the following result. By ``depth $d$ PIT axioms'' we mean a variant where the algebraic circuits $C$ (encoded as $[C]$ in the statement of the axioms) have depth at most $d$. Note that, even over finite fields, for any $d \geq 4$ super-polynomial lower bounds on depth $d$ algebraic circuits are a notoriously open problem. (The chasm at depth $4$ says that depth $4$ lower bounds of size $2^{\omega(\sqrt{n} \log n)}$ imply super-polynomial size lower bounds on general algebraic circuits, but this does not give any indication of why merely super-polynomial lower bounds on depth $4$ circuits should be difficult.)

\begin{corollary} \label{cor:AC0p}
For any $d$, if there is a family of tautologies with no polynomial-size $\cc{AC}^0[p]$-Frege proof, and $\cc{AC}^0[p]$-Frege has polynomial-size proofs of the [depth $d$] PIT axioms for some $K$, then $\cc{VNP}_{\F_p}$ does not have polynomial-size [depth $d$] algebraic circuits.
\end{corollary}

This corollary makes the following question of central importance in getting lower bounds on $\cc{AC}^0[p]$-Frege:

\begin{open}
For some $d \geq 4$, is there some $K$ computing depth $d$ PIT, for which the depth $d$ PIT axioms have $\cc{AC}^0[p]$-Frege proofs of polynomial size?
\end{open}

This question has the virtue that answering it either way is highly interesting:
\begin{itemize}
\item If $\cc{AC}^0[p]$-Frege does not have polynomial-size proofs of the [depth $d$] PIT axioms for any $K$, then we have super-polynomial size lower bounds on $\cc{AC}^0[p]$-Frege, answering a question that has been open for nearly thirty years.

\item Otherwise, super-polynomial size lower bounds on $\cc{AC}^0[p]$-Frege imply that the permanent does not have polynomial-size algebraic circuits [of depth $d$] over any finite field of characteristic $p$. This would then explain why getting superpolynomial lower bounds on $\cc{AC}^0[p]$-Frege has been so difficult.
\end{itemize}

This dichotomy is in some sense like a ``completeness result for $\cc{AC}^0[p]$-Frege, modulo proving strong algebraic circuit lower bounds on $\cc{VNP}$'': if one hopes to prove $\cc{AC}^0[p]$-Frege lower bounds \emph{without proving} strong lower bounds on $\cc{VNP}$, then one must prove $\cc{AC}^0[p]$-Frege lower bounds on the PIT axioms. For example, if you believe that proving $\cc{VP} \neq \cc{VNP}$ [or that proving $\cc{VNP}$ does not have bounded-depth polynomial-size circuits] is very difficult, and that proving $\cc{AC}^0[p]$-Frege lower bounds is comparatively easy, then to be consistent you must also believe that proving $\cc{AC}^0[p]$-Frege lower bounds \emph{on the [bounded-depth] PIT axioms} is easy.

Similarly, along with Theorem~\ref{thm:depth}, we get the following corollary.

\begin{corollary}
If for every constant $d$, there is a constant $d'$ such that the depth $d$ PIT axioms have polynomial-size depth $d'$ $\cc{AC}^0_{d'}[p]$-Frege proofs , then $\cc{AC}^0[p]$-Frege is polynomially equivalent to constant-depth $\I_{\F_p}$.
\end{corollary}

Using the chasms at depth 3 and 4 for algebraic circuits \cite{agrawalVinay,koiranChasm,tavenas} (see Observation~\ref{obs:chasm} above), we can also help explain why sufficiently strong exponential lower bounds for $\cc{AC}^0$-Frege---that is, lower bounds that don't depend on the depth, or don't depend so badly on the depth, which have also been open for nearly thirty years---have been difficult to obtain:

\begin{corollary}
Let $\F$ be any field, and let $c$ be a sufficiently large constant. If there is a family of tautologies $(\varphi_n)$ such that any $\cc{AC}^0$-Frege proof of $\varphi_n$ has size at least $2^{c\sqrt{n} \log n}$, and $\cc{AC}^0$-Frege has polynomial-size proofs of the depth $4$ PIT$_{\F}$ axioms for some $K$, then $\cc{VP}^0_{\F} \neq \cc{VNP}^0_{\F}$.

If $\F$ has characteristic zero, we may replace ``depth $4$'' above with ``depth $3$.''
\end{corollary}

\begin{proof}
Suppose that $\cc{AC}^0$-Frege can efficiently prove the depth $4$ PIT$_\F$ axioms for some Boolean circuit $K$. Let $(\varphi_n)$ be a family of tautologies. If $\cc{VNP}^0_{\F} = \cc{VP}^0_{\F}$, then there is a polynomial-size \I proof of $\varphi_n$. By Observation~\ref{obs:chasm}, the same certificate is computed by a depth $4$ $\F$-algebraic circuit of size $2^{O(\sqrt{n} \log n)}$. By assumption, $\cc{AC}^0$-Frege can efficiently prove the depth $4$ PIT$_\F$ axioms for $K$, and therefore $\cc{AC}^0$-Frege p-simulates depth 4 \I. Thus there are $\cc{AC}^0$-Frege proofs of $\varphi_n$ of size $2^{O(\sqrt{n} \log n)}$. 

If $\F$ has characteristic zero, we may instead use the best-known chasm at depth 3, for which we only need depth 3 PIT and depth 3 \I, and yields the same bounds.
\end{proof}

As with Corollary~\ref{cor:AC0p}, we conclude a similar dichotomy: either $\cc{AC}^0$-Frege can efficiently prove the depth 4 PIT axioms (depth 3 in characteristic zero), or proving $2^{\omega(\sqrt{n} \log n)}$ lower bounds on $\cc{AC}^0$-Frege implies $\cc{VP}^0 \neq \cc{VNP}^0$.

\subsection{Towards lower bounds} \label{sec:syzygy}
Theorem~\ref{thm:VNP} shows that proving lower bounds on (even Hilbert-like) $\I$, or on the number of lines in Polynomial Calculus proofs (equivalent to Hilbert-like $\det$-\I),  is at least as hard as proving algebraic circuit lower bounds. In this section we begin to make the difference between proving proof complexity lower bounds and proving circuit lower bounds more precise, and use this precision to suggest a direction for proving new proof complexity lower bounds, aimed at proving the long-sought-for length-of-proof lower bounds on an algebraic proof system.

The key fact we use is embodied in Lemma~\ref{lem:fgsyz}, which says that the set of (Hilbert-like) certificates for a given unsatisfiable system of equations is, in a precise sense, ``finitely generated.'' The basic idea is then to leverage this finite generation to extend lower bound techniques from individual polynomials to entire ``finitely generated'' sets of polynomials.

Because Hilbert-like certificates are somewhat simpler to deal with, we begin with those and then proceed to general certificates. But keep in mind that all our key conclusions about Hilbert-like certificates will also apply to general certificates. For this section we will need the notion of a module over a ring (the ring-analogue of a vector space over a field) and a few basic results about such modules; these are reviewed in Appendix~\ref{app:background:algebra}.

Recall that a \definedWord{Hilbert-like} $\I$-certificate $C(\vec{x}, \vec{\f})$ is one that is linear in the $\f$-variables, that is, it has the form $\sum_{i=1}^{m}G_i(\vec{x}) \f_i$.
Each function of the form $\sum_i G_i(\vec{x})\f_i$ is completely determined by the tuple $(G_1(\vec{x}), \dotsb, G_m(\vec{x}))$, and the set of all such tuples is exactly the $R[\vec{x}]$-module $R[\vec{x}]^{m}$. 

The algebraic circuit size of a Hilbert-like certificate $C=\sum_i G_i(\vec{x}) \f_i$ is equivalent (up to a small constant factor and an additive $O(n)$) to the algebraic circuit size of computing the entire tuple $(G_1(\vec{x}), \dotsc, G_m(\vec{x}))$. A circuit computing the tuple can easily be converted to a circuit computing $C$ by adding $m$ times gates and a single plus gate. 
Conversely, for each $i$ we can recover $G_i(\vec{x})$ from $C(\vec{x}, \vec{\f})$ by plugging in $0$ for all $\f_j$ with $j \neq i$ and $1$ for $\f_i$. 
So from the point of view of lower bounds, we may consider Hilbert-like certificates, and their representation as tuples, essentially without loss of generality. This holds even in the setting of Hilbert-like depth 3 \I-proofs.

Using the representation of Hilbert-like certificates as tuples, we find that Hilbert-like \I-certificates are in bijective correspondence with $R[\vec{x}]$ solutions (in the new variables $g_i$) to the following $R[\vec{x}]$-linear equation:
\[
\left(\begin{array}{ccc}
F_1(\vec{x}) & \dotsb & F_m(\vec{x}) 
\end{array}\right)
\left(\begin{array}{c}
g_1 \\
\vdots \\
g_m
\end{array}
\right) = 1
\]
Just as in linear algebra over a field, the set of such solutions can be described by taking one solution and adding to it all solutions to the associated homogeneous equation:
\begin{equation} \label{eqn:homog}
\left(\begin{array}{ccc}
F_1(\vec{x}) & \dotsb & F_m(\vec{x}) 
\end{array}\right)
\left(\begin{array}{c}
g_1 \\
\vdots \\
g_m
\end{array}
\right) = 0
\end{equation}
(To see why this is so, mimic the usual linear algebra proof: given two solutions of the inhomogeneous equation, consider their difference.) Solutions to the latter equation are commonly called ``syzygies'' amongst the $F_i$. 
Syzygies and their properties are well-studied---though not always well-understood---in commutative algebra and algebraic geometry, so lower and upper bounds on Hilbert-like \I-proofs may benefit from known results in algebra and geometry. 

We now come to the key lemma for Hilbert-like certificates.

\begin{lemma} \label{lem:fgsyz}
For a given set of unsatisfiable polynomial equations $F_1(\vec{x})=\dotsb=F_m(\vec{x})=0$ over a Noetherian ring $R$ (such as a field or $\Z$), the set of Hilbert-like \I-certificates is a coset of a finitely generated submodule of $R[\vec{x}]^{m}$. 
\end{lemma}

\begin{proof}
The discussion above shows that the set of Hilbert-like certificates is a coset of a $R[\vec{x}]$-submodule of $R[\vec{x}]^{m}$, namely the solutions to (\ref{eqn:homog}). As $R$ is a Noetherian ring, so is $R[\vec{x}]$ (by Hilbert's Basis Theorem). Thus $R[\vec{x}]^{m}$ is a Noetherian $R[\vec{x}]$-module, and hence every submodule of it is finitely generated.
\end{proof}

Lemma~\ref{lem:fgsyz} seems so conceptually important that it is worth re-stating: 
\begin{quotation}
\noindent \textbf{The set of all Hilbert-like $\I$-certificates for a given system of equations can be described by giving a single Hilbert-like \I-certificate, together with a finite generating set for the syzygies.}
\end{quotation}
Its importance may be underscored by contrasting the preceding statement with the structure (if any?) of the set of all proofs in other proof systems, particularly non-algebraic ones.

Note that a finite generating set for the syzygies (indeed, even a \Grobner basis) can be found in the process of computing a \Grobner basis for the $R[\vec{x}]$-ideal $\langle F_1(\vec{x}), \dotsc, F_m(\vec{x}) \rangle$. This process is to Buchberger's \Grobner basis algorithm as the extended Euclidean algorithm is to the usual Euclidean algorithm; an excellent exposition can be found in the book by Ene and Herzog \cite{eneHerzog} (see also \cite[Section~15.5]{eisenbud}).

Lemma~\ref{lem:fgsyz} suggests that one might be able to prove size lower bounds on Hilbert-like-\I along the following lines: 1) find a single family of Hilbert-like \I-certificates $(G_n)_{n=1}^{\infty}$, $G_n = \sum_{i=1}^{\poly(n)} \f_i G_i(\vec{x})$ (one for each input size $n$), 2) use your favorite algebraic circuit lower bound technique to prove a lower bound on the polynomial family $G$, 3) find a (hopefully nice) generating set for the syzygies, and 4) show that when adding to $G$ any $R[\vec{x}]$-linear combinations of the generators of the syzygies, whatever useful property was used in the lower bound on $G$ still holds. Although this indeed seems significantly more difficult than proving a single algebraic circuit complexity lower bound, it at least suggests a recipe for proving lower bounds on Hilbert-like \I (and its subsystems such as homogeneous depth $3$, depth $4$, multilinear, etc.), which should be contrasted with the difficulty of transferring lower bounds for a circuit class to lower bounds on previous related proof systems, \eg transferring $\cc{AC}^0[p]$ lower bounds \cite{razborov,smolensky} to $\cc{AC}^0[p]$-Frege. 

This entire discussion also applies to general \I-certificates, with the following modifications. We leave a certificate $C(\vec{x}, \vec{\f})$ as is, and instead of a module of syzygies we get an ideal (still finitely generated) of what we call zero-certificates. The difference between any two \I-certificates is a zero-certificate; equivalently, a \definedWord{zero-certificate} is a polynomial $C(\vec{x}, \vec{\f})$ such that $C(\vec{x}, \vec{0}) = 0$ and $C(\vec{x}, \vec{F}(\vec{x})) = 0$ as well (contrast with the definition of \I certificate, which has $C(\vec{x}, \vec{F}(\vec{x})) = 1$). The set of \I-certificates is then the coset intersection
\[
\langle \f_1, \dotsc, \f_m \rangle \cap \left( 1 + \langle \f_1 - F_1(\vec{x}), \dotsc, \f_m - F_m(\vec{x})\rangle\right)
\]
which is either empty or a coset of the ideal of zero-certificates: $\langle \f_1, \dotsc, \f_m \rangle \cap \langle \f_1 - F_1(\vec{x}), \dotsc, \f_m - F_m(\vec{x})\rangle$. The intersection ideal $\langle \f_1, \dotsc, \f_m \rangle \cap \langle \f_1 - F_1(\vec{x}), \dotsc, \f_m - F_m(\vec{x}) \rangle$ plays the role here that the set of syzygies played for Hilbert-like \I-certificates.\footnote{Note that the ideal of zero-certificates is not merely the set of all functions in the ideal $\langle \f_1 - F_1(\vec{x}), \dotsc, \f_m - F_m(\vec{x}) \rangle$ that only involve the $\f_i$, since the ideal $\langle \f_1, \dotsc, \f_m \rangle \subseteq R[\vec{x}, \vec{\f}]$ consists of all polynomials in the $\f_i$ with coefficients in $R[\vec{x}]$. Certificates only involving the $\f_i$ do have a potentially useful geometric meaning, however, which we consider in Appendix~\ref{app:geom}.} 

A finite generating set for the ideal of zero-certificates can be computed using \Grobner bases (see, \eg, \cite[Section~3.2.1]{eneHerzog}).

Just as for Hilbert-like certificates, we get:
\begin{quotation}
\noindent \textbf{The set of all $\I$-certificates for a given system of equations can be described by giving a single \I-certificate, together with a finite generating set for the ideal of zero-certificates.} 
\end{quotation}
Our suggestions above for lower bounds on Hilbert-like \I apply \emph{mutatis mutandis} to general \I-certificates, suggesting a route to proving true size lower bounds on \I using known techniques from algebraic complexity theory.

The discussion here raises many basic and interesting questions about the complexity of sets of (families of) functions in an ideal or module, which we propose in Section~\ref{sec:conclusion}.

\subsection{Summary and open questions} \label{sec:conclusion}
We introduced the Ideal Proof System \I (Definition~\ref{def:IPS}) and showed that it is a very close algebraic analog of Extended Frege---the most powerful, natural system currently studied for proving propositional tautologies. We showed that lower bounds on \I imply (algebraic) circuit lower bounds, which to our knowledge is the first time that lower bounds on a proof system have been shown to imply any sort of computational lower bounds. Using the same techniques, we were also able to show that lower bounds on the number of \emph{lines} (rather than the usual measure of number of monomials) in Polynomial Calculus proofs also imply strong algebraic circuit lower bounds. Because proofs in \I are just algebraic circuits satisfying certain polynomial identity tests, many results from algebraic circuit complexity apply immediately to \I. In particular, the chasms at depth 3 and 4 in algebraic circuit complexity imply that lower bounds on even depth 3 or 4 \I proofs would be very interesting. We introduced natural propositional axioms for polynomial identity testing (PIT), and showed that these axioms play a key role in understanding the thirty-year open question of $\cc{AC}^0[p]$-Frege lower bounds: either there are $\cc{AC}^0[p]$-Frege lower bounds on the PIT axioms, or any $\cc{AC}^0[p]$-Frege lower bounds are as hard as showing $\cc{VP} \neq \cc{VNP}$ over a field of characteristic $p$. In appendices, we discuss a variant of the Ideal Proof System that allows divisions, and its utility and limitations, as well as a geometric variant of the Ideal Proof System which suggests further geometric properties that might be of interest for computational and proof complexity. And finally, through an analysis of the set of all \I proofs of a given unsatisfiable system of equations, we suggest how one might transfer techniques from algebraic circuit complexity to prove lower bounds on \I (and thus on Extended Frege).

The Ideal Proof System raises many new questions, not only about itself, but also about PIT, new examples of $\cc{VNP}$ functions coming from propositional tautologies, and the complexity of ideals or modules of polynomials.

In Proposition~\ref{prop:gen2Hilb} we show that if a general \I-certificate $C$ has only polynomially many $\vec{\f}$-monomials (with coefficients in $\F[\vec{x}]$), and the maximum degree of each $\f_i$ is polynomially bounded, then $C$ can be converted to a polynomial-size Hilbert-like certificate. However, without this sparsity assumption general \I appears to be stronger than Hilbert-like \I.

\begin{open}
What, if any, is the difference in size between the smallest Hilbert-like and general \I certificates for a given unsatisfiable system of equations? What about for systems of equations coming from propositional tautologies?
\end{open}

\begin{open}[Degree versus size] 
Is there a super-polynomial size separation---or indeed any nontrivial size separation---between \I certificates of degree $\leq d_{small}(n)$ and \I certificates of degree $\geq d_{large}(n)$ for some bounds $d_{small} < d_{large}$? 
\end{open}

This question is particularly interesting in the following cases: a) certificates for systems of equations coming from propositional tautologies, where $d_{small}(n) = n$ and $d_{large}(n) \geq \omega(n)$, since we know that every such system of equations has \emph{some} (not necessarily small) certificate of degree $\leq n$, and b) certificates for unsatisfiable systems of equations taking $d_{small}$ to be the bound given by the best-known effective Nullstellens\"{a}tze, which are all exponential \cite{brownawell,kollar,sombraSparse}.

\begin{open} \label{open:min}
Are there tautologies for which the certificate family constructed in Theorem~\ref{thm:VNP} is the one of minimum complexity (under p-projections or c-reductions, see Appendix~\ref{app:background:complexity})?
\end{open}

If there is any family $\varphi = (\varphi_n)$ of tautologies for which Question~\ref{open:min} has a positive answer and for which the certificates constructed in Theorem~\ref{thm:VNP} are $\cc{VNP}$-complete (Question~\ref{open:complete} below), then super-polynomial size lower bounds on \I-proofs of $\varphi$ would be equivalent to $\cc{VP} \neq \cc{VNP}$. This highlights the potential importance of understanding the structure of the set of certificates under computational reducibilities.

Since the set of all [Hilbert-like] \I-certificates is a coset of a finitely generated ideal [respectively, module], the preceding question is a special case of considering, for a given family of cosets of ideals or modules $(f^{(0)}_n + I_n)$ ($I_n \subseteq R[x_1, \dotsc, x_{\poly(n)}]$), the relationships under various reductions between all families of functions $(f_n)$ with $f_n \in f^{(0)}_n + I_n$ for each $n$. This next question is of a more general nature than the others we ask; we think it deserves further study.

\begin{question} \label{open:ideal}
Given a family of cosets of ideals $f^{(0)}_n + I_n$ (or more generally modules) of polynomials, with $I_n \subseteq R[x_1, \dotsc, x_{\poly(n)}]$, consider the function families $(f_n) \in (f^{(0)}_n + I_n)$ (meaning that $f_n \in f^{(0)}_n + I_n$ for all $n$) under any computational reducibility $\leq$ such as p-projections. What can the $\leq$ structure look like? When, if ever, is there such a unique $\leq$-minimum (even a single nontrivial example would be interesting, as in Question~\ref{open:min})? Can there be infinitely many incomparable $\leq$-minima? 

Say a $\leq$-degree $\mathbf{d}$ is ``saturated'' in $(f^{(0)}_n + I_n)$ if every $\leq$-degree $\mathbf{d'} \geq \mathbf{d}$ has some representative in $f^{(0)} + I$. Must saturated degrees always exist? We suspect yes, given that one may multiply any element of $I$ by arbitrarily complex polynomials. What can the set of saturated degrees look like for a given $(f^{(0)}_n + I_n)$? Must every $\leq$-degree in $f^{(0)} + I$ be \emph{below} some saturated degree? What can the $\leq$-structure of $f^{(0)} + I$ look like below a saturated degree?
\end{question}

Question~\ref{open:ideal} is of interest even when $f^{(0)} = 0$, that is, for ideals and modules of functions rather than their nontrivial cosets. 

\begin{open}
Can we leverage the fact that the set of \I certificates is not only a finitely generated coset intersection, but also closed under multiplication?
\end{open}

We note that it is not difficult to show that a coset $c + I$ of an ideal is closed under multiplication if and only if $c^2 - c \in I$. Equivalently, this means that $c$ is idempotent ($c^2 = c$) in the quotient ring $R/I$. For example, if $I$ is a prime ideal, then $R/I$ has no zero-divisors, and thus the only choices for $c+I$ are $I$ and $1+I$. We note that the ideal generated by the $n^2$ equations $XY-I=0$ in the setting of the Hard Matrix Identities is prime (see Appendix~\ref{app:RIPS}). It seems unlikely that all ideals coming from propositional tautologies are prime, however. 

The complexity of \Grobner basis computations obviously depends on the degrees and the number of polynomials that one starts with. From this point of view, Mayr and Meyer \cite{mayrMeyer} showed that the doubly-exponential upper bound on the degree of a \Grobner basis \cite{hermann} (see also \cite{seidenberg,masserWustholz}) could not be improved in general. However, in practice many \Grobner basis computations seem to work much more efficiently, and even theoretically many classes of instances---such as proving that $1$ is in a given ideal---can be shown to have only a singly-exponential degree upper bound \cite{brownawell, kollar, sombraSparse}. These points of view are reconciled by the more refined measure of the (Castelnuovo--Mumford) \emph{regularity} of an ideal or module. For the definition of regularity and a discussion of its close connection with the complexity of \Grobner basis and syzygy computations, we refer the reader to the original papers \cite{bayerStillman1, bayerStillman2, bayerStillman3} or the survey \cite{bayerMumfordSurvey}.

Given that the syzygy module or ideal of zero-certificates are so crucial to the complexity of \I-certificates, and the tight connection between these modules/ideals and the computation of the \Grobner basis of the ideal one started with, we ask:

\begin{question}
Is there a formal connection between the proof complexity of individual instances of TAUT (in, say, the Ideal Proof System), and the Castelnuovo--Mumford regularity of the corresponding syzygy module or ideal of zero-certificates?
\end{question}

The certificates constructed in the proof of Theorem~\ref{thm:VNP} provide many new examples of polynomial families in $\cc{VNP}$. There are many natural questions one can ask about these polynomials. For example, the construction itself depends on the order of the clauses; does the complexity of the resulting polynomial family depend on this order? As another example, we suspect that, for any $\equiv_{p}$ or $\equiv_{c}$-degree within $\cc{VNP}$ (see Appendix~\ref{app:background:complexity}), there is some family of tautologies for which the above polynomials are of that degree. However, we have not yet proved this for even a single degree.

\begin{open} \label{open:complete}
Are there tautologies for which the certificates constructed in Theorem~\ref{thm:VNP} are $\cc{VNP}$-complete? More generally, for any given $\equiv_{p}$ or $\equiv_{c}$-degree within $\cc{VNP}$, are there tautologies for which this certificate is of that degree?
\end{open}

Prior to our work, much work was done on bounds for the Ideal Membership Problem---$\cc{EXPSPACE}$-complete \cite{mayrMeyer, mayrEXPSPACE}---the so-called Effective Nullstellensatz---where exponential degree bounds are known, and known to be tight \cite{brownawell,kollar,sombraSparse,einLazarsfeld}---and the arithmetic Nullstellensatz over $\Z$, where one wishes to bound not only the degree of the polynomials but the sizes of the integer coefficients appearing \cite{krickPardoSombra}. The viewpoint afforded by the Ideal Proof Systems raises new questions about potential strengthening of these results.

In particular, the following is a natural extension of Definition~\ref{def:IPS}.

\begin{definition} \label{def:IPSideal}
An \definedWord{\I certificate} that a polynomial $G(\vec{x}) \in \F[\vec{x}]$ is in the ideal [respectively, radical of the ideal] generated by $F_1(\vec{x}), \dotsc, F_m(\vec{x})$ is a polynomial $C(\vec{x}, \vec{\f})$ such that
\begin{enumerate}
\item $C(\vec{x}, \vec{0}) = 0$, and

\item $C(\vec{x}, F_1(\vec{x}), \dotsc, F_m(\vec{x})) = G(\vec{x})$ [respectively, $G(\vec{x})^{k}$ for any $k > 0$].
\end{enumerate}
An \definedWord{\I derivation} of $G$ from $F_1, \dotsc, F_m$ is a circuit computing some \I certificate that $G \in \langle F_1, \dotsc, F_m \rangle$.
\end{definition}

For the Ideal Membership Problem, the $\cc{EXPSPACE}$ lower bound \cite{mayrMeyer, mayrEXPSPACE} implies an subexponential-size lower bound on constant-free circuits computing \I-certificates of ideal membership (or non-constant-free circuits in characteristic zero, assuming GRH, see Proposition~\ref{prop:coMAGRH}). Here by ``sub-exponential'' we mean a function $f(n) \in \bigcap_{\varepsilon > 0} O(2^{n^{\varepsilon}})$. Indeed, if for every $G(\vec{x}) \in \langle F_1(\vec{x}), \dotsc, F_m(\vec{x}) \rangle$ there were a constant-free circuit of subexponential size computing some \I certificate for the membership of $G$ in $\langle F_1, \dotsc, F_m \rangle$, then guessing that circuit and verifying its correctness using PIT gives a $\cc{MA}_{\text{subexp}} \subseteq \cc{SUBEXPSPACE}$ algorithm for the Ideal Membership Problem. The $\cc{EXPSPACE}$-completeness of Ideal Membership would then imply that $\cc{EXPSPACE} \subseteq \cc{SUBEXPSPACE}$, contradicting the Space Hierarchy Theorem \cite{hartmanisStearns}. Under special circumstances, of course, one may be able to achieve better upper bounds.

However, for the effective Nullstellensatz and its arithmetic variant, we leave the following open:

\begin{open}
For any $G, F_1, \dotsc, F_m$ on $x_1, \dotsc, x_n$, as in Definition~\ref{def:IPSideal}, is there always an \I-certificate of subexponential size that $G$ is in the \emph{radical} of $\langle F_1, \dotsc, F_m \rangle$? Similarly, if $G, F_1, \dotsc, F_m \in \Z[x_1, \dotsc, x_n]$ is there a constant-free $\I_{\Z}$-certificate of subexponential size that $aG(\vec{x})$ is in the \emph{radical} of the ideal $\langle F_1, \dotsc, F_m \rangle$ for some integer $a$?
\end{open}

\section{Simulations} \label{sec:simulations}
In this section we start with a result we haven't yet mentioned relating general \I to Hilbert-like \I, and then complete the proofs of any remaining simulation results that we've stated previously. Namely, we relate Pitassi's previous algebraic systems \cite{pitassi96, pitassiICM} and number-of-lines in Polynomial Calculus proofs with subsystems of \I; we show that $\I_{\F_p}$ p-simulates $\cc{AC}^0[p]$-Frege in a depth-preserving way; and we show that over fields of characteristic zero, $\I$-proofs of polynomial size \emph{with arbitrary constants} can be simulated in $\cc{coAM}$, assuming the Generalized Riemann Hypothesis.

\subsection{General versus Hilbert-like \texorpdfstring{\I}{\Itext}}
\begin{proposition} \label{prop:gen2Hilb}
Let $F_1 = \dotsb = F_m = 0$ be a polynomial system of equations in $n$ variables $x_1, \dotsc, x_n$ and let $C(\vec{x}, \vec{\f})$ be an \I-certificate of the unsatisfiability of this system. Let $D = \max_{i} \deg_{\f_i} C$ and let $t$ be the number of terms of $C$, when viewed as a polynomial in the $\f_i$ with coefficients in $\F[\vec{x}]$. Suppose $C$ and each $F_i$ can be computed by a circuit of size $\leq s$.

Then a Hilbert-like \I-certificate for this system can be computed by a circuit of size $poly(D,t,n,s)$.\footnote{If the base field $\F$ has size less than $T = Dt\binom{n}{2}$, and the original circuit had multiplication gates of fan-in bounded by $k$, then the size of the resulting Hilbert-like certificate should be multiplied by $(\log T)^{k}$.}
\end{proposition}

The proof uses known sparse multivariate polynomial interpolation algorithms. The threshold $T$ is essentially the number of points at which the polynomial must be evaluated in the course of the interpolation algorithm. Here we use one of the early, elegant interpolation algorithms due to Zippel \cite{zippel}. Although Zippel's algorithm chooses random points at which to evaluate polynomials for the interpolation, in our nonuniform setting it suffices merely for points with the required properties to exist (which they do as long as $|\F| \geq T$). Better bounds may be achievable using more recent interpolation algorithms such as those of Ben-Or and Tiwari \cite{benOrTiwari} or Kaltofen and Yagati \cite{kaltofenYagati}. We note that all of these interpolation algorithms only give us limited control on the \emph{depth} of the resulting Hilbert-like \I-certificate (as a function of the depth of the original \I-certificate $f$), because they all involve solving linear systems of equations, which is not known to be computable efficiently in constant depth.

\begin{proof}
Using a sparse multivariate interpolation algorithm such as Zippel's \cite{zippel}, for each monomial in the placeholder variables $\vec{\f}$ that appears in $C$, there is a polynomial-size algebraic circuit for its coefficient, which is an element of $\F[\vec{x}]$. For each such monomial $\vec{\f}^{\vec{e}} = \f_1^{e_1} \dotsb \f_m^{e_m}$, with coefficient $c_{\vec{e}}(\vec{x})$, there is a small circuit $C'$ computing $c_{\vec{e}}(\vec{x}) \vec{\f}^{\vec{e}}$. Since every $\vec{\f}$-monomial appearing in $C$ is non-constant, at least one of the exponents $e_i > 0$. Let $i_0$ be the least index of such an exponent. Then we get a small circuit computing $c(\vec{e})(\vec{x}) \f_{i_0} F_{i_0}(\vec{x})^{e_{i_0}-1} F_{i_0 + 1}(\vec{x})^{e_{i_0 + 1}} \dotsb F_{m}(\vec{x})^{e_m}$ as follows. Divide $C'$ by $\f_{i_0}$, and then eliminate this division using Strassen \cite{strassenDivision} (or alternatively consider $\frac{1}{e_{i_0}} \frac{\partial C'}{\partial \f_{i_0}}$ using Baur--Strassen \cite{baurStrassen}). In the resulting circuit, replace each input $\f_i$ by a small circuit computing $F_i(\vec{x})$. Then multiply the resulting circuit by $\f_{i_0}$. Repeat this procedure for each monomial appearing (the list of monomials appearing in $C$ is one of the outputs of the sparse multivariate interpolation algorithm), and then add them all together.
\end{proof}

\subsection{Number of lines in Polynomial Calculus is equivalent to determinantal \texorpdfstring{\I}{\Itext}}
\label{sec:pitassi}
We begin by recalling Pitassi's 1996 and 1997 algebraic proof systems \cite{pitassi96, pitassiICM}. In the 1996 system, a proof of the unsatisfiability of $F_1(\vec{x}) = \dotsb = F_m(\vec{x}) = 0$ is a circuit computing a vector $(G_1(\vec{x}), \dotsc, G_m(\vec{x}))$ such that $\sum_i F_i(\vec{x}) G_i(\vec{x}) = 1$. Size is measured by the size of the corresponding circuit.

In the 1997 system, a proof is a rule-based derivation of $1$ starting from the $F_i$. Recall that rule-based algebraic derivations have the following two rules: 1) from $G$ and $H$, derive $\alpha G + \beta H$ for any fields elements $\alpha, \beta \in \F$, and 2) from $G$, derive $Gx_i$ for any variable $x_i$. This is essentially the same as the Polynomial Calculus, but with size measured by the number of lines, rather than by the total number of monomials appearing.

\begin{proposition} \label{prop:pitassi}
Pitassi's 1996 algebraic proof system \cite{pitassi96} is p-equivalent to Hilbert-like \I.

Pitassi's 1997 algebraic proof system \cite{pitassiICM}---equivalent to the number-of-lines measure on PC proofs---is p-equivalent to Hilbert-like $\det$-\I or $\cc{VP}_{ws}$-\I.
\end{proposition}

\begin{proof}
Let $C$ be a proof in the 1996 system \cite{pitassi96}, namely a circuit computing $(G_1(\vec{x}), \dotsc, G_m(\vec{x}))$. Then with $m$ product gates and a single fan-in-$m$ addition gate, we get a circuit $C'$ computing the Hilbert-like \I certificate $\sum_{i=1}^{m} \f_i G_i(\vec{x})$.

Conversely, if $C'$ is a Hilbert-like \I-proof computing the certificate $\sum_i \f_i G_i'(\vec{x})$, then by Baur--Strassen \cite{baurStrassen} there is a circuit $C$ of size at most $O(|C'|)$ computing the vector $(\frac{\partial C'}{\f_1}, \dotsc, \frac{C'}{\f_m}) = (G_1'(\vec{x}), \dotsc, G_m'(\vec{x}))$, which is exactly a proof in the 1996 system. (Alternatively, more simply, but at slightly more cost, we may create $m$ copies of $C'$, and in the $i$-th copy of $C'$ plug in $1$ for one of the $\f_i$ and $0$ for all of the others.

The proof of the second statement takes a bit more work. At this point the reader may wish to recall the definition of weakly skew circuit from Appendix~\ref{app:background:complexity}. 

Suppose we have a derivation of $1$ from $F_1(\vec{x}), \dotsc, F_m(\vec{x})$ in the 1997 system \cite{pitassiICM}. First, replace each $F_i(\vec{x})$ at the beginning of the derivation with the corresponding placeholder variable $\f_i$. Since size in the 1997 system is measured by number of lines in the proof, this has not changed the size. Furthermore, the final step no longer derives $1$, but rather derives an \I certificate. By structural induction on the two possible rules, one easily sees that this is in fact a Hilbert-like \I-certificate. Convert each linear combination step into a linear combination gate, and each ``multiply by $x_i$'' step into a product gate one of whose inputs is a new leaf with the variable $x_i$. As we create a new leaf for every application of the product rule, these new leaves are clearly cut off from the rest of the circuit by removing their connection to their product gate. As these are the only product gates introduced, we have a weakly-skew circuit computing a Hilbert-like \I certificate.

The converse takes a bit more work, so we first show that a Hilbert-like \emph{formula}-\I proof can be converted at polynomial cost into a proof in the 1997 system \cite{pitassiICM}, and then explain why the same proof works for $\cc{VP}_{ws}$-\I. This proof is based on a folklore result (see the remark after Definition~2.6 in Raz--Tzameret \cite{razTzameret}); we thank Iddo Tzameret for a conversation clarifying it, which led us to realize that the result also applies to weakly skew circuits.

Let $C$ be a formula computing a Hilbert-like \I-certificate $\sum_{i=1}^{m} \f_i G_i(\vec{x})$. Using the trick above of substituting in $\{0,1\}$-values for the $\f_i$ (one $1$ at a time), we find that each $G_i(\vec{x})$ can be computed by a formula $\Gamma_i$ no larger than $|C|$. For each $i$ we show how to derive $F_i(\vec{x}) G_i(\vec{x})$ in the 1997 system. These can then be combined using the linear combination rule. Thus for simplicity we drop the subscript $i$ and refer to $\f$, $F(\vec{x})$, $G(\vec{x})$, and the formula $\Gamma$ computing $G$. Without loss of generality (with a polynomial blow-up if needed) we can assume that all of $\Gamma$'s gates have fan-in at most $2$. 

We proceed by induction on the size of the formula $\Gamma$. Our inductive hypothesis is: for all formulas $\Gamma'$ of size $|\Gamma'| < |\Gamma|$, for all polynomials $P(\vec{x})$, in the 1997 system one can derive $P(\vec{x}) \Gamma'(\vec{x})$ starting from $P(\vec{x})$, using at most $|\Gamma'|$ lines. The base case is $|\Gamma|=1$, in which case $G(\vec{x})$ is a single variable $x_i$, and from $P(\vec{x})$ we can compute $P(\vec{x}) x_i$ in a single step using the variable-product rule. 

If $\Gamma$ has a linear combination gate at the top, say $\Gamma = \alpha \Gamma_1 + \beta \Gamma_2$. By induction, from $P(\vec{x})$ we can derive $P(\vec{x}) \Gamma_i(\vec{x})$ in $|\Gamma_i|$ steps for $i=1,2$. Do those two derivations, then apply the linear combination rule to derive $\alpha P(\vec{x}) \Gamma_1(\vec{x}) + \beta P(\vec{x}) \Gamma_2(\vec{x}) = P(\vec{x}) \Gamma(\vec{x})$ in one additional step. The total length of this derivation is then $|\Gamma_1| + |\Gamma_2| + 1 = |\Gamma|$.

If $\Gamma$ has a product gate at the top, say $\Gamma = \Gamma_1 \times \Gamma_2$. Unlike the case of linear combinations where we proceeded in parallel, here we proceed sequentially and use more of the strength of our inductive assumption. Starting from $P(\vec{x})$, we derive $P(\vec{x}) \Gamma_1(\vec{x})$ in $|\Gamma_1|$ steps. Now, starting from $P'(\vec{x}) = P(\vec{x}) \Gamma_1(\vec{x})$, we derive $P'(\vec{x}) \Gamma_2(\vec{x})$ in $|\Gamma_2|$ steps. But $P' \Gamma_2 = P \Gamma_1 \Gamma_2 = P \Gamma$, which we derived in $|\Gamma_1| + |\Gamma_2| \leq |\Gamma|$ steps. This completes the proof of this direction for Hilbert-like \emph{formula}-\I. 

For Hilbert-like weakly-skew \I the proof is similar. However, because gates can now be reused, we must also allow lines in our constructed proof to be reused (otherwise we'd be effectively unrolling our weakly skew circuit into a formula, for which the best known upper bound is only quasi-polynomial). We still induct on the size of the weakly-skew circuit, but now we allow circuits with multiple outputs. We change the induction hypothesis to: for all weakly skew circuits $\Gamma'$ of size $|\Gamma'| < |\Gamma|$, possibly with multiple outputs that we denote $\Gamma'_{out,1}, \dotsc, \Gamma'_{out,s}$, from any $P(\vec{x})$ one can derive the tuple $P \Gamma'_{out,1}, \dotsc, P \Gamma'_{out,s}$ in the 1997 system using at most $|\Gamma'|$ lines.

To simplify matters, we assume that every multiplication gate in a weakly skew circuit has a label indicating which one of its children is separated from the rest of the circuit by this gate.

The base case is the same as before, since a circuit of size one can only have one output, a single variable.

Linear combinations are similar to before, except now we have a multi-output weakly skew circuit of some size, say $s$, that outputs $\Gamma_1$ and $\Gamma_2$. By the induction hypothesis, there is a derivation of size $\leq s$ that derives both $P \Gamma_1$ and $P \Gamma_2$. Then we apply one additional linear combination rule, as before.

For a product gate $\Gamma = \Gamma_1 \times \Gamma_2$, suppose without loss of generality that $\Gamma_2$ is the child that is isolated from the larger circuit by this product gate (recall that we've assumed $\Gamma$ comes with an indicator of which child this is). Then we proceed as before, first computing $P \Gamma_1$ from $P$, and then $(P \Gamma_1) \Gamma_2$ from $(P \Gamma_1)$. Because we apply ``multiplication by $\Gamma_1$'' and ``multiplication by $\Gamma_2$'' in sequence, it is crucial that the gates computing $\Gamma_2$ don't depend on those computing $\Gamma_1$, for the gates $g$ in $\Gamma_1$ get translated into lines computing $P g$, and if we reused \emph{that} in computing $\Gamma_2$, rather than getting $g$ as needed, we would be getting $Pg$.
\end{proof}

It is interesting to note that the condition of being weakly skew is precisely the condition we needed to make this proof go through.

\subsection{Depth-preserving simulation of Frege systems by the Ideal Proof System}
\newcommand{\depth}{\text{depth}}

\begin{theorem} \label{thm:depth}
For any $d(n)$, depth-$(d+2)$ $\I_{\F_p}$ p-simulates depth-$d$ Frege proofs with unbounded fan-in $\lor, \neg, MOD_p$ connectives.
\end{theorem}

\begin{proof}
For simplicity we will present the proof for $p=2$.
The generalization to other values of $p$ is straightforward.
We will use a small modification of the formalization of $\cc{AC}^0[2]$-Frege as given by Maciel and Pitassi \cite{macielPitassi}.
The underlying connectives are: negation, unbounded fanin OR, unbounded fanin AND, and unbounded fanin XOR gates.
We will work in a sequent calculus style proof system, where lines are cedents of the form
$\Gamma \rightarrow \Delta$, where both $\Gamma$ and $\Delta$ are $\{\lor, \neg, MOD_p\}$-formulas, where each of $\neg \Gamma_i$ (for $\Gamma_i \in \Gamma$) and $\Delta_i \in \Delta$ has depth at most $d(n)$; the intended meaning is that the conjunction of the formulas in $\Gamma$ implies the disjunction of the formulas in $\Delta$. For notational convenience, we state the rest of the proof only for $\cc{AC}^0[2]$-Frege, but it will be clear that nothing in the proof depends on the depth being constant.
The axioms are as follows.
\begin{enumerate}
\item $A \rightarrow A$
\item (false implies nothing) $\lor () \rightarrow $
\item $\rightarrow \neg \parity ()$ 
\end{enumerate}

The rules of inference are as follows:

\medskip

\begin{tabular}{rccp{0cm}rcc}
Weakening & $\displaystyle \frac{\Gamma \rightarrow \Delta}{\Gamma \rightarrow \Delta,A}$ & $\displaystyle \frac{\Gamma \rightarrow \Delta}{A, \Gamma \rightarrow \Delta}$ & &
\\[0.25in]
Cut & \multicolumn{2}{c}{$\displaystyle \frac{\rightarrow A, \Gamma \qquad \rightarrow \neg A, \Gamma}{\rightarrow \Gamma}$} & & 
Negation & $\displaystyle \frac{\Gamma, A \rightarrow \Delta}{\Gamma \rightarrow \neg A, \Delta}$ & $\displaystyle \frac{\Gamma \rightarrow A, \Delta}{\Gamma, \neg A \rightarrow \Delta}$ \\[0.25in]
Or-Left & \multicolumn{5}{c}{$\displaystyle \frac{A_1,\Gamma \rightarrow \Delta \qquad \lor(A_2,\dotsc,A_n),\Gamma \rightarrow \Delta}{\lor(A_1,\dotsc,A_n),\Gamma \rightarrow \Delta}$} \\[0.25in]
Or-Right & \multicolumn{5}{c}{$\displaystyle \frac{\Gamma \rightarrow A_1,\lor(A_2,\dotsc,A_n),\Delta}{\Gamma \rightarrow \lor(A_1,\dotsc,A_n),\Delta}$} \\[0.25in]
Parity-Left & \multicolumn{5}{c}{$\displaystyle \frac{A_1, \neg \parity(A_2,\dotsc,A_n),\Gamma \rightarrow \Delta \qquad \parity(A_2,\dotsc,A_n),\Gamma \rightarrow A_1, \Delta}{\parity(A_1,\dotsc,A_n),\Gamma \rightarrow \Delta}$} \\[0.25in]
Parity-Right & \multicolumn{5}{c}{$\displaystyle \frac{A_1,\Gamma \rightarrow \neg \parity(A_2,\dotsc,A_n),\Delta \qquad \Gamma \rightarrow A_1, \parity(A_2,\dotsc,A_n),\Delta}{\Gamma \rightarrow \parity(A_1,\dotsc,A_n),\Delta}$} \\[0.25in]
\end{tabular}

A refutation of a 3CNF formula $\varphi=\kappa_1 \land \kappa_2 \land \dotsb \land \kappa_m$ in $\cc{AC}^0[2]$-Frege
is a sequence of cedents, where each cedent is either one of the $\kappa_i$'s, or
an instance of an axiom scheme, or follows from two earlier cedents by
one of the above inference rules, and the final cedent is the empty cedent.
It is well known that any Frege refutation can be efficiently converted into a 
tree-like proof.\footnote{By tree-like, we mean that the underlying directed acyclic graph structure
of the proof is a tree, and therefore every cedent, other than the final empty cedent, in the refutation is used
exactly once.}

We define a translation $t(A)$ from Boolean formulas to algebraic
circuits (over $\F_2$) such that
for any assignment $\alpha$, $A(\alpha)=1$ if and only if
$t(A)(\alpha)=0$.
The translation is defined inductively as follows.
\begin{enumerate}
\item $t(x)=1-x$ for $x$ atomic (a Boolean variable).
\item $t(\neg A) = 1-t(A)$
\item $t(\lor(A_1,\dotsc,A_n))=t(A_1)t(A_2)\dotsb t(A_n)$
\item $t(\parity(A_1,\dotsc,A_n)) = n-t(A_1)-t(A_2)\dotsb - t(A_n)$ (recall $n$ will be interpreted mod $2$).
\end{enumerate}

Note that the depth of $t(A)$ as an algebraic formula is at most the depth of $A$ as a Boolean formula.

For a cedent $\Gamma \rightarrow \Delta$, we will translate the cedent by moving everything to the right of the arrow. That is, the cedent $L = A_1, \dotsc,A_n \rightarrow B_1,\dotsc,B_m$ will be translated to $t(L) = t(\neg A_1 \lor \dotsb \lor \neg A_n \lor B_1 \lor \dotsb \lor B_m) = (1-t(A_1))(1-t(A_2))\dotsb(1-t(A_n))t(B_1)\dotsb t(B_n)$.

Let $R$ be a tree-like $\cc{AC}^0[2]$-Frege refutation of $\varphi$. We will prove by induction on the number of cedents of $R$ that for each cedent $L$ in the refutation, we can derive $t(L)$ via a Hilbert-like IPS proof (see Definition~\ref{def:IPSideal}) of the form $\sum_i G_i \f_i$, where the $\f_i$'s are the placeholder variables for the initial polynomials (the sum may contain each $\f_i$ more than once), each $G_i$ is a depth $d$ formula, and the overall size is polynomial in the size of the original $\cc{AC}^0[2]$-Frege refutation. (NB: as will become clear below, in order to preserve the depth, we wait to gather like terms in the sum until the end of the proof.) The placeholder variables $\f_1, \dotsc, \f_m$ correspond to $t(\kappa_1), \dotsc, t(\kappa_m)$, and $\f_{m+1}, \dotsc, \f_{m+n}$ correspond to the Boolean axioms $x_1^2 - x_1, \dotsc, x_n^2 - x_n$. 

For the base case, each initial cedent of the form $\rightarrow \kappa_i$ translates to $\f_i$, and thus has the right form.

The axiom $A \rightarrow A$ translates to $t(A)(1-t(A))$. A simple induction on the structure of $A$ shows that $t(A)(1-t(A))$ can be derived from the $x_i^2 - x_i$ by an \I derivation of depth at most the depth of $A$.
The other axioms translate to the identically zero polynomial, so again have the right form. 

For the inductive step, it is a matter of going through all of the rules. We assume inductively that we have a list $L$ of circuits each of the form $G_i \f_i$, such that each $G_i$ has a product gate at its output, and $\sum_{L} G_i \f_i$ is a derivation of the antecedents of the rule (note that, as $L$ is a list, each $\f_i$ may appear more than once in this sum).

\begin{enumerate}
\item (Weakening) Assume $\sum G_i \f_i$ is a derivation of $t(\Gamma \rightarrow \Delta)$.
We want to obtain a derivation of $t(\Gamma \rightarrow \Delta, A)$.
Since we move everything to the right when we translate, this is
equivalent to showing that if $\sum G_i \f_i$ is a derivation of 
$t(\rightarrow A_1,\dotsc,A_n) = t(A_1)t(A_2)\dotsb t(A_n)$,
that we can obtain a derivation of
$t(\rightarrow A_1,\dotsc,A_n,B) = t(A_1)t(A_2)\dotsb t(A_n)t(B)$.
Multiplying each $G_i \f_i$ by $t(B)$ achieves this. The resulting derivation 
is equivalent to $\sum G_i' \f_i$, where the depth of $G_i'$ is $\max\{\depth(G_i), \depth(B)\}$ (we do not need to add $1$ to the depth because we've assumed that $G_i$ has a product gate at the top). 
\item (Cut) We want to show that if $\sum G_i \f_i$ is a derivation
of $t(\rightarrow \neg A, B_1,\dotsc, B_n) = (1-t(A))t(B_1)\dotsb t(B_n)$ and 
$\sum G_i' \f_i$ is a derivation of
$t(\rightarrow A, B_1,\dotsc B_n) = t(A)t(B_1)\dotsb t(B_n)$, that
we can derive $t(\rightarrow B_1 \dotsc B_n) = t(B_1)\dotsb t(B_n)$. Semantically, adding these two derivations gives what we want. In order to preserve the inductive assumption, we do \emph{not} gather terms, but rather concatenate the two lists $(G_i \f_i)$ and $(G_i'\f_i)$, so that each term still has a product gate at the top without increasing the depth.
\item (Negation) Because our translation moves everything to the right, the translated versions become syntactically identical, and there is nothing to do for the negation rules.
\item (Or-Left) We want to show that if $\sum G_i \f_i$ is a derivation of
$t( \rightarrow \neg A_1, \Delta)$, and $\sum G_i' \f_i$ is a derivation
of $t( \rightarrow \neg \lor(A_2,\dotsc,A_n), \Delta)$, then we
can derive $t(\rightarrow \neg \lor(A_1,\dotsc,A_n), \Delta)$.
We have $$\sum G_i F_i = t(\rightarrow \neg A_1, \Delta) = (1-t(A_1))t(\Delta),$$
$$\sum G_i' F_i = t(\rightarrow \neg \lor(A_2,\dotsc,A_n), \Delta) = (1-t(A_2)t(A_3)\dotsb t(A_n))t(\Delta).$$
Multiplying the second by $t(A_1)$ and ``adding'' to the first gives the desired derivation. Again, when we ``add'' we do not gather terms, but rather just concatenate lists, so that each $G_i$ has a product gate at the top. 
\item (Or-Right) The translation of the derived formula is syntactically identical to the original formula, so there is nothing to do.

\item (Parity-Left) We want to show that if $\sum G_i \f_i$ is a derivation
of $t(\rightarrow \neg A_1, \parity(A_2,\dotsc,A_n), \Delta)$ and $\sum G_i' \f_i$ is
a derivation of
$t( \rightarrow A_1, \neg \parity(A_2,\dotsc,A_n), \Delta)$, then we can derive 
$t(\rightarrow \neg \parity(A_1,\dotsc,A_n), \Delta)$.
We have
$$t(\rightarrow \neg A_1, \parity(A_2,\dotsc,A_n), \Delta) = (1-t(A_1))(n-1-t(A_2)-t(A_3)- \dotsb - t(A_n))t(\Delta),$$
$$t(\rightarrow A_1, \neg \parity(A_2, \dotsc, A_n), \Delta)=t(A_1)(1 - (n-1-t(A_2)-t(A_3)- \dotsb - t(A_n)))t(\Delta).$$
It is easily verified that subtracting the latter from the former yields $t(\rightarrow \neg \parity(A_1,\dotsc,A_n),\Delta)$. To perform ``subtraction'' while maintaining a product gate at the top, we multiply the latter by $-1$ and then concatenate the two lists.
\item (Parity-Right) This case is similar to Parity-left.
\end{enumerate}

In all cases, we can derive the bottom cedent as $\sum_i G_i \f_i$, where
each $G_i$ has constant depth (in fact, depth at most one greater than the depth
of the original proof), and the overall size is polynomial in the original
proof size. Since we've actually just been maintaing a list of terms $G_i \f_i$ in which the $\f_i$ may appear multiple times, the final step is to add these all together and gather terms, leading to our final derivation
of polynomial size, and depth at most $d+2$, where $d$ was the original depth.
\end{proof}

\subsection{Simulating \texorpdfstring{\I}{\Itext}-proofs with arbitrary constants in \texorpdfstring{$\cc{coAM}$}{coAM}}
The following proposition shows how we may conclude that $\cc{NP} \subseteq \cc{coAM}$ from the assumption of polynomial-size \I proofs for all tautologies, \emph{without} assuming the \I proofs are constant-free (but using the Generalized Riemann Hypothesis). We thank Pascal Koiran for the second half of the proof.

\begin{proposition} \label{prop:coMAGRH}
Assuming the Generalized Riemann Hypothesis, over any field $\F$ of characteristic zero, if every propositional tautology has a polynomial-size $\I_{\F}$-proof of polynomial degree, then $\cc{NP} \subseteq \cc{coAM}$.
\end{proposition}

We do not know how to improve this result from $\cc{coAM}$ to $\cc{coMA}$ (as in Proposition~\ref{prop:coMA}).

\begin{proof}[Proof (with P. Koiran)]
We reduce to the fact that deciding Hilbert's Nullstellensatz---that is, given a system of integer polynomials over $\Z$, deciding if they have a solution over $\C$---is in $\cc{AM}$ \cite{koiranNS}. Rather than looking at solvability of the original set of equations $F_1(\vec{x}) = \dotsb = F_m(\vec{x}) = 0$, we consider solvability of a set of equations whose solutions describe all of the polynomial-size $\I$-certficiates for $F$. Namely, consider a \emph{generic} polynomial-size circuit, meaning a layered circuit of $\poly(n)$ depth and $\poly(n)$ width, with $n$ inputs $x_1, \dotsc, x_n,\f_1,\dotsc,\f_m$, and alternating layers of linear combination and product gates, where every edge $e$ terminating at any linear combination gate gets its own independent variable $z_{e}$. The output gate of this generic circuit computes a polynomial $C(\vec{x}, \vec{\f}, \vec{z})$, and for any setting of the $z_{e}$ variables to constants $\zeta_{e}$, we get a particular polynomial-size circuit computing a polynomial $C_{\vec{\zeta}}(\vec{x}, \vec{\f}) := C(\vec{x}, \vec{\f}, \vec{\zeta})$. Furthermore, any function computed by a polynomial-size circuit is equal to $C_{\vec{\zeta}}(\vec{x},\vec{\f})$ for some setting of $\vec{\zeta}$. In particular, if there is a polynomial size \I proof $C'$ for $F$, then there is some $\vec{\zeta} \in \F^{n}$ such that $C' = C_{\vec{\zeta}}(\vec{x}, \vec{\f})$.

We will translate the conditions that a circuit be an \I certificate into \emph{equations} on the new $z$ variables. Pick sufficiently many random values $\vec{\xi}^{(1)}, \vec{\xi}^{(2)}, \dotsc, \vec{\xi}^{(h)}$ to be substituted into $\vec{x}$; think of the $\vec{\xi}^{(i)}$ as a hitting set for the $x$-variables. Then we consider the solvability of the following set of $2h$ equations in $\vec{z}$:
\begin{eqnarray*}
\text{(For $i=1,\dotsc,h$)} & & C(\vec{\xi}^{(i)}, \vec{0}, \vec{z}) = 0 \\
\text{(For $i=1,\dotsc,h$)} & & C(\vec{\xi}^{(i)}, \vec{F}(\vec{\xi}^{(i)}), \vec{z}) = 1 \\
\end{eqnarray*}
Determining whether a system of polynomial equations, given by circuits over a field $\F$ of characteristic zero, has a solution in the algebraic closure $\overline{\F}$ can be done in $\cc{AM}$ \cite{koiranNS}. If $\vec{\zeta}$ is such that $C_{\vec{\zeta}}(\vec{x}, \vec{\f})=C(\vec{x}, \vec{\f}, \vec{\zeta})$ is in fact an \I certificate, then the preceding equalities will be satisfied regardless of the choices of the $\vec{\xi}^{(i)}$. Otherwise, at least one monomial in $C(\vec{x}, 0, \vec{\zeta})$ or $C(\vec{x}, \vec{F}(\vec{x}),\vec{\zeta})-1$ will be nonzero. Since all the monomials have polynomial degree, the usual DeMillo--Lipton--Schwarz--Zippel lemma implies that with high probability, a random point $\vec{\xi}$ will make any such nonzero monomial evaluate to a nonzero value. Choosing polynomially many points thus suffices. Composing Koiran's $\cc{AM}$ algorithm for the Nullstellensatz with the random guesses for the $\vec{\xi}^{(i)}$, and assuming that every family of propositional tautologies has $\cc{VP}$-\I certificates, we get an $\cc{AM}$ algorithm for TAUT. 
\end{proof}

\section{Lower bounds on \texorpdfstring{\I}{\Itext} imply circuit lower bounds} \label{sec:VNP}
Here we complete the proof of the following theorem:

\begin{theorem} \label{thm:VNP}
A super-polynomial lower bound on [constant-free] Hilbert-like $\I_{R}$ proofs of any family of tautologies implies $\cc{VNP}_{R} \neq \cc{VP}_{R}$ [respectively, $\cc{VNP}^0_{R} \neq \cc{VP}^0_{R}$], for any ring $R$.

A super-polynomial lower bound on the number of lines in Polynomial Calculus proofs implies the Permanent versus Determinant Conjecture ($\cc{VNP} \neq \cc{VP}_{ws}$).
\end{theorem}

In Section~\ref{sec:VNPeabs} we proved this theorem assuming the following key lemma, which we now prove in full.

\begin{lemma} \label{lem:VNP}
Every family of CNF tautologies $(\varphi_n)$ has a Hilbert-like family of \I certificates $(C_n)$ in $\cc{VNP}^{0}_{R}$.
\end{lemma}

\begin{proof}
We mimic one of the proofs of completeness for Hilbert-like \I \cite[Theorem~1]{pitassi96} (recall Proposition~\ref{prop:pitassi}), and then show that this proof can in fact be carried out in $\cc{VNP}^{0}$. We omit any mention of the ground ring, as it will not be relevant.

Let $\varphi_n(\vec{x}) = \kappa_1(\vec{x}) \wedge \dotsb \wedge \kappa_m(\vec{x})$ be an unsatisfiable CNF, where each $\kappa_i$ is a disjunction of literals. Let $C_i(\vec{x})$ denote the (negated) polynomial translation of $\kappa_i$ via $\neg x \mapsto x$, $x \mapsto 1-x$ and $f \vee g \mapsto fg$; in particular, $C_i(\vec{x}) = 0$ if and only if $\kappa_i(\vec{x}) = 1$, and thus $\varphi_n$ is unsatisfiable if and only if the system of equations $C_1(\vec{x})=\dotsb=C_m(\vec{x})=x_1^2 - x_1 = \dotsb = x_n^2 - x_n = 0$ is unsatisfiable. In fact, as we'll see in the course of the proof, we won't need the equations $x_i^2 - x_i = 0$. It will be convenient to introduce the function $b(e,x)=ex + (1-e)(1-x)$, \ie, $b(1,x) = x$ and $b(0,x)=1-x$. For example, the clause $\kappa_i(\vec{x}) = (x_1 \vee \neg x_{17} \vee x_{42})$ gets translated into $C_i(\vec{x}) = (1-x_1)x_{17}(1-x_{42}) = b(0,x_1)b(1,x_{17})b(0,x_{42})$, and therefore an assignment falsifies $\kappa_i$ if and only if $(x_1,x_{17},x_{42}) \mapsto (0,1,0)$.

Just as $1 = x_1 x_2 + x_1(1-x_2) + (1-x_2)x_1 + (1-x_2)(1-x_1)$, an easy induction shows that
\begin{equation} \label{eqn:1}
1 = \sum_{\vec{e} \in \{0,1\}^{n}} \prod_{i=1}^{n}b(e_i,x_i).
\end{equation}
We will show how to turn this expression---which is already syntactically in $\cc{VNP}^{0}$ form---into a $\cc{VNP}$ certificate refuting $\varphi_n$. Let $c_i$ be the placeholder variable corresponding to $C_i(\vec{x})$.

The idea is to partition the assignments $\{0,1\}^{n}$ into $m$ parts $A_1,\dotsc,A_m$, where all assignments in the $i$-th part $A_i$ falsify clause $i$. This will then allow us to rewrite equation (\ref{eqn:1}) as
\begin{equation} \label{eqn:rewrite}
1 = \sum_{i=1}^{m} C_i(\vec{x})\left(\sum_{\vec{e} \in A_i} \prod_{j : x_j \notin \kappa_i} b(e_j,x_j)\right),
\end{equation} 
where ``$x_j \notin \kappa_i$'' means that neither $x_j$ nor its negation appears in $\kappa_i$. Equation (\ref{eqn:rewrite}) then becomes the \I-certificate $\sum_{i=1}^{m} c_i \cdot \left(\sum_{\vec{e} \in A_i} \prod_{j : x_j \notin \kappa_i} b(e_j,x_j)\right)$. What remains is to show that the sum can indeed be rewritten this way, and that there is some partition $(A_1,\dotsc, A_m)$ as above such that the resulting certificate is in fact in $\cc{VNP}$.

First, let us see why such a partition allows us to rewrite (\ref{eqn:1}) as (\ref{eqn:rewrite}). The key fact here is that the clause polynomial $C_i(\vec{x})$ divides the term $t_{\vec{e}}(\vec{x}) := \prod_{i=1}^{n} b(e_i, x_i)$ if and only if $C_i(\vec{e}) = 1$, if and only if $\vec{e}$ \emph{falsifies} $\kappa_i$. Let $C_i(\vec{x})=\prod_{i \in I} b(f_i,x_i)$, where $I \subseteq [n]$ is the set of indices of the variables appearing in clause $i$. By the properties of $b$ discussed above, $1=C_i(\vec{e})=\prod_{i \in I} b(f_i, e_i)$ if and only if $b(f_i,e_i)=1$ for all $i \in I$, if and only if $f_i=e_i$ for all $i \in I$. In other words, if $1=C_i(\vec{e})$ then $C_i = \prod_{i \in I} b(e_i, x_i)$, which clearly divides $t_{\vec{e}}$. Conversely, suppose $C_i(\vec{x})$ divides $t_{\vec{e}}(\vec{x})$. Since $t_{\vec{e}}(\vec{e})=1$ and every factor of $t_{\vec{e}}$ only takes on Boolean values on Boolean inputs, it follows that every factor of $t_{\vec{e}}$ evaluates to $1$ at $\vec{e}$, in particular $C_i(\vec{e})=1$.

Let $A_1, \dotsc, A_m$ be a partition of $\{0,1\}^n$ such that every assignment in $A_i$ falsifies $\kappa_i$. Since $C_i$ divides every term $t_{\vec{e}}$ such that $\vec{e}$ falsifies clause $i$, $C_i$ divides every term $t_{\vec{e}}$ with $\vec{e} \in A_i$, and thus we can indeed rewrite (\ref{eqn:1}) as (\ref{eqn:rewrite}). 

Next, we show how to construct a partition $A_1, \dotsc, A_m$ as above so that the resulting certificate is in $\cc{VNP}$. The partition we'll use is a greedy one. $A_1$ will consist of \emph{all} assignments that falsify $\kappa_1$. $A_2$ will consist of all \emph{remaining} assignments that falsify $\kappa_2$. And so on. In particular, $A_i$ consists of all assignments that falsify $\kappa_i$ and \emph{satisfy} all $A_j$ with $j < i$. (If at some clause $\kappa_i$ before we reach the end, we have used up all the assignments---which happens if and only if the first $i$ clauses on their own are unsatisfiable---that's okay: nothing we've done so far nor anything we do below assumes that all $A_i$ are nonempty.) 

Equivalently, $A_i = \{\vec{e} \in \{0,1\}^n | C_i(\vec{e})=1 \text{ and } C_j(\vec{e})=0 \text{ for all } j < i\}$. For any property $\Pi$, we write $\llbracket \Pi(\vec{e}) \rrbracket$ for the indicator function of $\Pi$: $\llbracket \Pi(\vec{e}) \rrbracket=1$ if and only if $\Pi(\vec{e})$ holds, and $0$ otherwise. We thus get the certificate:
\begin{eqnarray*}
& & \sum_{i=1}^{m} c_i \cdot \left(\sum_{\vec{e} \in \{0,1\}^n} \llbracket \vec{e} \text{ falsifies } \kappa_i \text{ and satisfies $\kappa_j$ for all } j < i \rrbracket \prod_{j : x_j \notin \kappa_i} b(e_j, x_j) \right) \\
& = & \sum_{i=1}^{m} c_i \cdot \left(\sum_{\vec{e} \in \{0,1\}^n} \llbracket C_i(\vec{e})=1 \text{ and } C_j(\vec{e})=0 \text{ for all } j < i \rrbracket \prod_{j : x_j \notin \kappa_i} b(e_j, x_j) \right) \\
& = & \sum_{i=1}^{m} c_i \cdot \left(\sum_{\vec{e} \in \{0,1\}^n} \left(C_i(\vec{e}) \prod_{j < i}(1-C_j(\vec{e})) \right) \prod_{j : x_j \notin \kappa_i} b(e_j, x_j) \right) \\
& = & \sum_{e \in \{0,1\}^{n}} \sum_{i=1}^{m} c_i C_i(\vec{e})\left(\prod_{j < i}(1-C_j(\vec{e}))\right)\left(\prod_{j : x_j \notin \kappa_i} b(e_j, x_j)\right)
\end{eqnarray*}
Finally, it is readily visible that the polynomial function of $\vec{c}$, $\vec{e}$, and $\vec{x}$ that is the summand of the outermost sum $\sum_{\vec{e} \in \{0,1\}^{n}}$ is computed by a polynomial-size circuit of polynomial degree, and thus the entire certificate is in $\cc{VNP}$. Indeed, the expression as written exhibits it as a small \emph{formula} of constant depth with unbounded fan-in gates. By inspection, this circuit only uses the constants $0,1,-1$, hence the certificate is in $\cc{VNP}^{0}$.
\end{proof}

\section{PIT as a bridge between circuit complexity and proof complexity} \label{sec:PIT}
Having already introduced and discussed our PIT axioms in Section~\ref{sec:PITeabs}, here we complete the proofs of Theorems~\ref{thm:EF} and \ref{thm:AC0}. We maintain the notations and conventions of Section~\ref{sec:PITeabs}.

\subsection{Extended Frege is p-equivalent to \texorpdfstring{\I}{\Itext} if PIT is EF-provably easy} \label{sec:EF}
\begin{theorem} \label{thm:EF}
If there is a family $K$ of polynomial-size Boolean circuits computing PIT, such that the PIT axioms for $K$ have polynomial-size EF proofs, then EF is polynomially equivalent to $\I$.
\end{theorem}

To prove the theorem, we will first show that EF is p-equivalent to $\I$ if a family of propositional formulas expressing soundness of $\I$ are efficiently EF provable. Then we will show that efficient EF proofs of $Soundness_{\I}$ follows from efficient EF proofs for our PIT axioms.

\myparagraph{Soundness of \I}

It is well-known that for standard Cook--Reckhow proof systems, a proof system $P$ can p-simulate another proof system $P'$
if and only if $P$ can prove soundness of $P'$. Our proof system is not standard because verifying a proof requires probabilistic, rather than deterministic, polynomial-time. Still we will show how to formalize soundness of $\I$ propositionally, and we will show that if EF can efficiently prove soundness of $\I$ then EF is p-equivalent to $\I$.

Let $\varphi = \kappa_1 \land \ldots \land \kappa_m$ be an unsatisfiable propositional 3CNF formula over variables $p_1,\ldots,p_n$, and let $Q^\varphi_1, \ldots, Q^\varphi_m$ be the corresponding polynomial equations (each of degree at most 3) such that $\kappa_i(\alpha)=1$ if and only if $Q^\varphi_i(\alpha)=0$ for $\alpha \in \{0,1\}^{n}$. An $\I$-refutation of $\varphi$ is an algebraic circuit, $C$, which demonstrates that $1$ is in the ideal generated by the polynomial equations $\vec{Q}^\varphi$. (This demonstrates that the polynomial equations $\vec{Q}^\varphi=0$ are unsolvable, which is equivalent to proving that $\varphi$ is unsatisfiable.) In particular, recall that $C$ has two types of inputs: $x_1,\dotsc,x_n$ (corresponding to the propositional variables $p_1,\dotsc,p_n$) and the placeholder variables $\f_1, \ldots, \f_m$ (corresponding to the equation $Q^{\varphi}_1,\dotsc,Q^\varphi_m$), and satisfies the following two properties:

\begin{enumerate}
\item $C(\vec{x},\vec{0})=0$. This property essentially states that the polynomial computed by $C(\vec{x},\vec{Q}(\vec{x}))$ is in the ideal generated by $Q^\varphi_1,\ldots,Q^\varphi_m$.

\item $C(\vec{x},\vec{Q}^\varphi(\vec{x}))=1$. This property states that the polynomial computed by $C$, when we substitute the $Q^\varphi_i$'s for the $\f_i$'s, is the identically 1 polynomial.
\end{enumerate}

\myparagraph{Encoding \I Proofs}

Let $K$ be a family of polynomial-size circuits for PIT.
Using $K_{m,n}$, we can create a polynomial-size Boolean circuit, $Proof_\I([C], [\varphi])$ that
is true if and only if $C$ is an $\I$-proof of the unsatisfiability of $\vec{Q}^\varphi=0$.
The polynomial-sized Boolean circuit $Proof_\I([C],[\varphi])$ first takes the encoding of the 
algebraic circuit $C$ (which has
$x$-variables and placeholder variables), and creates the encoding of a new algebraic circuit, $[C']$, where
$C'$ is like $C$ but with each $\f_i$ variable replaced by 0.
Secondly, it takes the encoding of $C$ and $[\varphi]$ and creates the encoding of a new
circuit $C''$, where $C''$ is like $C$ but now with each $\f_i$ variable replaced by $Q^\varphi_i$.
(Note that whereas $C$ has $n+m$ underlying algebraic variables, both $C'$ and $C''$ have only $n$ underlying variables.)
$Proof_\I([C], [\varphi])$ is true if and only if 
$K([C'])$---that is, $C'(\vec{x})=C(\vec{x},\vec{0})$ computes the $0$ polynomial---and
$K([1-C''])=0$---that is, $C''(\vec{x})=C(\vec{x},\vec{Q}^{\varphi}(\vec{x}))$ computes the $1$ polynomial.

\begin{definition}
Let formula $Truth_{bool}(\vec{p},\vec{q})$ state that the truth assignment $\vec{q}$
satisfies the Boolean formula coded by $\vec{p}$.
The soundness of $\I$ says that if
$\varphi$ has a refutation in $\I$, then $\varphi$ is unsatisfiable.
That is, $Soundness_{\I,m,n}([C], [\varphi], \vec{p})$ has variables that encode
a size $m$ \I-proof $C$, variables that encode a 3CNF formula $\varphi$ over $n$ variables, and $n$
additional Boolean variables, $\vec{p}$.
$Soundness_{\I,m,n}([C], [\varphi], \vec{p})$ states:
\[
Proof_\I(\prop{[C]},\prop{[\varphi]}) \rightarrow \neg Truth_{bool}(\prop{[\varphi]}, \vec{p}).
\]
\end{definition}

\begin{lemma} \label{lem:soundness}
If EF can efficiently prove $Soundness_\I$ for some polynomial-size Boolean circuit family $K$ computing PIT, then EF is p-equivalent to $\I$.
\end{lemma}

\begin{proof} 
Because $\I$ can p-simulate EF, it suffices to show that
if EF can prove Soundness of $\I$, then EF can p-simulate $\I$.
Assume that we have a polynomial-size EF proof
of $Soundness_\I$. 
Now suppose that $C$ is an \I-refutation of an unsatisfiable 3CNF formula $\varphi$
on variables $\vec{p}$.
We will show that EF can also prove $\neg \varphi$ with a proof of size polynomial in $|C|$.

First, 
we claim that it follows from a natural encoding (see Section~\ref{sec:encoding}) that EF can efficiently prove: 
\[
\varphi \rightarrow Truth_{bool}([\varphi],\vec{p}).
\]
(Variables of this statement just the $p$ variables, because $\varphi$ is a fixed 3CNF formula, so the encoding $[\varphi]$ is a variable-free Boolean string.)

Second, 
if $C$ is an $\I$-refutation of $\varphi$, then EF can prove $Proof_\I([C],[\varphi])$.\footnote{The fact that $Proof_\I([C],[\varphi])$ is even true, given that $C$ is an \I-refutation of $\varphi$, follows from the completeness of the circuit $K$ computing PIT---that is, if $C \equiv 0$, then $K([C])$ accepts. This is one of only two places in the proof of Theorem~\ref{thm:EF} that we actually need the assumption that $K$ correctly computes PIT, rather than merely assuming that $K$ satisfies our PIT axioms. However, it is clear that this usage of this assumption is crucial. The other usage is in Step~\ref{step:axioms:1} of Lemma~\ref{lem:axioms}.} This holds because both $C$ and $\varphi$ are fixed, so this formula is variable-free. Thus, EF can just verify that it is true.

Third,
by soundness of $\I$, which we are assuming is EF-provable, and the
fact that EF can prove $Proof_\I([C],[\varphi])$ (step 2),
it follows by modus ponens that EF can prove
$\neg Truth_{bool}([\varphi], \vec{p})$. 
(The statement $Soundness_\I([C],[\varphi],\vec{p})$ for this instance will only involve variables $\vec{p}$: the other two sets of inputs to the $Soundness_\I$ statement, $[C]$ and $[\varphi]$, are constants here since both $C$ and $\varphi$ are fixed.)

Finally, by modus ponens and the contrapositive of $\varphi \rightarrow Truth_{bool}([\varphi], \vec{p})$, we conclude in EF $\neg \varphi$, as desired.
\end{proof}

Theorem~\ref{thm:EF} follows from the preceding lemma and the next one.

\begin{lemma} \label{lem:axioms}
If EF can efficiently prove the PIT axioms for some polynomial-size Boolean circuit family $K$ computing PIT, 
then EF can efficiently prove $Soundness_{\I}$ (for that same $K$).
\end{lemma}

\begin{proof}
Starting with $Truth_{bool}(\prop{[\varphi]},\vec{p})$, $K(\prop{[C(\vec{x},\vec{0})]})$,
$K(\prop{[1-C(\vec{x},\vec{Q}(\vec{x}))]})$, we will derive a contradiction.
\begin{enumerate}
\item\label{step:axioms:1} First show for every $i \in [m]$,
$Truth_{bool}([\varphi],\vec{p}) \rightarrow K(\prop{[Q_i^\varphi(\vec{p})]})$, where
$Q_i^\varphi$ is the low degree polynomial corresponding to the clause, $\kappa_i$, of $\varphi$. Note that, as $\varphi$ is not a fixed formula but is determined by the propositional variables encoding $\prop{[\varphi]}$, the encoding $\prop{[Q_i^\varphi]}$ depends on a subset of these variables. 

$Truth_{bool}(\prop{[\varphi]},\vec{p})$ states that each clause $\kappa_i$ in $\varphi$ evaluates
to true under $\vec{p}$. It is a tautology that if $\kappa_i$
evaluates to true under $\vec{p}$, then $Q_i^{\varphi}$ evaluates to $0$ at $\vec{p}$. Since $K$ correctly computes PIT, \begin{equation} \label{eqn:truth}
Truth_{bool}(\prop{[\kappa_i]},\vec{p}) \rightarrow K(\prop{[Q_i^\varphi(\vec{p})]})\tag{*}
\end{equation}
is a tautology. Furthermore, although both the encoding $\prop{[\kappa_i]}$ and $\prop{[Q_i^{\varphi}]}$ depend on the propositional variables encoding $\prop{[\varphi]}$, since we assume that $\varphi$ is a 3CNF, these only depend on \emph{constantly many} of the variables encoding $\prop{[\varphi]}$. Thus the tautology (\ref{eqn:truth}) can be proven in EF by brute force. 
Putting these together we can derive $Truth_{bool}(\prop{[\varphi]},\vec{p}) \rightarrow K(\prop{[Q_i^\varphi(\vec{p})]})$,
as desired.

\item\label{step:axioms:2} Using the assumption $Truth_{bool}(\prop{[\varphi]},\vec{p})$ together with (\ref{step:axioms:1}) we
derive $K(\prop{[Q_i^\varphi(\vec{p})]})$ for all $i \in [m]$.

\item\label{step:axioms:3} Using Axiom~\ref{axiom:Boolean} we can prove
$K(\prop{[C(\vec{x},\vec{0})]}) \rightarrow K(\prop{[C(\vec{p},\vec{0})]})$. Using modus ponens with the assumption $K(\prop{[C(\vec{x},\vec{0})]})$, we derive $K(\prop{[C(\vec{p},\vec{0})]})$.

\item\label{step:axioms:4} Repeatedly using Axiom~\ref{axiom:subzero} and Axiom~\ref{axiom:perm} we can prove
\[
K(\prop{[Q_1^\varphi(\vec{p})]}), K(\prop{[Q_2^\varphi(\vec{p})]}), \ldots , K(\prop{[Q_m^\varphi(\vec{p})]}), K(\prop{[C(\vec{p},\vec{0})]}) \rightarrow
K(\prop{[C(\vec{p},\vec{Q}(\vec{p}))]}).
\]

\item\label{step:axioms:5} Applying modus ponens repeatedly with (\ref{step:axioms:4}), (\ref{step:axioms:2}) and (\ref{step:axioms:3}) we can prove
$K(\prop{[C(\vec{p},\vec{Q}(\vec{p}))]})$.

\item\label{step:axioms:6} Applying Axiom~\ref{axiom:one} to (\ref{step:axioms:5}) we get $\neg K(\prop{[ 1 - C(\vec{p},\vec{Q}(\vec{p}))]})$.

\item\label{step:axioms:7} Using Axiom~\ref{axiom:Boolean} we can prove
$K(\prop{[1 - C(\vec{x},\vec{Q}(\vec{x}))]}) \rightarrow K(\prop{[1-C(\vec{p},\vec{Q}(\vec{p}))]})$.
Using our assumption $K(\prop{[1-C(\vec{x},\vec{Q}(\vec{x}))]})$ and modus ponens, we conclude
$K(\prop{[1-C(\vec{p},\vec{Q}(\vec{p}))]})$.
\end{enumerate}
Finally, (\ref{step:axioms:6}) and (\ref{step:axioms:7}) give a contradiction.
\end{proof}

\subsection{Proofs relating \texorpdfstring{$\cc{AC}^0[p]$-Frege}{AC0[p]-Frege} lower bounds, PIT, and circuit lower bounds} 
\label{sec:AC0p}
Having already discussed the corollaries and consequences of Theorem~\ref{thm:AC0}, here we merely complete its proof.

\begin{theorem} \label{thm:AC0}
Let $\mathcal{C}$ be any class of circuits closed under $\cc{AC}^0$ circuit reductions. If there is a family $K$ of polynomial-size Boolean circuits for PIT such that the PIT axioms for $K$ have polynomial-size $\mathcal{C}$-Frege proofs, 
then $\mathcal{C}$-Frege is polynomially equivalent to $\I$, and consequently polynomially equivalent to Extended Frege.
\end{theorem}

\newcommand{\op}[1]{\, op_{#1} \,}
Note that here we \emph{do not} need to restrict the circuit $K$ to be in the class $\mathcal{C}$. This requires one more technical device compared to the proofs in the previous section. The proof of Theorem~\ref{thm:AC0} follows the proof of Theorem~\ref{thm:EF} very closely. The main new ingredient is a folklore technical device that allows even very weak systems such as $\cc{AC}^0$-Frege to make statements about arbitrary circuits $K$, together with a careful analysis of what was needed in the proof of Theorem~\ref{thm:EF}.

\myparagraph{Encoding $K$ into weak proof systems}

Extended Frege can easily reason about arbitrary circuits $K$: for each gate $g$ of $K$ (or even each gate of each instance of $K$ in a statement, if so desired), with children $g_{\ell}, g_{r}$, EF can introduce a new variable $k_g$ together with the requirement that $k_g \leftrightarrow k_{g_{\ell}} \op{g} k_{g_{r}}$, where $\op{g}$ is the corresponding operation $g = g_{\ell} \op{g} g_{r}$ (\eg, $\land$, $\lor$, etc.). But weaker proof systems such as Frege (=$\cc{NC}^1$-Frege), $\cc{AC}^0[p]$-Frege, or $\cc{AC}^0$-Frege do not have this capability. We thus need to help them out by introducing these new variables and formulae ahead of time.

For each gate $g$, the statement $k_g \leftrightarrow k_{g_{\ell}} \op{g} k_{g_{r}}$ only involves 3 variables, and thus can be converted into a 3CNF of constant size. We refer to these clauses as the ``$K$-clauses.'' Note that the $K$-clauses do not set the inputs of $K$ to any particular values nor require its output to be any particular value. We denote the variables corresponding to $K$'s inputs as $k_{in,i}$ and the variable corresponding to $K$'s output as $k_{out}$. 

The modified statement $Proof_\I(\prop{[C]},\prop{[\varphi]})$ now takes the following form. Recall that $Proof_\I$ involves two uses of $K$: $K(\prop{[C(\vec{x},\vec{0})]})$ and $K(\prop{[1-C(\vec{x}, \vec{Q}^\varphi(\vec{x}))]})$. Each of these instances of $K$ needs to get its own set of variables, which we denote $k^{(1)}_{g}$ for gate $g$ in the first instance, and $k^{(2)}_{g}$ for gate $g$ in the second instance, together with their own copies of the $K$-clauses. For an encoding $[C]$ or $[\varphi]$, let $[C]_{i}$ denote it's $i$-th bit, which may be a constant, a propositional variable, or even a propositional formula. Then $Proof_\I(\prop{[C]}, \prop{[\varphi}])$ is
\[
\begin{split}
& \bigwedge_{g} \left(k^{(1)}_{g} \leftrightarrow k^{(1)}_{g_{\ell}} \op{g} k^{(1)}_{g_{r}}\right) 
\land \bigwedge_{i} \left(k^{(1)}_{in,i} \leftrightarrow \prop{[C(\vec{x},\vec{0})]}_{i}\right) \\
\land & \bigwedge_{g} \left(k^{(2)}_{g} \leftrightarrow k^{(2)}_{g_{\ell}} \op{g} k^{(2)}_{g_{r}}\right) 
\land \bigwedge_{i} \left(k^{(2)}_{in,i} \leftrightarrow \prop{[1-C(\vec{x}, \vec{Q}^{\varphi}(\vec{x}))]}_{i}\right) \\
\rightarrow & k^{(1)}_{out} \land k^{(2)}_{out} \\
\end{split}
\]
Throughout, we use the same notation $Proof_\I(\prop{[C]}, \prop{[\varphi}])$ as before to mean this modified statement (we will no longer be referring to the original, EF-style statement). The modified statement $Soundness_\I(\prop{[C]}, \prop{[\varphi]}, \prop{\vec{p}})$ will now take the form
\[
\left( (\text{dummy statements}) \land Proof_\I(\prop{[C]}, \prop{[\varphi}]) \right)\rightarrow \neg Truth_{bool}(\prop{[\varphi]}, \vec{p}),
\]
using the new version of $Proof_\I$. Here ``dummy statements'' refers to certain statements that we will explain in Lemma~\ref{lem:AC0soundness}. These dummy statements will only involve variables that do not appear in the rest of $Soundness_\I$, and therefore will be immediately seen not to affect its truth or provability.

\myparagraph{The proofs}

Lemmata~\ref{lem:AC0soundness} and \ref{lem:AC0axioms} are the $\cc{AC}^0$-analogs of Lemmata~\ref{lem:soundness} and \ref{lem:axioms}, respectively. The proof of Lemma~\ref{lem:AC0soundness} will cause no trouble, and the proof of Lemma~\ref{lem:AC0axioms} will need one additional technical device (the ``dummy statements'' above). 

Before getting to their proofs, we state the main additional lemma that we use to handle the new $K$ variables. We say that a variable $k^{(i)}_{in,j}$ corresponding to an input gate of $K$ is \definedWord{set to $\psi$} by a propositional statement if $k^{(i)}_{in,j} \leftrightarrow \psi$ occurs in the statement.

\begin{lemma} \label{lem:K}
Let $(\varphi_n)$ be a sequence of tautologies on $\poly(n)$ variables, including any number of copies of the $K$ variables, of the form $\varphi = \left(\left(\bigwedge_{i} \alpha_i\right) \rightarrow \omega\right)$. Let $\vec{p}$ denote the other (non-$K$) variables. Suppose that 1) there are at most $O(\log n)$ non-$K$ variables in $\varphi$, 2) for each copy of $K$, the corresponding $K$-clauses appear amongst the $\alpha_i$, 3) the only $K$ variables that appear in $\omega$ are output variables $k^{(i)}_{out}$, and 4) if $k^{(i)}_{out}$ appears in $\omega$, then all the inputs to $K^{(i)}$ are set to formulas that syntactically depend on at most $\vec{p}$. 

Then there is a $\poly(n)$-size $\cc{AC}^0$-Frege proof of $\varphi$. 
\end{lemma}

\begin{proof}[Proof sketch]
The basic idea is that $\cc{AC}^0$-Frege can brute force over all $\poly(n)$-many assignments to the $O(\log n)$ non-$K$ variables, and for each such assignment can then just evaluate each copy of $K$ gate by gate to verify the tautology. Any copy $K^{(i)}$ of $K$ all of whose input variables are unset must not affect the truth of $\varphi$, since none of the $k^{(i)}$ variables can appear in the consequent $\omega$ of $\varphi$. In fact, for such copies of $K$, the $K$-clauses merely appear as disjuncts of $\varphi$, since it then takes the form $\varphi = \bigvee_{i} (\neg \alpha_i) \vee \omega = \left(\bigvee_{g} \neg (k^{(i)}_{g} \leftrightarrow k^{(i)}_{g_{\ell}} \op{g} k^{(i)}_{g_r}) \right) \vee \left(\bigvee_{\text{remaining clauses $i$}} \neg \alpha_i \right) \vee \omega$. Thus, if $\cc{AC}^0$-Frege can prove that the rest of $\varphi$, namely $\left(\bigvee_{\text{remaining clauses $i$}} \neg \alpha_i \right) \vee \omega$ is a tautology, then it can prove that $\varphi$ is a tautology.
\end{proof}

Now we state the analogs of Lemmata~\ref{lem:soundness} and \ref{lem:axioms} for $\mathcal{C}$-Frege. Because of the similarity of the proofs to the previous case, we merely indicate how their proofs differ from the Extended Frege case.

\begin{lemma}[$\cc{AC}^0$ analog of Lemma~\ref{lem:soundness}] \label{lem:AC0soundness}
Let $\mathcal{C}$ be a class of circuits closed under $\cc{AC}^0$ circuit reductions. If there is a family $K$ of polynomial-size Boolean circuits computing PIT, such that the PIT axioms for $K$ have polynomial-size $\mathcal{C}$-Frege proofs, 
then $\mathcal{C}$-Frege is polynomially equivalent to $\I$.
\end{lemma}

\begin{proof}
Mimic the proof of Lemma~\ref{lem:soundness}. The third and fourth steps of that proof are just modus ponens, so we need only check the first two steps. 

The first step is to show that $\mathcal{C}$-Frege can prove $\varphi \rightarrow Truth_{bool}([\varphi], \prop{\vec{p}})$. This follows directly from the details of the encoding of $[\varphi]$ and the full definition of $Truth_{bool}$; see Lemma~\ref{lem:truth}.

The second step is to show that $\mathcal{C}$-Frege can prove $Proof_\I([C],[\varphi])$ for a fixed $C,\varphi$. In Lemma~\ref{lem:soundness}, this followed because this statement was variable-free. Now this statement is no longer variable-free, since it involve two copies of $K$ and the corresponding variables and $K$-clauses. However, $Proof_\I([C],[\varphi])$ satisfies the requirements of Lemma~\ref{lem:K}, and applying that lemma we are done.
\end{proof}

\begin{lemma}[$\cc{AC}^0$ analog of Lemma~\ref{lem:axioms}] \label{lem:AC0axioms}
Let $\mathcal{C}$ be a class of circuits closed under $\cc{AC}^0$ circuit reductions. If $\mathcal{C}$-Frege can efficiently prove the PIT axioms for some polynomial-sized family of circuits $K$ computing PIT, then $\mathcal{C}$-Frege can efficiently prove $Soundness_{\I}$ (for that same $K$).
\end{lemma}

\begin{proof}
We mimic the proof of Lemma~\ref{lem:axioms}. In steps (\ref{step:axioms:1}), (\ref{step:axioms:2}), and (\ref{step:axioms:4}) of that proof we used $m$ additional copies of $K$, where $m$ is the number of clauses in the CNF $\varphi$ encoded by $\prop{[\varphi]}$, and thus $m \leq \poly(n)$. In order to talk about these copies of $K$ in $\mathcal{C}$-Frege, however, the $K$ variables must already be present in the statement we wish to prove in $\mathcal{C}$-Frege. The ``dummy statements'' in the new version of soundness are the $K$-clauses---with inputs and outputs not set to anything---for each of $m$ new copies of $K$, which we denote $K^{(3)}, \dotsc, K^{(m+2)}$ (recall that the first two copies $K^{(1)}$ and $K^{(2)}$ are already used in the statement of $Proof_\I$). We won't actually need these clauses anywhere in the proof, we just need their variables to be present from the beginning.

Starting with $Truth_{bool}(\prop{[\varphi]},\vec{p})$, $K^{(1)}(\prop{[C(\vec{x},\vec{0})]})$, $K^{(2)}(\prop{[1-C(\vec{x},\vec{Q}(\vec{x}))]})$ we'll derive a contradiction. 
The only step of the proof of Lemma~\ref{lem:axioms} that was not either the use of an axiom or modus ponens was step (\ref{step:axioms:1}), so it suffices to verify that this can be carried out in $\cc{AC}^0$-Frege with the $K$-clauses.

Step (\ref{step:axioms:1}) was to show for every $i \in [m]$, $Truth_{bool}([\varphi],\vec{p}) \rightarrow K(\prop{[Q_i^\varphi(\vec{p})]})$, where $Q_i^\varphi$ is the low degree polynomial corresponding to the clause, $\kappa_i$, of $\varphi$. Note that, as $\varphi$ is not a fixed formula but is determined by the propositional variables encoding $\prop{[\varphi]}$, the encoding $\prop{[Q_i^\varphi]}$ depends on a subset of these variables. 

$Truth_{bool}(\prop{[\varphi]},\vec{p})$ states that each clause $\kappa_i$ in $\varphi$ evaluates
to true under $\vec{p}$. It is a tautology that if $\kappa_i$
evaluates to true under $\vec{p}$, then $Q_i^{\varphi}$ evaluates to $0$ at $\vec{p}$. Since $K$ correctly computes PIT, \begin{equation} \label{eqn:truthAC0}
Truth_{bool}(\prop{[\kappa_i]},\vec{p}) \rightarrow K^{(i+2)}(\prop{[Q_i^\varphi(\vec{p})]})\tag{**}
\end{equation}
is a tautology. Furthermore, although both the encoding $\prop{[\kappa_i]}$ and $\prop{[Q_i^{\varphi}]}$ depend on the propositional variables encoding $\prop{[\varphi]}$, since we assume that $\varphi$ is a 3CNF, these only depend on \emph{constantly many} of the variables encoding $\prop{[\varphi]}$. Writing out (\ref{eqn:truthAC0}) it has the form 
\[
Truth_{bool} \rightarrow \left(\text{$K^{(i+2)}$-clauses } \land (\text{ setting inputs of $K^{(i+2)}$ to $\prop{[Q_i^\varphi(\vec{p})]}$}) \rightarrow k^{(i+2)}_{out}  \right),
\]
which is equivalent to 
\[
Truth_{bool} \land (K^{(i+2)}\text{-clauses}) \land (\text{ setting inputs of $K^{(i+2)}$ to $\prop{[Q_i^\varphi(\vec{p})]}$}) \rightarrow k^{(i+2)}_{out}. 
\]
Thus (\ref{eqn:truthAC0}) satisfies the conditions of Lemma~\ref{lem:K} and has a short $\cc{AC}^0$-Frege proof. Since $Truth_{bool}(\prop{[\varphi]}, \vec{p})$ is defined as $\bigwedge_i Truth_{bool}(\prop{[\kappa_i]}, \vec{p})$ (see Section~\ref{sec:encoding}), we can then derive 
\[
Truth_{bool}(\prop{[\varphi]},\vec{p}) \rightarrow K^{(i+2)}(\prop{[Q_i^\varphi(\vec{p})]}),
\]
as desired.
\end{proof}

\subsection{Some details of the encodings} \label{sec:encoding}
For an $\leq m$-clause, $\leq n$-variable 3CNF $\varphi = \kappa_1 \land \dotsb \land \kappa_m$, its encoding is a Boolean string of length $3m( \lceil \log_2(n) \rceil+1)$. Each literal $x_i$ or $\neg x_i$ is encoded as the binary encoding of $i$ ($\lceil \log_2(n) \rceil$ bits) plus a single other bit indicating whether the literal is positive (1) or negative (0). The encoding of a single clause is just the concatenation of the encodings of the three literals, and the encoding of $\varphi$ is the concatenation of these encodings.

We define
\[
Truth_{bool,n,m}(\prop{[\varphi]}, \vec{p}) \defeq \bigwedge_{i=1}^{m} Truth_{bool,n}(\prop{[\kappa_i]}, \vec{p}).
\]

For a single 3-literal clause $\kappa$, we define $Truth_{bool,n}(\prop{[\kappa]}, \vec{p})$ as follows. For an integer $i$, let $[i]$ denote the standard binary encoding of $i-1$ (so that the numbers $1,\dotsc,2^k$ are put into bijective correspondence with $\{0,1\}^{k}$). Let $\prop{[\kappa]} = \vec{q_1} s_1 \vec{q_2} s_2 \vec{q_3} s_3$ where each $s_i$ is the sign bit (positive/negative) and each $\vec{q_i}$ is a length-$\lceil \log_2 n \rceil$ string of variables corresponding to the encoding of the index of a variable. We write $\vec{q} = [k]$ as shorthand for $\bigwedge_{i=1}^{\lceil \log_2 n \rceil} (q_i \leftrightarrow [k]_i)$, where $x \leftrightarrow y$ is shorthand for $(x \land y) \lor (\neg x \land \neg y)$. Finally, we define:
\[
Truth_{bool,n}(\prop{[\kappa]}, \vec{p}) \defeq \bigvee_{j=1}^{3} \bigvee_{i=1}^{n} (\vec{q}_j = [i] \land (p_i \leftrightarrow s_j)).
\]
(Hereafter we drop the subscripts $n,m$; they should be clear from context.)

\begin{lemma} \label{lem:truth}
For any 3CNF $\varphi$ on $n$ variables, there are $\poly(n)$-size $\cc{AC}^0$-Frege proofs of $\varphi(\vec{p}) \rightarrow Truth_{bool}([\varphi], \vec{p})$.
\end{lemma}

\begin{proof}
In fact, we will see that for a fixed clause $\kappa$, after simplifying constants---that is, $\varphi \land 1$ and $\varphi \lor 0$ both simplify to $\varphi$, $\varphi \land 0$ simplifies to $0$, and $\varphi \lor 1$ simplifies to $1$---that $Truth_{bool}([\kappa], \vec{p})$ in fact becomes \emph{syntactically identical} to $\kappa(\vec{p})$. By the definition of $Truth_{bool}([\varphi], \vec{p})$, we get the same conclusion for any fixed CNF $\varphi$. Simplifying constants can easily be carried out in $\cc{AC}^0$-Frege. 

For a fixed $\kappa$, $\vec{q}_j$ and $s_j$ become fixed to constants for $j=1,2,3$. Denote the indices of the three variables in $\kappa$ by $i_1, i_2, i_3$. The only variables left in the statement $Truth_{bool}([\kappa], \vec{p})$ are $\vec{p}$. Since the $\vec{q}_{j}$ and $[i]$ are all fixed, every term in $\bigvee_{i}( \vec{q}_j = [i] \land (p_i \leftrightarrow s_j))$ except for the $i_j$ term simplifies to $0$, so this entire disjunction simplifies to $(p_{i_j} \leftrightarrow s_j)$. Since the $s_j$ are also fixed, if $s_j=1$ then $(p_{i_j} \leftrightarrow s_j)$ simplifies to $p_{i_j}$, and if $s_j=0$ then it simplifies to $\neg p_{i_j}$. With this understanding, we write $\pm p_{i_j}$ for the corresponding literal. Then $Truth_{bool}([\kappa], \vec{p})$ simplifies to $(\pm p_{i_1} \lor \pm p_{i_2} \lor \pm p_{i_3})$ (with signs as described previously). This is exactly $\kappa(\vec{p})$.
\end{proof}

\section*{Acknowledgments}
We thank David Liu for many interesting discussions, and for collaborating with us on some of the open questions posed in this paper. We thank Eric Allender and Andy Drucker for asking whether ``Extended Frege-provable PIT'' implied that $\I$ was equivalent to Extended Frege, which led to the results of Section~\ref{sec:EF}. We thank Pascal Koiran for providing the second half of the proof of Proposition~\ref{prop:coMAGRH}. We thank Iddo Tzameret for useful discussions that led to Proposition~\ref{prop:pitassi}. Finally, in addition to several useful discussions, we also thank Eric Allender for suggesting the name ``Ideal Proof System''---all of our other potential names didn't even hold a candle to this one. We gratefully acknowledge financial support from NSERC; in particular, J. A. G. was supported by A. Borodin's NSERC Grant \# 482671.

\bibliographystyle{ams-alpha}
\bibliography{algpf}

\appendix
\section{Additional Background} \label{app:background}
\subsection{Algebraic Complexity} \label{app:background:complexity}
A polynomial $f(\vec{x})$ is a \definedWord{projection} of a polynomial $g(\vec{y})$ if $f(\vec{x}) = g(L(\vec{x}))$ identically as polynomials in $\vec{x}$, for some map $L$ that assigns to each $y_i$ either a variable or a constant. A family of polynomials $(f_n)$ is a polynomial projection or \definedWord{p-projection} of another family $(g_n)$, denoted $(f_n) \leq_{p} (g_n)$, if there is a function $t(n) = n^{\Theta(1)}$ such that $f_n$ is a projection of $g_{t(n)}$ for all (sufficiently large) $n$. The primary value of projections is that they are very simple, and thus preserve bounds on nearly all natural complexity measures. Valiant \cite{valiant, valiantProjections} was the first to point out not only their value but their ubiquity in computational complexity---nearly all problems that are known to be complete for some natural class, even in the Boolean setting, are complete under p-projections.  We say that two families $f=(f_n)$ and $g=(g_n)$ are of the same p-degree if each is a p-projection of the other, which we denote $f \equiv_{p} g$. 

By analogy with Turing reductions or circuit reductions, \Burgisser \cite{burgisserbook} introduced the more general, but somewhat messier, notion of c-reduction (``c'' for ``computation''). An oracle computation of $f$ from $g$ is an algebraic circuit $C$ with ``oracle gates'' such that when $g$ is plugged in for each oracle gate, the resulting circuit computes $f$. We say that a family $(f_n)$ is a c-reduction of $(g_n)$ if there is a function $t(n) = n^{\Theta(1)}$ such that there is a polynomial-size oracle reduction from $f_n$ to $g_{t(n)}$ for all sufficiently large $n$. We define c-degrees by analogy with p-degrees, and denote them by $\equiv_{c}$.

Despite its central role in computation, and the fact that $\cc{VP} = \cc{VNC}^2$ \cite{VSBR}, the determinant is not known to be $\cc{VP}$-complete. The determinant is $\cc{VQP}$-complete ($\cc{VQP}$ is defined just like $\cc{VP}$ but with a quasi-polynomial bound on the size and degree of the circuits) under qp-projections (like p-projections, but with a quasi-polynomial bound). Weakly skew circuits help clarify the complexity of the determinant (see Malod and Portier \cite{malodPortier} for some history of weakly skew circuits and for highlights of their utility). A circuit of fan-in at most $2$ is \definedWord{weakly skew} if for every multiplication gate $g$ receiving inputs from gates $g_1$ and $g_2$, at least one of the subcircuits $C_i$ rooted at $g_i$ is only connected to the rest of the circuit through $g$. In other words, for every multiplication gate, one of its two incoming factors was computed entirely and solely for the purpose of being used in that multiplication gate. Toda \cite{todaDet2} (see also Malod and Portier \cite{malodPortier} showed that a polynomial family $f=(f_n)$ is a p-projection of the determinant family $(\det_n)$ if and only if $f$ is computed by polynomial-size weakly skew circuits.

\subsection{Proof Complexity} \label{app:background:proof}
Here we give formal definitions of proof systems and probabilistic proof systems
for $\cc{coNP}$ languages, and discuss several important and standard proof systems
for TAUT.

\begin{definition}
Let $L \subseteq \{0,1\}^*$ be a $\cc{coNP}$ language. A \definedWord{proof system $P$ for $L$} is a polynomial-time
function of two inputs $x,y \in \{0,1\}^*$ with the following properties:
\begin{enumerate}
\item (Perfect Soundness) If $x$ is not in $L$, then 
for every $y$, $P(x,y)=0$.
\item (Completeness) If $x$ is in $L$, then there exists a $y$ 
such that $P(x,y)=1$.
\end{enumerate}
$P$ is \definedWord{polynomially bounded} if for every $x \in L$, there exists a $y$
such that $|y|\leq poly(|x|)$ and $P(x,y)=1$.
\end{definition}

As this is just the definition of an $\cc{NP}$ procedure for $L$,
it follows that for any $\cc{coNP}$-complete language $L$, $L$ has
a polynomially bounded proof system if and only if $\cc{coNP} \subseteq \cc{NP}$.

Cook and Reckhow \cite{cookReckhow} formalized proof systems for the language TAUT (all
Boolean tautologies) in a slightly different way, although their definition is
essentialy equivalent to the one above. We prefer the above definition as it
is consistent with definitions of interactive proofs.

\begin{definition}
A \definedWord{Cook--Reckhow proof system} is a polynomial-time function
$P'$ of just one input $y$, and whose range is the set of all yes instances of $L$. 
If $x \in L$, then any $y$ such that
$P'(y)=x$ is called a $P'$ proof of $x$. $P'$ must satisfy the 
following properties:
\begin{enumerate}
\item  (Soundness) For every $x,y \in \{0,1\}^*$, if $P'(y)=x$, then $x \in L$.
\item (Completeness) For every $x \in L$, there exists an $y$ such that $P'(y)=x$.
\end{enumerate}

$P'$ is \definedWord{polynomially bounded} if for every $x \in L$, there exists a $y$
such that $|y| \leq poly(|x|)$ and $P(y)=x$.
\end{definition}

Intuitively, we think of $P'$ as a procedure for verifying that $y$ is a proof that some $x \in L$
and if so, it outputs $x$. 
(For all strings $x$ that do not encode valid proofs, $P'(x)$ may just output some canonical $x_0 \in L$.)
It is a simple exercise to see that for every language $L$, any propositional proof system $P$
according to our definition can be converted to a Cook-R-eckow proof system $P'$, and vice versa,
and furthermore the runtime properties of $P$ and $P'$ will be the same.
In the forward direction, say $P$ is a proof system for $L$ according to our definition.
Define Merlin's string $y$ as encoding a pair $(x,y')$, and on input $y=(x,y')$, $P'$
runs $P$ on the pair $(x,y')$. If $P$ accepts, then $P'(y)$ outputs $x$,
and if $P$ rejects, then $P'(y)$ outputs (the encoding of) a canonical $x^0$ in $L$.
Conversely, say that $P'$ is a Cook-Reckhow proof system for $L$.
$P(x,y)$ runs $P'$ on $y$ and accepts if and only if $P'(y)=x$.


\begin{definition}
Let $P_1$ and $P_2$ be two proof systems for a language $L$ in $\cc{coNP}$.
$P_1$ p-simulates $P_2$ if for every $x \in L$ and for every $y$ such that
$P_2(x,y)=1$, there exists $y'$ such that $|y'|\leq \poly(|y|)$, and $P_1(x,y')=1$.
\end{definition}

Informally, $P_1$ p-simulates $P_2$ if proofs in $P_1$ are no longer than proofs in $P_2$ (up to
polynomial factors) .

\begin{definition}
Let $P_1$ and $P_2$ be two proof systems for a language $L$ in $\cc{coNP}$.
$P_1$ and $P_2$ are \definedWord{p-equivalent} if $P_1$ p-simulates $P_2$ and $P_2$ p-simulates $P_1$.
\end{definition}

\noindent{\bf Standard Propositional Proof Systems}
For TAUT (or UNSAT), there are a variety of standard and
well-studied proof systems, the most important ones including Extended Frege (EF), Frege, Bounded-depth Frege,
and Resolution. A Frege rule is an inference rule of the form:
$B_1, \ldots, B_n \implies B$, where $B_1,\ldots, B_n,B$ are propositional formulas.
If $n=0$ then the rule is an axiom.
For example, $A \lor \neg A$ is a typical Frege axiom, and
$A, \neg A \lor B \implies B$ is a typical Frege rule.
A Frege system is specified by a finite set, $R$ of rules.
Given a collection $R$ of rules, a derivation of 3DNF formula $f$ is a sequence
of formulas $f_1,\ldots,f_m$ such that each $f_i$ is either an instance
of an axiom scheme, or follows from two previous formulas by one of the rules in $R$,
and such that the final formula $f_m$ is $f$.
In order for a Frege system to be a proof system in the Cook-Reckhow sense, its
corresponding set of rules must be sound and complete.
Work by Cook and Reckhow in the 70's (REF) showed that Frege systems are
very robust in the sense that all Frege systems are
polynomially-equivalent.

Bounded-depth Frege proofs ($\cc{AC}^0$-Frege) are proofs that are Frege
proofs but with the additional restriction that each formula in the
proof has bounded depth. (Because our connectives are binary AND, OR and negation,
by depth we assume the formula has all negations at the leaves, and
we count the maximum number of alternations of AND/OR connectives in the formula.)
Polynomial-sized $AC^0$-Frege proofs correspond to the complexity class
$AC^0$ because such proofs allow a polynomial number of lines, each of which
must be in $AC^0$. 

Extended Frege systems generalize Frege systems by allowing, in addition to
all of the Frege rules, a new axiom of the form $y \leftrightarrow A$,
where $A$ is a formula, and $y$ is a new variable not occurring in $A$.
Whereas polynomially-size Frege proofs allow a polynomial number of
lines, each of which must be a polynomial-sized formula,
using the new axiom, polynomial-size EF proofs allow a polynomial number
of lines, each of which can be a polynomial-sized circuit.
See \cite{krajicekBook} for precise definitions of Frege, $\cc{AC}^0$-Frege, and EF proof systems.

\medskip

\noindent {\bf Probabilistic Proof Systems}
The concept of a proof system for a language in $\cc{coNP}$ can be generalized 
in the natural way, to obtain Merlin--Arthur style proof systems.

\begin{definition}
Let $L$ be a language in $\cc{coNP}$, and let $V$ be a probabilistic polynomial-time algorithm
with two inputs $x,y \in \{0,1\}^*$.
(We think of $V$ as the verifier.)
$V$ is a \definedWord{probabilistic proof system} for $L$ if:
\begin{enumerate}
\item (Perfect Soundness) For every $x$ that is not in $L$,
and for every $y$, 
$$Pr_r[P(x,y) =1] =0 ,$$
where the probability is over the random coin tosses, $r$ of $P$.
\item (Completeness) For every $x$ in $L$,
there exists a $y$ such that
$$Pr_r[P(x,y) =1] \geq 3/4.$$
\end{enumerate}
$P$ is \definedWord{polynomially bounded} if for every $x \in L$, there exists $y$ such that
$|y|=poly(|x|)$ and $Pr_r[P(x,y)=1] \geq 3/4$.
\end{definition}

It is clear that for any $\cc{coNP}$-complete language $L$, there is a polynomially
bounded probabilistic proof system for $L$ if and only if $\cc{coNP} \subseteq \cc{MA}$.

Again we have chosen to define our probabilitic proof systems to
match the definition of $\cc{MA}$. The probabilistic proof system that
would be analogous to the standard Cook--Reckhow proof system would
be somewhat different, as defined below.
Again, a simple argument like the one above shows that our probablistic proof systems
are essentially equivalent to a probabilistic Cook--Reckhow proof systems.

\begin{definition}
A \definedWord{probabilistic Cook--Reckhow proof system} is a probabilistic polynomial-time algorithm 
$A$ (whose run time is independent of its random choices) such that
\begin{enumerate}
\item There is a surjective function $f\colon \Sigma^{*} \to TAUT$ such that $A(x)=f(x)$ with probability at least $2/3$ (over $A$'s random choices), and

\item Regardless of $A$'s random choices, its output is always a tautology.
\end{enumerate}

Such a proof system is \definedWord{polynomially bounded} or \definedWord{p-bounded} if for every tautology $\varphi$, there is some $\pi$ (for ``proof'') such that $f(\pi)=\varphi$ and $|\pi| \leq \poly(|\varphi|)$.
\end{definition}

We note that both Pitassi's algebraic proof system \cite{pitassi96} and the Ideal Proof System are 
probabilistic Cook--Reckhow systems. The algorithm $P$ 
takes as input a description of a (constant-free) algebraic circuit $C$ together with a tautology $\varphi$, 
and then verifies that the circuit is indeed an \I-certificate for $\varphi$ by using the standard $\cc{coRP}$ 
algorithm for polynomial identity testing. 
The proof that Pitassi's algebraic proof system is a probabilistic Cook--Reckhow system is essentially the same.
                                                                                                                             
\subsection{Commutative algebra} \label{app:background:algebra}
The following preliminaries from commutative algebra are needed only in Section~\ref{sec:syzygy} and Appendix~\ref{app:RIPS}.

A \definedWord{module} over a ring $R$ is defined just like a vector space, except over a ring instead of a field. That is, a module $M$ over $R$ is a set with two operations: addition (making $M$ an abelian group), and multiplication by elements of $R$ (``scalars''), satisfying the expected axioms (see any textbook on commutative algebra, \eg, \cite{atiyahMacdonald,eisenbud}). A module over a field $R = \F$ is exactly a vector space over $\F$. Every ring $R$ is naturally an $R$-module (using the ring multiplication for the scalar multiplication), as is $R^{n}$, the set of $n$-tuples of elements of $R$. Every ideal $I \subseteq R$ is an $R$-module---indeed, an ideal could be defined, if one desired, as an $R$-submodule of $R$---and every quotient ring $R/I$ is also an $R$-module, by $r \cdot (r_0 + I) = rr_0 + I$.

Unlike vector spaces, however, there is not so nice a notion of ``dimension'' for modules over arbitrary rings. Two differences will be particularly relevant in our setting. First, although every vector subspace of $\F^{n}$ is finite-dimensional, hence finitely generated, this need not be true of every submodule of $R^n$ for an arbitrary ring $R$. Second, every (finite-dimensional) vector space $V$ has a basis, and every element of $V$ can be written as a \emph{unique} $\F$-linear combination of basis elements, but this need not be true of every $R$-module, even if the $R$-module is finitely generated, as in the following example.

\begin{example}
Let $R=\C[x,y]$ and consider the ideal $I = \langle x, y \rangle$ as an $R$-module. For clarity, let us call the generators of this $R$-module $g_1 = x$ and $g_2 = y$. First, $I$ cannot be generated as an $R$-module by fewer than two elements: if $I$ were generated by a single element, say, $f$, then we would necessarily have $x=r_1 f$ and $y=r_2 f$ for some $r_1,r_2 \in R$, and thus $f$ would be a common divisor of $x$ and $y$ in $R$ (here we are using the fact that $I$ is both a module and a subset of $R$). But the GCD of $x$ and $y$ is $1$, and the only submodule of $R$ containing $1$ is $R \neq I$. So $\{g_1, g_2\}$ is a minimum generating set of $I$. But not every element of $I$ has a unique representation in terms of this (or, indeed, any) generating set: for example, $xy \in I$ can be written either as $r_1 g_1$ with $r_1=y$ or $r_2 g_2$ with $r_2 = x$.
\end{example}

A ring $R$ is \definedWord{Noetherian} if there is no strictly increasing, infinite chain of ideals $I_1 \subsetneq I_2 \subsetneq I_3 \subsetneq \dotsb$. Fields are Noetherian (every field has only two ideals: the zero ideal and the whole field), as are the integers $\Z$. Hilbert's Basis Theorem says that every ideal in a Noetherian ring is finitely generated. Hilbert's (other) Basis Theorem says that if $R$ is finitely generated, then so is the polynomial ring $R[x]$ (and hence any polynomial ring $R[\vec{x}]$. Quotient rings of Noetherian rings are Noetherian, so every ring that is finitely generated over a field (or more generally, over a Noetherian ring $R$) is Noetherian.

Similarly, an $R$-module $M$ is Noetherian if there is no strictly increasing, infinite chain of submodules $M_1 \subsetneq M_2 \subsetneq M_3 \subsetneq \dotsb$. If $R$ is Noetherian as a ring, then it is Noetherian as an $R$-module. It is easily verified that direct sums of Noetherian modules are Noetherian, so if $R$ is a Noetherian ring, then it is a Noetherian $R$-module, and consequently $R^{n}$ is a Noetherian $R$-module for any finite $n$. Just as for ideals, every submodule of a Noetherian module is finitely generated.

The remaining preliminaries from commutative algebra are only needed in Appendix~\ref{app:RIPS}.

The \definedWord{radical} of an ideal $I \subseteq R$ is the ideal $\sqrt{I}$ consisting of all $r \in R$ such that $r^k \in I$ for some $k > 0$. An ideal $I$ is \definedWord{prime} if whenever $rs \in P$, at least one of $r$ or $s$ is in $P$. For any ideal $I$, its radical is equal to the intersection of the prime ideals containing $I$: $\sqrt{I} = \bigcap_{\text{prime } P \supseteq I} P$. We refer to prime ideals that are minimal under inclusion, subject to containing $I$, as ``minimal over $I$;'' there are only finitely many such prime ideals. The radical $\sqrt{I}$ is thus also equal to the intersection of the primes minimal over $I$.

An \definedWord{algebraic set} in $\F^n$ is any set of the form $\{\vec{x} \in \F^{n} : F_1(\vec{x}) = \dotsb = F_m(\vec{x}) = 0\}$, which we denote $V(F_1, \dotsc, F_m)$ (``$V$'' for ``variety''). The algebraic set $V(F_1, \dotsc, F_m)$ depends only on the ideal $\langle F_1, \dotsc, F_m \rangle$, and even its radical, in the sense that $V(F_1, \dotsc, F_m) = V(\sqrt{\langle F_1, \dotsc, F_m \rangle})$. Conversely, the set of all polynomials vanishing on a given algebraic set $V$ is a radical ideal, denoted $I(V)$. An algebraic set is \definedWord{irreducible} if it cannot be written as a union of two algebraic proper subsets. $V$ is irreducible if and only if $I(V)$ is prime. The \definedWord{irreducible components} of an algebraic set $V = V(I)$ are the maximal irreducible algebraic subsets of $V$, which are exactly the algebraic sets corresponding to the prime ideals minimal over $I$.

If $U$ is any subset of a ring $R$ that is closed under multiplication---$a,b \in U$ implies $ab \in U$---we may define the localization of $R$ at $U$ to be the ring in which we formally adjoin multiplicative inverses to the elements of $U$. Equivalently, we may think of the localization of $R$ at $U$ as the ring of fractions over $R$ where the denominators are all in $U$. If $P$ is a prime ideal, its complement is a multiplicatively closed subset (this is an easy and instructive exercise in the definition of prime ideal). In this case, rather than speak of the localization of $R$ at $R \backslash P$, it is common usage to refer to the localization of $R$ and $P$, denoted $R_P$. Similar statements hold for the union of finitely many prime ideals. We will use the fact that the localization of a Noetherian ring is again Noetherian (however, even if $R$ is finitely generated its localizations need not be, \eg the localization of $\Z$ at $P = \langle 2 \rangle$ consists of all rationals with odd denominators; this is one of the ways in which the condition of being Noetherian is nicer than that of being merely finitely generated).
\section{Divisions: the Rational Ideal Proof System} \label{app:RIPS}
We begin with an example where it is advantageous to include divisions in an \I-certificate. Note that this is different than merely computing a polynomial \I-certificate using divisions. In the latter case, divisions can be eliminated \cite{strassenDivision}. In the case we discuss here, the certificate itself is no longer a polynomial but is a rational function.

\begin{example} \label{ex:inversion}
The inversion principle, one of the so-called ``Hard Matrix Identities'' \cite{soltysCook}, states that 
\[
XY = I \Rightarrow YX = I.
\]
They are called ``Hard'' because they were proposed as possible examples---over $\F_2$ or $\Z$---of propositional tautologies separating Extended Frege from Frege. Indeed, it was only in the last 10 years that they were shown to have efficient Extended Frege proofs \cite{soltysCook}, and it was quite nontrivial to show that they have efficient $\cc{NC}^2$-Frege proofs \cite{hrubesTzameretDet}, despite the fact that the determinant can be computed in $\cc{NC}^2$. It is still open whether the Hard Matrix Identities have ($\cc{NC}^1$)-Frege proofs, and believed not to be the case.

In terms of ideals, the inversion principle says that the $n^2$ polynomials $(YX - I)_{i,j}$ (the entries of the matrix $YX -I$) are in the ideal generated by the $n^2$ polynomials $(XY-I)_{i,j}$. The simplest rational proof of the inversion principle that we are aware of is as follows:
\[
X^{-1} (XY-I) X = YX-I
\]
Note that $X^{-1}$ here involves dividing by the determinant. When converted into a certificate, if we write $Q$ for a matrix of placeholder variables $q_{i,j}$ corresponding to the entries of the matrix $XY-I$, we get $n^2$ certificates from the entries of $X^{-1} Q X$. Note that each of these certificates is a rational function that has $\det(X)$ in its denominator. Turning this into a proof that does not use divisions is the main focus of the paper \cite{hrubesTzameretDet}; thus, if we had a proof system that allowed divisions in this manner, it would potentially allow for significantly simpler proofs. In this particular case, we assure ourselves that this is a valid proof because if $XY-I=0$, then $X$ is invertible, so $X^{-1}$ exists (or equivalently, $\det(X) \neq 0$).
\end{example}

In order to introduce an \I-like proof system that allows rational certificates, we generalize the preceding reasoning. We must be careful what we allow ourselves to divide by. If we are allowed to divide by arbitrary polynomials, this would yield an unsound proof system, because then from any polynomials $F_1(\vec{x}), \dotsc, F_m(\vec{x})$ we could derive \emph{any} other polynomial $G(\vec{x})$ via the false ``certificate'' $\frac{G(x)}{F(x)}\f_1$. The following definition is justified by Proposition~\ref{prop:RIPS}.

Unfortunately, although we try to eschew as many definitions as possible, the results here are made much cleaner by using some additional (standard) terminology from commutative algebra which is covered in Appendix~\ref{app:background:algebra} such as prime ideals, irreducible components of algebraic sets, and localization of rings. 

\begin{definition}[Rational Ideal Proof System] \label{def:RIPS}
A \definedWord{rational \I certificate} or \definedWord{R\I-certificate} that a polynomial $G(\vec{x}) \in \F[\vec{x}]$ is in the radical of the $\overline{\F}[\vec{x}]$-ideal generated by $F_1(\vec{x}), \dotsc, F_m(\vec{x})$ is a rational function $C(\vec{x}, \vec{\f})$ such that
\begin{enumerate}
\setcounter{enumi}{-1}
\item \label{condition:local} Write $C = C'/D$ with $C',D$ relatively prime polynomials. Then $1/D(\vec{x}, \vec{F}(\vec{x}))$ must be in the localization of $\F[\vec{x}]$ at the union of the prime ideals that are minimal subject to containing the ideal $\langle F_1(\vec{x}), \dotsc, F_m(\vec{x}) \rangle$ (We give a more elementary explanation of this condition below), 

\item \label{condition:RIPSideal}$C(x_1,\dotsc,x_n,\vec{0}) = 0$, and

\item \label{condition:RIPSnss} $C(x_1,\dotsc,x_n,F_1(\vec{x}),\dotsc,F_m(\vec{x})) = G(\vec{x})$.
\end{enumerate}
A \definedWord{R\I proof} that $G(\vec{x})$ is in the radical of the ideal $\langle F_1(\vec{x}), \dotsc, F_m(\vec{x}) \rangle$ is an $\F$-algebraic circuit with divisions on inputs $x_1,\ldots,x_n,\f_1,\ldots,\f_m$ computing some R\I certificate.
\end{definition}

Condition~(\ref{condition:local}) is equivalent to: if $G(\vec{x})$ is an invertible constant, then $D(\vec{x}, \vec{\f})$ is also an invertible constant and thus $C$ is a polynomial; otherwise, after substituting the $F_i(\vec{x})$ for the $\f_i$, the denominator $D(\vec{x}, \vec{F}(\vec{x}))$ does not vanish identically on any of the irreducible components (over the algebraic closure $\overline{\F}$) of the algebraic set $V(\langle F_1(\vec{x}), \dotsc, F_m(\vec{x}) \rangle) \subseteq \overline{\F}^{n}$. In particular, for proofs of unsatisfiability of systems of equations, the Rational Ideal Proof System reduces by definition to the Ideal Proof System. For derivations of one polynomial from a set of polynomials, this need not be the case, however; indeed, there are examples for which \emph{every} R\I-certificate has a nonconstant denominator, that is, there is a R$\I$-certifiate but there are no \I-certificates (see Example~\ref{ex:divNeeded}).

The following proposition establishes that Definition~\ref{def:RIPS} indeed defines a sound proof system. 

\begin{proposition} \label{prop:RIPS}
If there is a R\I-certificate that $G(\vec{x})$ is in the radical of $\langle F_1(\vec{x}), \dotsc, F_m(\vec{x}) \rangle$, then $G(\vec{x})$ is in fact in the radical of $\langle F_1(\vec{x}), \dotsc, F_m(\vec{x}) \rangle$. 
\end{proposition}

\begin{proof}
Let $C(\vec{x}, \vec{\f}) = \frac{1}{D(\vec{x}, \vec{\f})} C'(\vec{x}, \vec{\f})$ be a R\I certificate that $G$ is in $\sqrt{\langle F_1, \dotsc, F_m \rangle}$, where $D$ and $C'$ are relatively prime polynomials. Then $C'(\vec{x}, \vec{\f})$ is an \I-certificate that $G(\vec{x})D(\vec{x}, \vec{F}(\vec{x}))$ is in the ideal $\langle F_1(\vec{x}), \dotsc, F_m(\vec{x}) \rangle$ (recall Definition~\ref{def:IPSideal}). Let $D_{F}(\vec{x}) = D(\vec{x}, \vec{F}(\vec{x}))$. 

Geometric proof: since $G(\vec{x}) D_{F}(\vec{x}) \in \langle F_1(\vec{x}), \dotsc, F_m(\vec{x}) \rangle$, $GD_{F}$ must vanish identically on every irreducible component of the algebraic set $V(F_1, \dotsc, F_m)$. On each irreducible component $V_i$, since $D_{F}(\vec{x})$ does not vanish identically on $V_i$, $G(\vec{x})$ must vanish everywhere except for the proper subset $V(D_{F}(\vec{x})) \cap V_i$. Since $D_{F}$ does not vanish identically on $V_i$, we have $\dim V(D_{F}) \cap V_i \leq \dim V_i - 1$ (in fact this is an equality). In particular, this means that $G$ must vanish on a dense subset of $V_i$. Since $G$ is a polynomial, by (Zariski-)continuity, $G$ must vanish on all of $V_i$. Finally, since $G$ vanishes on every irreducible component of $V(F_1, \dotsc, F_m)$, it vanishes on $V(F_1, \dotsc, F_m)$ itself, and by the Nullstellensatz, $G \in \sqrt{\langle F_1, \dotsc, F_m\rangle}$.

Algebraic proof: for each prime ideal $P_i \subseteq \overline{\F}[\vec{x}]$ that is minimal subject to containing $\langle F_1, \dotsc, F_m \rangle$, $D_{F}$ is not in $P_i$, by the definition of $R\I$-certificate. Since $GD_{F} \in \langle F_1, \dotsc, F_m \rangle \subseteq P_i$, by the definition of prime ideal $G$ must be in $P_i$. Hence $G$ is in the intersection $\bigcap_i P_i$ over all minimal prime ideals $P_i \supseteq \langle F_1, \dotsc, F_m \rangle$. This intersection is exactly the radical $\sqrt{\langle F_1, \dotsc, F_m \rangle}$.
\end{proof}

Any derivation of a polynomial $G$ that is in the radical of an ideal $I$ but not in $I$ itself will require divisions. Although it is not \emph{a priori} clear that R\I could derive even one such $G$, the next example shows that this is the case. In other words, the next example shows that certain derivations \emph{require} rational functions.

\begin{example} \label{ex:divNeeded}
Let $G(x_1, x_2) = x_1$, $F_1(\vec{x}) = x_1^2$, $F_2(\vec{x}) = x_1 x_2$. Then $C(\vec{x}, \vec{\f}) = \frac{1}{x_1-x_2}(\f_1 - \f_2)$ is a R\I-certificate that $G \in \sqrt{\langle F_1, F_2 \rangle}$: by plugging in one can verify that $C(\vec{x}, \vec{F}(\vec{x})) = G(\vec{x})$. For Condition~(\ref{condition:local}), we see that $V(F_1, F_2)$ is the entire $x_2$-axis, on which $x_1 - x_2$ only vanishes at the origin. However, there is no \I-certificate that $G \in \langle F_1, F_2 \rangle$, since $G$ is \emph{not} in $\langle F_1, F_2 \rangle$: $\langle F_1, F_2 \rangle = \{ x(H_1(\vec{x}) x_1 + H_2(\vec{x}) x_2)\}$ where $H_1, H_2$ may be arbitrary polynomials. Since the only constant of the form $H_1(\vec{x}) x_1 + H_2(\vec{x}) x_2$ is zero, $G(x) = x \notin \langle F_1, F_2 \rangle$.
\end{example}

In the following circumstances a R\I-certificate can be converted into an \I-certificate.

\paragraph{Notational convention.} Throughout, we continue to use the notation that if $D$ is a function of the placeholder variables $\f_i$ (and possibly other variables), then $D_{F}$ denotes $D$ after substituting in $F_i(\vec{x})$ for the placeholder variable $\f_i$.

\begin{proposition} \label{prop:RIPS2IPS}
If $C = C'/D$ is a R\I proof that $G(\vec{x}) \in \sqrt{\langle F_1(\vec{x}), \dotsc, F_m(\vec{x}) \rangle}$, such that $D_{F}(\vec{x})$ does not vanish \emph{anywhere} on the algebraic set $V(F_1(\vec{x}), \dotsc, F_m(\vec{x}))$, then $G(\vec{x})$ is in fact in the ideal $\langle F_1(\vec{x}), \dotsc, F_m(\vec{x}) \rangle$. Furthermore, there is an \I proof that $G(\vec{x}) \in \langle F_1(\vec{x}), \dotsc, F_m(\vec{x}) \rangle$ of size $\poly(|C|,|E|)$ where $E$ is an $\I$ proof of the unsolvability of $D_{F}(\vec{x}) = F_1(\vec{x}) = \dotsb = F_m(\vec{x}) = 0$. 
\end{proposition}

\begin{proof}
Since $D_{F}(\vec{x})$ does not vanish anywhere on $V(F_1, \dotsc, F_m)$, the system of equations $D_F(\vec{x}) = F_1(\vec{x}) = \dotsb = F_m(\vec{x}) = 0$ is unsovlable. 

Geometric proof idea: The preceding means that when restricted to the algebraic set $V(F_1, \dotsc, F_m)$, $D_{F}$ has a multiplicative inverse $\Delta$. Rather than dividing by $D$, we then multiply by $\Delta$, which, for points on $V(F_1, \dotsc, F_m)$, amounts to the same thing.

Algebraic proof: Let $E(\vec{x}, \vec{\f}, d)$ be an \I-certificate for the unsolvability of this system, where $d$ is a new placeholder variable corresponding to the polynomial $D_{F}(\vec{x}) = D(\vec{x}, \vec{F}(\vec{x}))$. By separating out all of the terms involving $d$, we may write $E(\vec{x}, \vec{\f}, d)$ as $d\Delta(\vec{x}, \vec{\f}, d) + E'(\vec{x}, \vec{\f})$. As $E(\vec{x}, \vec{F}(\vec{x}), D_{F}(\vec{x})) = 1$ (by the definition of \I), we get:
\[
D_{F}(\vec{x})\Delta(\vec{x}, \vec{F}(\vec{x}), D_{F}(\vec{x})) = 1 - E'(\vec{x}, \vec{F}(\vec{x})).
\]
Since $E'(\vec{x}, \vec{\f}) \in \langle \f_1, \dotsc, \f_m \rangle$, this tells us that $\Delta(\vec{x}, \vec{F}(\vec{x}), D_{F}(\vec{x}))$ is a multiplicative inverse of $D_{F}(\vec{x})$ modulo the ideal $\langle F_1, \dotsc,  F_m \rangle$. The idea is then to multiply by $\Delta$ instead of dividing by $D$. More precisely, the following is an \I-proof that $G \in \langle F_1, \dotsc, F_m \rangle$:
\begin{equation} \label{eqn:RIPScert}
C_{\Delta}(\vec{x}, \vec{\f}) \defeq C'(\vec{x}, \vec{\f})\Delta(\vec{x}, \vec{\f}, D(\vec{x}, \vec{\f})) + G(\vec{x})E'(\vec{x}, \vec{\f}).
\end{equation}
Since $C'$ and $E'$ must individually be in $\langle \f_1, \dotsc, \f_m \rangle$, the entirety of $C_{\Delta}$ is as well. To see that we get $G(\vec{x})$ after plugging in the $F_i(\vec{x})$ for the $\f_i$, we compute:
\begin{eqnarray*}
C_{\Delta}(\vec{x}, \vec{F}(\vec{x})) & = & C'(\vec{x}, \vec{F}(\vec{x}))\Delta(\vec{x}, \vec{F}(\vec{x}), D(\vec{x}, \vec{F}(\vec{x}))) + G(\vec{x})E'(\vec{x}, \vec{F}(\vec{x})) \\
 & = & C'(\vec{x}, \vec{F}(\vec{x}))\left(\frac{1-E'(\vec{x}, \vec{F}(\vec{x}))}{D_{F}(\vec{x})} \right) + G(\vec{x})E'(\vec{x}, \vec{F}(\vec{x})) \\
 & = & G(\vec{x})\left(1-E'(\vec{x}, \vec{F}(\vec{x}))\right) + G(\vec{x})E'(\vec{x}, \vec{F}(\vec{x})) \\
 & = & G(\vec{x}).
\end{eqnarray*}

Finally, we give an upper bound on the size of a circuit for $C_{\Delta}$. The numerator and denominator of a rational function computed by a circuit of size $s$ can be computed individually by circuits of size $O(s)$. The basic idea, going back to Strassen \cite{strassenDivision}, is to replace each wire by a pair of wires explicitly encoding the numerator and denominator, to replace a multiplication gate by a pair of multiplication gates---since $(A/B) \times (C/D) = (A \times C)/(B \times D)$---and to replace an addition gate by the appropriate gadget encoding the expression $(A/B) + (C/D) = (AD + BC)/BD$. In particular, we may assume that a circuit computing $C'/D$ has the following form: it first computes $C'$ and $D$ separately, and then has a single division gate computing $C'/D$. 
Thus from a circuit for $C$ we can get circuits of essentially the same size for both $C'$ and $D$. Given a circuit for $E = d' \Delta + E'$, we get a circuit for $E'$ by setting $d'=0$. We can then get a circuit for $d'\Delta$ as $E - E'$. From a circuit for $d'\Delta$ we can get a circuit for $\Delta$ alone by first dividing $d'\Delta$ by $d'$, and then eliminating that division using Strassen \cite{strassenDivision}. Combining these, we then easily construct a circuit for the \I-certificate $C_{\Delta}$ of size $\poly(|C|, |E|)$.
\end{proof}

\begin{example}
Returning to the inversion principle, we find that the certificate from Example~\ref{ex:inversion} only divided by $\det(X)$, which we've already remarked does not vanish \emph{anywhere} that $XY - I$ vanishes. By the preceding proposition, there is thus an \I-certificate for the inversion principle of polynomial size, \emph{if} there is an \I-certificate for the unsatisfiability of $\det(X) = 0 \land XY-I=0$ of polynomial size. In this case we can guess at the multiplicative inverse of $\det(X)$ modulo $XY-I$, namely $\det(Y)$, since we know that $\det(X)\det(Y) = 1$ if $XY=I$. Hence, we can try to find a certificate for the unsatisfiability of $\det(X) = 0 \land XY-I=0$ of the form 
\[
\det(X) \det(Y) + (\text{something in the ideal of } \langle (XY-I)_{i,j \in [n]} \rangle) = 1.
\]
In other words, we want a refutation-style \I-proof of the implication $XY = I \Rightarrow \det(X)\det(Y)=1$, which is another one of the Hard Matrix Identities. Such a refutation is exactly what Hrubes and Tzameret provide \cite{hrubesTzameretDet}.
\end{example}

In fact, for this particular example we could have anticipated that a rational certificate was unnecessary, because the ideal generated by $XY-I$ is prime and hence radical. (Indeed, the ring $\F[X, Y]/\langle XY - I \rangle$ is the coordinate ring of the algebraic group $\text{GL}_n(\F)$, which is an irreducible variety.)

Unfortunately, the Rational Ideal Proof System is not complete, as the next example shows.

\begin{example}
Let $F_1(x) = x^2$ and $G(x) = x$. Then $G(x) \in \sqrt{\langle F_1(\vec{x}) \rangle}$, but any R\I certificate would show $G(x) D(x) = F_1(x) H(x)$ for some $D, H$. Plugging in, we get $x D(x) = x^2 H(x)$, and by unique factorization we must have that $D(x) = x D'(x)$ for some $D'$. But then $D$ vanishes identically on $V(F_1)$, contrary to the definition of R\I-certificate.
\end{example}

To get a more complete proof system, we could generalize the definition of R\I to allow dividing by any polynomial that does not vanish to appropriate \emph{multiplicity} on each irreducible component (see, \eg, \cite[Section~3.6]{eisenbud} for the definition of multiplicity). For example, this would allow dividing by $x$ to show that $x \in \sqrt{\langle x^2 \rangle}$, but would disallow dividing by $x^2$ or any higher power of $x$. However, the proof of soundness of this generalized system is more involved, and the results of the next section seem not to hold for such a proof system. As of this writing we do not know of any better characterization of when R\I certificates exist other than the definition itself.

\begin{definition}
A R\I certificate is \definedWord{Hilbert-like} if the denominator doesn't involve the placeholder variables $\f_i$ and the numerator is $\vec{\f}$-linear. In other words, a Hilbert-like R\I certificate has the form $\frac{1}{D(\vec{x})}\sum_{i} \f_i G_i(\vec{x})$.
\end{definition}

\begin{lemma} \label{lem:RIPSgen2Hilb}
If there is a R\I certificate that $G \in \sqrt{\langle F_1, \dotsc, F_m \rangle}$, then there is a Hilbert-like R\I certificate proving the same.
\end{lemma}

\begin{proof}
Let $C = C'(\vec{x}, \vec{\f})/D(\vec{x}, \vec{\f})$ be a R\I certificate. First, we replace the denominator by $D_{F}(\vec{x}) = D(\vec{x}, \vec{F}(\vec{x}))$.  Next, for each monomial appearing in $C'$, we replace all but one of the $\f_i$ in that monomial with the corresponding $F_i(\vec{x})$, reducing the monomial to one that is $\vec{\f}$-linear.
\end{proof}

As in the case of \I, we only know how to guarantee a size-efficient reduction under a sparsity condition. The following is the R\I-analogue of Proposition~\ref{prop:gen2Hilb}.

\begin{corollary} \label{cor:RIPSgen2Hilb}
If $C = C'/D$ is a R\I proof that $G \in \sqrt{\langle F_1, \dotsc, F_m \rangle}$, where the numerator $C'$ satisfies the same sparsity condition as in Proposition~\ref{prop:gen2Hilb}, then there is a Hilbert-like R\I proof that $G \in \sqrt{\langle F_1, \dotsc, F_m \rangle}$, of size $\poly(|C|)$.
\end{corollary}

\begin{proof}
We follow the proof of Lemma~\ref{lem:RIPSgen2Hilb}, making each step effective. As in the last paragraph of the proof of Proposition~\ref{prop:RIPS2IPS}, any circuit with divisions computing a rational function $C'/D$, where $C',D$ are relatively prime polynomials can be converted into a circuit without divisions computing the pair $(C', D)$. By at most doubling the size of the circuit, we can assume that the subcircuits computing $C'$ and $D$ are disjoint. Now replace each $\f_i$ input to the subcircuit computing $D$ with a small circuit computing $F_i(\vec{x})$. Next, we apply sparse multivariate interpolation to the numerator $C'$ exactly as in Proposition~\ref{prop:gen2Hilb}. The resulting circuit now computes a Hilbert-like R\I certificate. 
\end{proof}

\subsection{Towards lower bounds}
We begin by noting that, since the numerator and denominator can be computed separately (originally due to Strassen \cite{strassenDivision}, see the proof of Proposition~\ref{prop:RIPS2IPS} above for the idea), it suffices to prove a lower bound on, for each R\I-certificate, either the denominator or the numerator.

As in the case of Hilbert-like \I and general \I (recall Section~\ref{sec:syzygy}), the set of R\I certificates showing that $G \in \sqrt{\langle F_1, \dotsc, F_m \rangle}$ is a coset of a finitely generated ideal. 

\begin{lemma}
The set of R\I-certificates showing that $G \in \sqrt{\langle F_1, \dotsc, F_m \rangle}$ is a coset of a finitely generated ideal in $R$, where $R$ is the localization of $\F[\vec{x}, \vec{\f}]$ at $\bigcup_i P_i$, where the union is over the prime ideals minimal over $\langle F_1, \dotsc, F_m \rangle$.

Similarly, the set of Hilbert-like R\I certificates is a coset of a finitely generated submodule of $R'^{m}$, where $R' = R \cap \F[\vec{x}]$ is the localization of $\F[\vec{x}]$ at $\bigcup_i (P_i \cap \F[\vec{x}])$.
\end{lemma}

\begin{proof}
The proof is essentially the same as that of Lemma~\ref{lem:fgsyz}, but with one more ingredient. Namely, we need to know that the rings $R$ and $R'$ are Noetherian. This follows from the fact that polynomial rings over fields are Noetherian, together with the general fact that any localization of a Noetherian ring is again Noetherian.
\end{proof}

Exactly analogous to the the case of \I certificates, we define general and Hilbert-like R\I zero-certificates to be those for which, after plugging in the $F_i$ for $\f_i$, the resulting function is identically zero. In the case of Hilbert-like R\I, these are again syzygies of the $F_i$, but now syzygies with coefficients in the localization $R' = \F[\vec{x}]_{P_1 \cup \dotsb \cup P_k}$. 

However, somewhat surprisingly, we seem to be able to go further in the case of R\I than \I, as follows. In general, the ring $\F[\vec{x}, \vec{\f}]_{P_1 \cup \dotsb \cup P_k}$ is a Noetherian \emph{semi-local} ring, that is, in addition to being Noetherian, it has finitely many maximal ideals, namely $P_1, \dotsc, P_k$. Ideals in and modules over semi-local rings enjoy properties not shared by ideals and modules over arbitrary rings.

In the special case when there is just a single prime ideal $P_1$, the localization is a \emph{local} ring (just one maximal ideal). We note that this is the case in the setting of the Inversion Principle, as the ideal generated by the $n^2$ polynomials $XY-I$ is prime. Local rings are in some ways very close to fields---if $R$ is a local ring with unique maximal ideal $P$, then $R/P$ is a field---and modules over local rings are much closer to vector spaces than are modules over more general rings. This follows from the fact that $M/P$ is then in fact a vector space over the field $R/P$, together with Nakayama's Lemma (see, \eg, \cite[Corollary~4.8]{eisenbud} or \cite[Section~2.8]{reidCA}). Once nice feature is that, if $M$ is a module over a local ring, then every minimal generating set has the same size, which is the dimension of $M/P$ as an $R/P$-vector space. We also get that for every minimal generating set $b_1, \dotsc, b_k$ of $M$ (``$b$'' for ``basis'', even though the word basis is reserved for free modules), for each $m \in M$, any two representations $m = \sum_{i=1}^{k} r_i b_i$ with $r_i \in R$ differ by an element in $PM$. This near-uniqueness could be very helpful in proving lower bounds, as normal forms have proved useful in proving many circuit lower bounds.

\begin{open}
Does every R\I proof of the $n \times n$ Inversion Principle $XY = I \Rightarrow YX = I$ require computing a determinant? That is, is it the case that for every R\I certificate $C=C'/D$, some determinant of size $n^{\Omega(1)}$ reduces to one of $C, C', D$ by a $O(\log n)$-depth circuit reduction?
\end{open}

A positive answer to this question would imply that the Hard Matrix Identities do not have $O(\log n)$-depth R\I proofs unless the determinant can be computed by a polynomial-size algebraic formula. Since \I (and hence R\I) simulates Frege-style systems in a depth-preserving way (Theorem~\ref{thm:depth}), a positive answer would also imply that there are not ($\cc{NC}^1$-)Frege proofs of the Boolean Hard Matrix Identities unless the determinant has polynomial-size \emph{algebraic} formulas. Although answering this question may be difficult, the fact that we can even \emph{state} such a precise question on this matter should be contrasted with the preceding state of affairs regarding Frege proofs of the Boolean Hard Matrix Identities (which was essentially just a strong intuition that they should not exist unless the determinant is in $\cc{NC}^1$).
\section{Geometric \texorpdfstring{\I}{\Itext}-certificates} \label{app:geom}
We may consider $F_1(x_1, \dotsc, x_n), \dotsc, F_m(x_1,\dotsc,x_n)$ as a polynomial map $F = (F_1,\dotsc,F_m)\colon\F^{n} \to \F^{m}$. Then this system of polynomials has a common zero if and only if $0$ is the image of $F$. In fact, we show that for any Boolean system of equations, which are those that include $x_1^2 - x_1 = \dotsb = x_n^2 - x_n = 0$, or multiplicative Boolean equations---those that include $x_1^2 - 1 = \dotsb = x_n^2 - 1 = 0$---the system of polynomials has a common zero if and only if $0$ is in the \emph{closure} of the image of $F$.

The preceding is the geometric picture we pursue in this section; next we describe the corresponding algebra. The set of \I certificates is the intersection of the ideal $\langle \f_1, \dotsc, \f_m \rangle$ with the coset $1 + \langle \f_1 - F_1(\vec{x}), \dotsc, \f_m - F_m(\vec{x}) \rangle$. The map $a \mapsto 1 - a$ is a bijection between this coset intersection and the coset intersection $\left(1 + \langle \f_1, \dotsc, \f_m \rangle \right) \cap \langle \f_1 - F_1(\vec{x}), \dotsc, \f_m - F_m(\vec{x}) \rangle$. In particular, the system of equations $F_1 = \dotsb = F_m = 0$ is unsatisfiable if and only if the latter coset intersection is nonempty. 

We show below that if the latter coset intersection contains a polynomial involving only the $\f_i$'s---that is, its intersection with the subring $\F[\vec{\f}]$ (rather than the much larger ideal $\langle \vec{\f} \rangle \subseteq \F[\vec{x}, \vec{\f}]$) is nonempty---then $0$ is not even in the closure of the image of $F$. Hence we call such polynomials ``geometric certificates:''

\begin{definition}[The Geometric Ideal Proof System] \label{def:geompf}
A \definedWord{geometric \I certificate} that a system of $\F$-polynomial equations $F_1(\vec{x}) = \dotsb = F_m(\vec{x}) = 0$ is unsatisfiable over $\overline{\F}$ is a polynomial $C \in \F[\f_1, \dotsc, \f_m]$ such that
\begin{enumerate}
\item \label{condition:geom_nonzero} $C(0,0,\dotsc,0) = 1$, and

\item \label{condition:geom_ideal} $C(F_1(\vec{x}), \dotsc, F_{m}(\vec{x})) = 0$. In other words, $C$ is a polynomial relation amongst the $F_i$.
\end{enumerate}
A \definedWord{geometric \I proof} of the unsatisfiability of $F_1 = \dotsb = F_m = 0$, or a \definedWord{geometric \I refutation} of $F_1 = \dotsb = F_m = 0$, is an $\F$-algebraic circuit on inputs $\f_1, \dotsc,\f_m$ computing some geometric certificate of unsatisfiability.
\end{definition}

If $C$ is a geometric certificate, then $1-C$ is an \I certificate that involves only the $\f_i$'s, somewhat the ``opposite'' of a Hilbert-like certificate. Hence the smallest circuit size of any geometric certificate is at most the smallest circuit size of any algebraic certificate. We do not know, however, if these complexity measures are polynomially related:

\begin{open} \label{question:geometric}
For Boolean systems of equations, Geometric \I polynomially equivalent to \I? That is, is there always a geometric certificate whose circuit size is at most a polynomial in the circuit size of the smallest algebraic certificate?
\end{open}

Although the Nullstellensatz doesn't guarantee the existence of geometric certificates for arbitrary unsatisfiable systems of equations---and indeed, geometric certificates need not always exist---for \emph{Boolean} systems of equations (usual or multiplicative) geometric certificates always exist. In fact, this holds for any system of equations which contains at least one polynomial containing only the variable $x_i$, for each variable $x_i$:

\begin{proposition} \label{prop:geometric}
Let $\F$ be either a (topologically) dense subfield of $\C$ or any algebraically closed field. A Boolean system of equations over $\F$---or more generally any system of equations containing, for each variable $x_i$, at least one non-constant equation involving only $x_i$\footnote{We believe that the ``correct'' generalization here is to systems of equations $F_1 = \dotsb = F_m = 0$ such that
the corresponding map $F \colon \F^{n} \to \image(F)$ is \emph{flat} (see, \eg, \cite[Chapter~6]{eisenbud}) and has zero-dimensional fibers, that is, the inverse image of any point is a finite set. Systems satisfying the hypothesis of Proposition~\ref{prop:geometric} satisfy these hypotheses as well, but we have not checked carefully if the result extends in this generality.}
---has a common root if and only if it does not have a geometric certificate.
\end{proposition}

The condition of this proposition is almost surely more stringent than necessary, but the next example shows that at least some condition is necessary. 

\begin{example}
Let $F_1(x,y) = xy - 1$ and $F_2(x,y) = x^2 y$. There is no solution to $F_1 = F_2 = 0$, as $F_1 = 0$ implies that both $x$ and $y$ are nonzero, but if this is the case then $x^2 y = F_2(x,y)$ is also nonzero. Yet $0$ is in the closure of the image of the map $F = (F_1, F_2)\colon \F^{2} \to \F^2$. There are (at least) two ways to see this. First, we exhibit $0$ as an explicit limit of points in the image. Let $\chi_1(\varepsilon) = \varepsilon$ and $\chi_2(\varepsilon) = 1/\varepsilon$. Then $F_1(\chi_1(\varepsilon), \chi_2(\varepsilon)) = 0$ identically in $\varepsilon$, and $F_2(\chi_1(\varepsilon), \chi_2(\varepsilon)) = \varepsilon$. Thus, if we take the limit as $\varepsilon \to 0$, we find that $0$ is in the closure of the image of $F$.\footnote{If $\F$ is a dense subfield of $\C$, this limit may be taken in the usual sense of the Euclidean topology. For arbitrary algebraically closed fields $\F$, the same construction works, but must now be interpreted in the context of Lemma~\ref{lem:eps}.}

Alternatively, in this case we can determine the entire image exactly (usually a very daunting task): it is $\{(a,b) \in \F^2 : a \neq -1 \text{ and } b \neq 0\} \cup \{(-1,0)\}$. This can be determined by solving the equations by the elementary method of substitution, and careful but not complicated case analysis. It is then clear (geometrically in the case of subfields of $\C$, and by a dimension argument over an arbitrary algebraically closed field) that the closure of the image is the entirety of $\F^2$, and in particular contains $0$.
\end{example}

The next example rules out another natural attempt at generalizing Proposition~\ref{prop:geometric}, and also shows that the existence of geometric certificates for a given set of equations can depend on the equations themselves, and not just on the ideal they generate. 

\begin{example}
Let $F_1(x,y) = xy-1$ and $F_2(x,y)=x^2y$ as before, and now also add $F_3(x,y) = x^2(1-y)$. We already saw that $F_1 = F_2 = 0$ is unsatisfiable, so $F_1 = F_2 = F_3 = 0$ is unsatisfiable as well. However, $F_1 = F_3 = 0$ has one, and only one, solution, namely $x=y=1$. Let $F = (F_1,F_2,F_3)\colon \F^2 \to \F^3$. To see that $\vec{0}$ is in the closure of the image of $F$, we again consider $\lim_{\varepsilon \to 0} F(\varepsilon, 1/\varepsilon)$. As before $F_1(\varepsilon,1/\varepsilon)=0$ and $F_2(\varepsilon,1/\varepsilon) = \varepsilon$, whose limit is zero as $\varepsilon \to 0$. Similarly, we get $F_3(\varepsilon, 1/\varepsilon) = \varepsilon^2 (1 - 1/\varepsilon) = \varepsilon (\varepsilon - 1)$, which again goes to $0$ as $\varepsilon \to 0$.

Note that if we replace equations $F_1$ and $F_3$ by another set of equations with the same set of solutions (in this case, a singleton set), but satisfying the conditions of Proposition~\ref{prop:geometric}, such as $F_1' = (x-1)^k$ and $F_3' = (y-1)^\ell$ for some $k,\ell > 0$, then $\vec{0}$ is no longer in the closure of the image. For if $(F_1',F_2,F_3')$ approaches $(0,0,0)$, then $x$ and $y$ must both approach $1$, but then $F_2 = x^2 y$ also approaches $1$. Furthermore, by the Nullstellensatz, for some $k,\ell > 0$, the polynomials $(x-1)^k$ and $(y-1)^\ell$ both in the ideal $\langle F_1, F_3 \rangle$. Thus, although the solvability of a system of equations is determined entirely by (the radical of) the ideal they generate, the geometry of the corresponding map---and even the existence of geometric certificates---can change depending on which elements of the ideal are used in defining the map.
\end{example}

The following lemma is the key to Proposition~\ref{prop:geometric}. 

\begin{lemma} \label{lem:geometric}
Let $\F$ be (1) a dense subfield of $\C$ (in the Euclidean topology), or (2) any algebraically closed field. Let $F_1(\vec{x}), \dotsc, F_{m}(\vec{x})$ be a system of equations over $\F$, and let $F=(F_1,\dotsc,F_m)\colon \F^{n} \to \F^{m}$ be the associated polynomial map, as above. If, for $i=1,\dotsc,n$, $F_i(\vec{x})$ is a nonzero function of $x_i$ alone, then the set of equations $F_1 = \dotsb = F_m = 0$ has a solution if and only if $0$ is in the closure $\overline{\image(F)}$. 
\end{lemma}

\begin{proof}
If the system $F$ has a common solutions, then $0$ is in the image of $F$ and hence in its closure. 

Conversely, suppose $0$ is in the closure of the image of $F$. We first prove case (1) (the characteristic zero case) as it is somewhat simpler and gives the main idea, and then we prove case (2), the case of an arbitrary algebraically closed field.

(1) Dense subfields of $\C$. First, we note that the closure of the image of $F$ in the Zariski topology agrees with its closure in the standard Euclidean topology on $\F^{n}$, induced by the Euclidean topology on $\C^{n}$. For $\F = \C$, see, \eg, \cite[Theorem~2.33]{mumford}. For other dense $\F \subsetneq \C$, suppose $\vec{y}$ is in the $\F$-Zariski-closure of $F(\F^{n})$, that is, every $\F$-polynomial that vanishes everywhere on $F(\F^{n})$ also vanishes at $\vec{y}$. By the aforementioned result for $\C$, there is a sequence of points $\vec{x}_1, \vec{x}_2, \dotsc \in \C^{n}$ such that $\vec{y} = \lim_{k \to \infty} F(\vec{x}_k)$. As $\F$ is dense in $\C$ in the Euclidean topology, there is similarly a sequence of points $\vec{x}'_1, \vec{x}'_2, \dotsc \in \F^{n}$ such that $|\vec{x}_k - \vec{x}'_k| \leq 1/k$ for all $k$. Hence $\lim_{k \to \infty} \vec{x}_k = \lim_{k \to \infty} \vec{x}'_k$. Each $F(\vec{x}'_k) \in \F^{m}$, so we get a sequence of points $F(\vec{x}'_1), F(\vec{x}'_2), \dotsc \in \F^{m}$ whose limit is $\vec{y}$.

In particular, $0$ is in the (Zariski-)closure of the image of $F$ if and only if there is a sequence of points $v^{(1)}, v^{(2)}, v^{(3)}, \dotsc \in \image(F)$ such that $\lim_{k \to \infty} v^{(k)} = 0$. As each $v^{(k)}$ is in the image of $F$, there is some point $\nu^{(k)} \in \F^{n}$ such that $v^{(k)} = F(\nu^{(k)})$. As the $v^{(k)}$ approach the origin, each $F_i(\nu^{(k)})$ approaches $0$, since it is the $i$-th coordinate of $v^{(k)}$: $v^{(k)}_i = F_i(\nu^{(k)})$. 

In particular, since $F_1(\vec{x})$ depends only on $x_1$ and is nonzero (by assumption), the first coordinates $\nu^{(k)}_{1}$ must accumulate around the finitely many zeroes of $F_1(x_1)$. Similarly for each coordinate $i=1,\dotsc,n$ of $\nu^{(k)}$. 
Thus there is an infinite subsequence of the $\nu^{(k)}$ that approaches one single solution $\vec{z}$ to $F=0$. By choosing such a subsequence and re-indexing, we may assume that $\lim_{k \to \infty} \nu^{(k)} = \vec{z}$.

Finally, by assumption and continuity, we have
\[
0 = \lim_{k \to \infty} v^{(k)} = \lim_{k \to \infty} F(\nu^{(k)}) = F(\lim_{k \to \infty} \nu^{(k)}) = F(\vec{z}),
\]
so $\vec{z}$ is a common root of the original system $F_1 = \dotsb = F_m = 0$. Hence, if $0$ is in the closure of the image of $F$, then $0$ is in the image.

(2) $\F$ any algebraically closed field. Here we cannot use an argument based on the Euclidean topology, but there is a suitable, purely algebraic analogue, encapsulated in the following lemma:

\begin{lemma}[{See, \eg, \cite[Lemma~20.28]{BCS}}] \label{lem:eps}
If $p$ is a point in the closure of the image of a polynomial map $F\colon \F^{n} \to \F^{m}$, then there are formal Laurent series\footnote{A formal Laurent series is a formal sum of the form $\sum_{k=-k_0}^{\infty} a_k \varepsilon^{k}$. By ``formal'' we mean that we are paying no attention to issues of convergence (which need not even make sense over various fields), but are just using the degree of $\varepsilon$ as an indexing scheme.} $\chi_1(\varepsilon), \dotsc, \chi_n(\varepsilon)$ in a new variable $\varepsilon$ such that $F_i(\chi_1(\varepsilon), \dotsc, \chi_n(\varepsilon))$ is in fact a \emph{power series}---that is, involves no negative powers of $\varepsilon$---for each $i=1,\dotsc,m$, and such that evaluating the power series $(F_1(\vec{\chi}(\varepsilon)), \dotsc, F_m(\vec{\chi}(\varepsilon))$ at $\varepsilon=0$ yields the point $p$.
\end{lemma}

Note that the evaluation at $\varepsilon=0$ must occur \emph{after} applying $F_i$, since each individual $\chi_i$ may involve negative powers of $\varepsilon$.

As $F_1$ involves only $x_1$, in order for $F_1(\vec{\chi}(\varepsilon)) = F_1(\chi_1(\varepsilon))$ to be a power series in $\varepsilon$, it must be the case that $\chi_1(\varepsilon)$ itself is a power series (contains no negative powers of $\varepsilon$). For if the highest degree term of $F_1$ is some constant times $x_1^{d}$, and the lowest degree term of $\chi_1(\varepsilon)$ is of degree $-D$, then $F_1(\chi_1(\varepsilon))$ contains the monomial $\varepsilon^{-dD}$ with nonzero coefficient. A similar argument applies to $\chi_i$ for $i=1,\dotsc, n$. Thus each $\chi_i$ is in fact a power series, involving no negative terms of $\varepsilon$, and hence can be evaluated at $0$. Since evaluating at $\varepsilon=0$ now makes sense even before applying the $F_i$, and is a ring homomorphism (we might say, ``is continuous with respect to the ring operations''), we get that
\[
0 = F_i(\vec{\chi}(\varepsilon))|_{\varepsilon=0} = F_i(\vec{\chi}(\varepsilon)|_{\varepsilon=0}) = F_i(\vec{\chi}(0))
\]
for each $i=1,\dotsc,m$, and hence $\vec{\chi}(0)$ is a solution to $F_1(\vec{x}) = \dotsb = F_m(\vec{x}) = 0$.
\end{proof}

\begin{proof}[Proof of Proposition~\ref{prop:geometric}]
Let $F_1, \dotsc, F_m$ be an unsatisfiable system of equations over $\F$ satisfying the conditions of Lemma~\ref{lem:geometric}, and let $F = (F_1, \dotsc, F_m) \colon \F^{n} \to \F^{m}$ be the corresponding polynomial map. 

First, suppose that $F_1 = \dotsb = F_m = 0$ has a solution. Then $0 \in \image(F)$, so any $C(\f_1, \dotsc, \f_m)$ that vanishes everywhere on $\image(F)$, as required by condition (\ref{condition:geom_ideal}) of Definition~\ref{def:geompf}, must vanish at $\vec{0}$. In other words, $C(0,\dotsc,0) = 0$, contradicting condition (\ref{condition:geom_nonzero}). So there are no geometric certificates.

Conversely, suppose $C(\f_1, \dotsc, \f_m)$ is a geometric certificate. Then $C$ vanishes at every point of the image $\image(F)$ and hence at every point of its closure $\overline{\image(F)}$, by (Zariski-)continuity. By condition (\ref{condition:geom_nonzero}) of Definition~\ref{def:geompf}, $C(0,\dotsc,0) = 1$. Since $C$ does not vanish at the origin, $\vec{0} \notin \overline{\image(F)}$. Then by Lemma~\ref{lem:geometric}, $\vec{0}$ is not in the image of $F$ and hence $F_1 = \dotsb = F_m = 0$ has no solution.
\end{proof}

Finally, as with \I certificates and Hilbert-like \I certificates (see Section~\ref{sec:syzygy}), a \definedWord{geometric zero-certificate} for a system of equations $F_1(\vec{x}), \dotsc, F_m(\vec{x})$ is a polynomial $C(\f_1, \dotsc, \f_m) \in \langle \f_1, \dotsc, \f_m \rangle$---that is, such that $C(0,\dotsc,0) = 0$---and such that $C(F_1(\vec{x}), \dotsc, F_{m}(\vec{x})) = 0$ identically as a polynomial in $\vec{x}$. The same arguments as in the case of algebraic certificates show that any two geometric certificates differ by a geometric zero-certificate, and that the geometric certificates are closed under multiplication. Furthermore, the set of geometric zero-certificates is the intersection of the ideal of (algebraic) zero-certificates $\langle \f_1, \dotsc, \f_m \rangle \cap \langle \f_1 - F_1(\vec{x}), \dotsc, \f_m - F_m(\vec{x}) \rangle$ with the subring $\F[\vec{\f}] \subset \F[\vec{x}, \vec{\f}]$. As such, it is an ideal of $\F[\vec{\f}]$ and so is finitely generated. Thus, as in the case of \I certificates, the set of all geometric certificates can be specified by giving a single geometric certificate and a finite generating set for the ideal of geometric zero-certificates, suggesting an approach to lower bounds on the Geometric Ideal Proof System.

We note that geometric zero-certificates are also called syzygies amongst the $F_i$---sometimes ``geometric syzygies'' or ``polynomial syzygies'' to distinguish them from the ``module-type syzygies'' we discussed above in relation to Hilbert-like \I. As in all the other cases we've discussed, a generating set of the geometric syzygies can be computed using \Grobner bases, this time using elimination theory: compute a \Grobner basis for the ideal $\langle \f_1 - F_1(\vec{x}), \dotsc, \f_m - F_m(\vec{x}) \rangle$ using an order that eliminates the $x$-variables, and then take the subset of the \Grobner basis that consists of polynomials only involving the $\f$-variables. The ideal of geometric syzygies is exactly the ideal of the closure of the image of the map $F$, and for this reason this kind of syzygy is also well-studied. This suggests that geometric properties of the image of the map $F$ (or its closure) may be useful in understanding the complexity of individual instances of $\cc{coNP}$-complete problems.

\end{document}